\documentclass[11pt, dvipsnames]{article}
\usepackage{jmlrlocal}
\usepackage{header}
\usepackage[square, numbers]{natbib}
\let\cite\citep
\usepackage{hyperref}       
\hypersetup{
    colorlinks = true,
    citecolor=Green,
    urlcolor=Red,
    linkcolor=Blue,
}

\usepackage[margin=1.0in]{geometry}

\DeclareMathOperator{\Tr}{Tr}

\newcommand{\reac}{\mc{R}}
\newcommand{\bmat}[1]{\begin{bmatrix}#1\end{bmatrix}}

\newcommand{\la}{\langle}
\newcommand{\ra}{\rangle}

\newcommand{\conj}{r}

\newcommand{\NE}{\text{\tt NE}}

\newcommand{\learnraten}[1]{\gamma_{#1}}
\newcommand{\omegas}{\omega_{\mc{S}}}
\newcommand{\omegasn}[1]{\omega_{\mc{S},#1}}
\newcommand{\radius}{q}
\newcommand{\rank}{\mathrm{rank}}
\newcommand{\Js}{J_{\mc{S}}}
\newcommand{\spec}{\mathrm{spec}}
\newcommand{\trans}{\top}

\newcommand{\diag}{\mathrm{diag}}
\newcommand{\poneevecs}{W_1}
\newcommand{\ptwoevecs}{W_2}
\title{Convergence of Learning Dynamics in Stackelberg Games}
\author{\name Tanner Fiez \email fiezt@uw.edu\\
\addr Department of Electrical and Computer Engineering\\
    University of Washington\\
\AND
\name Benjamin Chasnov \email bchasnov@uw.edu\\
\addr Department of Electrical and Computer Engineering\\
    University of Washington\\
    \AND
    \name Lillian J. Ratliff \email ratliffl@uw.edu \\
    \addr Department of Electrical and Computer Engineering\\
       University of Washington}
\ShortHeadings{}{Fiez, Chasnov, Ratliff}
\begin{document}

\maketitle

\begin{abstract}
This paper investigates the convergence of learning dynamics in Stackelberg games. In the class of games we consider, there is a hierarchical game being played between a leader and a follower with continuous action spaces. We establish a number of connections between the Nash and Stackelberg equilibrium concepts and characterize conditions under which attracting critical points of simultaneous gradient descent are Stackelberg equilibria in zero-sum games.
Moreover, we show that the only stable critical points of the Stackelberg gradient dynamics are Stackelberg equilibria in zero-sum games. Using this insight, we develop a gradient-based update for the leader while the follower employs a best response strategy for which each stable critical point is guaranteed to be a Stackelberg equilibrium in zero-sum games. As a result, the learning rule provably converges to a Stackelberg equilibria given an initialization in the region of attraction of a stable critical point. We then consider a follower employing a gradient-play update rule instead of a best response strategy and propose a two-timescale algorithm with similar asymptotic convergence guarantees. For this algorithm, we also provide finite-time high probability bounds for local convergence to a neighborhood of a stable Stackelberg equilibrium in general-sum games. Finally, we present extensive numerical results that validate our theory, provide insights into the optimization landscape of generative adversarial networks, and demonstrate that the learning dynamics we propose can effectively train generative adversarial networks.
\end{abstract}

\section{Introduction}
\label{sec:introduction}
Tools from game theory now play a prominent role in machine learning. The emerging coupling between the fields can be credited to the formulation of learning problems as interactions between competing algorithms and the desire to characterize the limiting behaviors of such strategic interactions. Indeed, game theory provides a systematic framework to model the strategic interactions found in modern machine learning problems. 

A significant portion of the game theory literature concerns games of simultaneous play and equilibrium analysis. In simultaneous play games, each player reveals the strategy they have selected concurrently. The solution concept often adopted in non-cooperative simultaneous play games is the Nash equilibrium. 
In a Nash equilibrium, the strategy of each player is a best response to the joint strategy of the competitors so that no player can benefit from unilaterally deviating from this strategy.  

The study of equilibrium gives rise to the question of when and why the observed play in a game can be expected to correspond to an equilibrium. A common explanation is that an equilibrium emerges as the long run outcome of a process in which players repeatedly play a game and compete for optimality over time~\cite{fudenberg1998theory}. Consequently, an important topic in the study of learning in games is the convergence behavior of learning algorithms reflecting the underlying game dynamics. Adopting this viewpoint and analyzing so-called `natural' dynamics~\cite{benaim1999mixed} often provides deep insights into the structure of a game. Moreover, a firm understanding of the structure of a game can inform how to design learning algorithms strictly for computing equilibria. Seeking equilibria via computationally efficient learning algorithms is an equally important perspective on equilibrium analysis~\cite{fudenberg1998theory}.

The classic objectives of learning in games are now being widely embraced in the machine learning community. While not encompassing, the prevailing research areas epitomizing this phenomenon are adversarial training and multi-agent learning. A considerable amount of this work has focused on generative adversarial networks (GANs)~\cite{goodfellow2014generative}. 
Finding Nash equilibria in GANs is challenging owing to the complex optimization landscape that arises when each player in the game is parameterized by a neural network. Consequently, significant effort has been spent lately on developing principled learning dynamics for this application~\cite{metz:2017aa, mescheder2017numerics, heusel2017gans,
mertikopoulos2018optimistic, balduzzi18a, mazumdar:2019aa}. In general, this line of work has analyzed learning dynamics designed to mitigate rotational dynamics and converge faster to stable fixed points or to avoid spurious stable points of the dynamics and reach equilibria almost surely. In our work, we draw connections to this literature and believe that the problem we study gives an under-explored perspective that may provide valuable insights moving forward.

Characterizing the outcomes of competitive interactions and seeking equilibria
in multi-agent learning gained prominence much earlier than adversarial
training. However, following initial works on this
topic~\cite{littman1994markov, greenwald2003correlated, hu2003nash}, scrutiny
was given to the solution concepts being considered and the field
cooled~\cite{shoham2007if}. Owing to the arising applications with interacting
agents, problems of this form are being studied extensively again. There has
been a shift toward analyzing gradient-based learning rules, in part due to
their scalability and success in single-agent reinforcement learning, and rigorous convergence analysis~\cite{zhang:2010aa, balduzzi18a, letcher:2018aa, foerster:2018aa, mazumdar2018convergence}.

The progress analyzing learning dynamics and seeking equilibria in games is promising, but the work has been narrowly focused on simultaneous play games and the Nash equilibrium solution concept. There are many problems exhibiting a hierarchical order of play between agents in a diverse set of fields such as human-robot collaboration and interacting autonomous systems in artificial intelligence~\cite{nikolaidis2017game, liu2016goal, fisac2018hierarchical, sadigh2016planning}, mechanism design and control~\cite{dimitrakakis2017multi, ratliff2018adaptive, ratliff2018perspective}, and organizational structures in economics~\cite{bresnahan1981duopoly, anderson1992stackelberg}. In game theory, this type of game is known as a Stackelberg game and the solution concept studied is called a Stackelberg equilibrium. 

In the simplest formulation of a Stackelberg game, there is a leader and a
follower that interact in a hierarchical structure. The sequential order of play
is such that the leader is endowed with the power to select an action with the
knowledge that the follower will then play a best-response. 
Specifically, the leader uses this knowledge to its advantage when selecting a strategy. 

In this paper, we study the convergence of learning dynamics in Stackelberg
games. Our motivation stems from the emergence of problems in which there is a
distinct order of play between interacting learning agents and the lack of
existing theoretical convergence guarantees in this domain. The dynamics analyzed in this work reflect the underlying game structure and characterize the expected outcomes of hierarchical game play. The rigorous study
of the learning dynamics in Stackelberg games we provide also has implications
for simultaneous play games relevant to adversarial training.

\paragraph{Contributions.}
We formulate and study a novel set of gradient-based learning rules in
continuous, general-sum games that emulate the natural structure of a
Stackelberg game. Building on work characterizing a local Nash equilibrium in
continuous games~\cite{ratliff:2016aa}, we define the \emph{differential Stackelberg equilibrium} solution concept (Definition~\ref{def:stackelberg}), which is a local notion of a Stackelberg equilibrium amenable to computation. An analogous local minimax equilibrium concept was developed concurrently with this work, but strictly for zero-sum games~\cite{jin2019minmax}. Importantly, the equilibrium notion we present generalizes the local minimax equilibrium concept to general-sum games. In our work, we draw several connections between Nash and
Stackelberg equilibria for the class of zero-sum games, which can be summarized
as follows:
\begin{itemize}[leftmargin=25pt, itemsep=-2pt, topsep=2pt]
\item We show in Proposition~\ref{prop:DNEareDSE} that stable Nash equilibria are differential Stackelberg equilibria in zero-sum games. Concurrent with our work,~\citet{jin2019minmax} equivalently show that local Nash equilibria are local minimax equilibria. This result indicates learning dynamics seeking Nash equilibria are simultaneously seeking Stackelberg equilibria.  
\item We reveal that there exist stable attractors of simultaneous gradient play that are Stackelberg equilibria and not Nash equilibria. Moreover, in Propositions~\ref{prop:nec} and~\ref{prop:suf} we give necessary and sufficient conditions under which the simultaneous gradient play dynamics can avoid Nash equilibria and converge to Stackelberg equilibria. To demonstrate the relevancy to deep learning applications, Propositions~\ref{prop:gans} and~\ref{prop:gans2} specialize the general necessary and sufficient conditions from~Propositions~\ref{prop:nec} and~\ref{prop:suf}  to GANs satisfying the realizable assumption~\cite{nagarajan2017gradient}, which presumes the generator is able to create the underlying data distribution. This set of results has implications for the optimization landscape in GANs as we explore in our numerical experiments.
\end{itemize}
Our primary contributions concern the convergence behavior of the gradient-based
learning rules we formulate that mirror the Stackelberg game structure. These
contributions can be summarized as follows:
\begin{itemize}[leftmargin=25pt, itemsep=-2pt, topsep=2pt]
\item We demonstrate in Proposition~\ref{prop:allstack} that the only stable critical points of the Stackelberg gradient dynamics are Stackelberg equilibria in zero-sum games. This is in contrast to the simultaneous gradient play dynamics, which can be attracted to non-Nash critical points in zero-sum games. This insight allows us to define a gradient-based learning rule for the leader while the follower plays a best response for which each attracting critical point is a Stackelberg equilibria in zero-sum games. As a result, the learning rule provably converges to an equilibria given an initialization in the region of attraction of a stable critical point. A formal exposition of this set of dynamics and results is provided in Section~\ref{sec:leader_solutions}.
\item Leveraging the Stackelberg game structure, for general-sum games, we formulate a gradient-based learning rule in which the leader and follower have an unbiased estimator of their gradient so that updates are stochastic. 
\item In Section~\ref{sec:tt_results}, we consider a formulation in which the follower uses a gradient-play update rule instead of an exact best response strategy and propose a two-timescale algorithm to learn Stackelberg equilibria. We show almost sure asymptotic convergence to Stackelberg equilibria in zero-sum games and to stable attractors in general-sum games; a finite-time high probability bound for local convergence to a neighborhood of a stable Stackelberg equilibrium in general-sum games is also given. 
\item We present this paper with a single leader and a single follower, but this is only for ease of presentation. The extension to $N$ followers that play in a staggered hierarchical structure or simultaneously is in Appendix~\ref{app:extensions}; equivalent results hold with some additional assumptions. 
\end{itemize}
Finally, we present several numerical experiments in Section~\ref{sec:examples}, which we now detail: 
\begin{itemize}[leftmargin=25pt, itemsep=-2pt, topsep=2pt]
\item We present a location game on a torus and a Stackelberg duopoly game. 
The examples are general-sum games with equilibrium that can be solved for directly, allowing us to numerically validate our theory. The games demonstrate the advantage the leader gains from the hierarchical order of play compared to the simultaneous play versions of the games. 
\item We evaluate the Stackelberg learning dynamics as a GAN training algorithm. In doing so, we find that the leader update removes rotational dynamics and prevents the type of cycling behavior that plagues simultaneous gradient play. Moreover, we discover that the simultaneous gradient dynamics can empirically converge to non-Nash attractors that are Stackelberg equilibria in GANs. The generator and the discriminator exhibit desirable performance at such points, indicating that Stackelberg equilibria can be as desirable as Nash equilibria. Lastly, the Stackelberg learning dynamics often converge to non-Nash attractors and reach a satisfying solution quickly using learning rates that can cause the simultaneous gradient descent dynamics to cycle. 
\end{itemize}
\paragraph{Related Work.} 
The perspective we explore on analyzing games in which there is an order of play or hierarchical decision making structure has been generally ignored in the modern learning literature. However, this viewpoint has long been researched in the control literature on games~\cite{basar1979closed, papavassilopoulos1979nonclassical, papavassilopoulos1980sufficient, jungers2011min, basar:1998aa}. Similarly, work on bilevel optimization~\cite{daskin:1967aa, daskin:1967ab, zaslavski2012necessary} adopts this perspective. 

The select few recent works in the machine learning literature on learning in games considering a hierarchical decision-making structure exclusively focus on
zero-sum games~\cite{metz:2017aa, jin2019minmax, lin2019gradient, nouiehed2019solving, daskalakis:2018aa}, unlike our work, which extends to general-sum games. A noteworthy paper in the line of work in the zero-sum setting adopting a min-max perspective was the introduction of unrolled GANs~\cite{metz:2017aa}. 
The authors consider a timescale separation between
the generator and discriminator, giving the generator the advantage as the
slower player. This work used
the Schur complement structure presented in~\citet{daskin:1967aa, daskin:1967ab} to define a minimax solution of a zero-sum game
abstraction of an adversarial training objective.
Essentially the discriminator is allowed to perform a finite
roll-out in an inner loop of the algorithm with multiple updates; this process
is referred to as `unrolling'. It is (informally) suggested that, using the
results of~\citet{daskin:1967aa, daskin:1967ab}, as the roll-out
horizon approaches infinity, the discriminator approaches a critical point of the cost function along the discriminators axis given a fixed generator parameter configuration.

The unrolling procedure has the same effect as a deterministic timescale separation between players. Formal convergence guarantees to minimax equilibria in zero-sum games characterizing the limiting behavior of simultaneous individual gradient descent with timescale separation were recently obtained~\cite{jin2019minmax, lin2019gradient}.
While related, simultaneous individual gradient play with
time-scale separation is a distinct set of dynamics that departs from the
dynamics we propose that reflect the Stackelberg
game structure.  

It is also worth pointing out that the multi-agent learning
papers of~\citet{foerster:2018aa} and~\citet{letcher:2018aa} do in some sense
seek to give a player an advantage, but nevertheless focus on the Nash
equilibrium concept in any analysis that is provided. 

\paragraph{Organization.} In Section~\ref{sec:prelims}, we formalize the problem we study and provide background material on Stackelberg games. We then draw connections between learning in Stackelberg games and existing work in zero-sum and general sum-games relevant to GANs and multi-agent learning, respectively. In Section~\ref{sec:results}, we give a rigorous convergence analysis of learning in Stackelberg games. Numerical examples are provided in Section~\ref{sec:examples} and we conclude in Section~\ref{sec:discussion}.

\section{Preliminaries}
\label{sec:prelims}
We leverage the rich theory of continuous games and dynamical systems in order
to analyze algorithms implemented by agents interacting in a hierarchical game.
In particular, each agent has an objective they want to selfishly optimize that
depends on not only their own actions but also on the actions of their competitor. 
However, there is an order of play in the sense that one player is the leader
and the other player is the follower\footnote{While we present the work for a
    single leader and a single follower, the theory extends to the
    multi-follower case (we discuss this in Appendix~\ref{app:extensions}) and
to the case where the single leader abstracts multiple cooperating agents.}. The leader optimizes its objective
with the knowledge that the follower will respond by selecting a best response. 
 We refer to algorithms for learning in this setting as \emph{hierarchical learning} algorithms.
We specifically consider a class of learning algorithms in which the agents act
myopically with respect to their given objective and role in the underlying hierarchical
game by following the gradient of their objective with respect to their choice
variable.

To substantiate this abstraction, consider a game between two agents where one agent is deemed the \emph{leader}
and the other the \emph{follower}. 
The leader has cost $f_1:X\rar \mb{R}$ and the follower has cost
$f_2:X\rar \mb{R}$, where $X=X_1\times X_2$ with the action
space of the leader being $X_1$ and the action space of the follower being
$X_2$. The designation of `leader' and
`follower' indicates the order of play between the two agents, meaning the leader
plays first and the follower second. The leader and the follower need not be cooperative. Such a game is known as a \emph{Stackelberg game}.

\subsection{Stackelberg Games }
Let us adopt the typical game theoretic notation in which the player index set
is $\mc{I}$ and $x_{-i}=(x_j)_{j\in \mc{I}/\{i\}}$ denotes the joint action profile of all
agents excluding agent $i$. In the Stackelberg case, $\mc{I}=\{1,2\}$ where
player $i=1$ is the leader and player $i=2$ is the follower. We assume throughout that each $f_i$ is
sufficiently smooth, meaning $f_i\in C^q(X, \mb{R})$ for some $q\geq 2$ and for each $i\in \mc{I}$.

The leader aims to solve the optimization problem given by
\[\min_{x_1\in X_1}\Big\{f_1(x_1,x_2)\big|\ x_2\in \arg\min_{y\in
X_2}f_2(x_1,y)\Big\}\]
and the follower aims to solve the optimization problem
$\min_{x_2\in X_2}f_2(x_1,x_2)$.
As noted above, the learning algorithms we study are such that the agents follow
myopic update rules which take steps in the direction of steepest descent with
respect to the above two optimizations problems, the former for the leader and
the latter for the follower. 

Before formalizing these updates, let us first discuss the equilibrium concept
studied for simultaneous play games and contrast it with that which is studied in the hierarchical play counterpart. 
The typical equilibrium notion in continuous games is the pure strategy Nash equilibrium in simultaneous play games and the Stackelberg equilibrium in hierarchical play games. Each notion of equilibria can be characterized as the intersection points of the reaction curves of the players~\cite{basar:1998aa}. 
\begin{definition}[Nash Equilibrium]
    The joint strategy $x^\ast\in X$ is a Nash equilibrium if for each $i\in
    \mc{I}$, \[f_i(x^\ast)\leq
    f_i(x_i,x_{-i}^\ast),  \ \ \forall\ x_i\in X_i.\] The strategy is a local Nash
    equilibrium on $W\subset X$ if for each $i\in
    \mc{I}$, 
\setlength{\belowdisplayskip}{-6pt}\setlength{\belowdisplayshortskip}{-6pt}
    \[f_i(x^\ast)\leq
    f_i(x_i,x_{-i}^\ast),  \ \ \forall\ x_i\in W_i\subset X_i.\]
\label{def:nash_intersection}
\end{definition}
\begin{definition}[Stackelberg Equilibrium]
    In a two-player game with player 1 as the leader, a strategy $x_1^\ast\in
    X_1$ is called a Stackelberg equilibrium strategy for the leader if 
      \[\sup_{x_2\in \reac(x_1^\ast)}  f_1(x_1^\ast, x_2)\leq \sup_{x_2\in
      \reac(x_1)}f_1(x_1,x_2), \ \ \forall x_1\in X_1,\]
    where $\reac(x_1)=\{y\in X_2|\ f_2(x_1,y)\leq f_2(x_1,x_{2}),  \forall
    x_2\in X_2\}$ is the rational reaction set of $x_2$.
\end{definition}
This definition naturally extends to the $n$-follower setting when
$\reac(x_1)$ is replaced with the set of Nash equilibria $\NE(x_1)$, 
given that player 1 is playing $x_1$ so that the follower's reaction set is a Nash equilibrium.

We denote $D_if_i$ as the
derivative of $f_i$ with respect to $x_i$, $D_{ij}f_i$ as the partial derivative of $D_if_i$ with respect to $x_j$, and $D(\cdot)$ as the total
derivative\footnote{For example, given a function $f(x,r(x))$, $Df=D_1f+D_2f\partial r/\partial x$.}.
Denote by $\omega(x)=(D_1f_1(x),D_2f_2(x))$ the vector of individual gradients for simultaneous play 
  and $\omega_{\mc{S}}(x)=(Df_1(x),D_2f_2(x))$ as the equivalent for hierarchical play where $Df_1$ is the total derivative of $f_1$ with respect
  to $x_1$ and $x_2$ is implicitly a function of $x_2$, which captures the fact
  that the leader operates under the assumption that the follower will play a
  best response to its choice of $x_1$.

  It is possible to characterize a local Nash equilibrium using sufficient
  conditions for Definition~\ref{def:nash_intersection}.
  \begin{definition}[Differential Nash Equilibrium~\citep{ratliff:2016aa}]
    The joint strategy $x^\ast\in X$ is a differential Nash equilibrium if
    $\omega(x^\ast)=0$ and $D_i^2f_i(x^\ast)>0$ for each $i\in \mc{I}$.
        \label{def:nash}
\end{definition}

Analogous sufficient conditions can be stated to characterize a \emph{local} Stackelberg equilibrium strategy for
the leader using first and second order conditions on the leader's optimization
problem. Indeed,
if $Df_1(x_1^\ast,\conj(x_1^\ast))=0$
and $D^2f_1(x_1^\ast,\conj(x_1^\ast))$ is positive definite, then $x_1^\ast$
is a local Stackelberg
equilibrium strategy for the leader. We use these sufficient conditions to
define the following refinement of the Stackelberg equilibrium concept.
\begin{definition}[Differential Stackelberg Equilibrium]
The pair $(x_1^\ast,x_2^\ast)\in X$
with $x_2^\ast=\conj(x_1^\ast)$, where $\conj$ is implicitly defined by
$D_2f_2(x_1^\ast,x_2^\ast)=0$, is a differential Stackelberg equilibrium for
the game $(f_1,f_2)$ with player 1 as the leader if
$Df_1(x_1^\ast,\conj(x_1^\ast))=0$, and
$D^2f_1(x_1^\ast,\conj(x_1^\ast))$ is positive definite..
 \label{def:stackelberg}
\end{definition}
\begin{remark}
    Before moving on, let us make a few remarks about similar, and in some cases analogous, equilibrium definitions. For zero-sum games, the differential Stackelberg equilibrium
    notion is the same as a local min-max equilibrium for a sufficiently smooth
    cost function $f$. This is a well-known concept in
    optimization (see, e.g.,~\cite{basar:1998aa, daskin:1967aa, daskin:1967ab},
    among others), and it has
    recently been introduced in the learning literature~\cite{jin2019minmax}. 
    The benefit of the Stackelberg perspective is that it generalizes from zero-sum games to general-sum games, while the min-max equilibrium notion does not. A number of adversarial learning formulations are in fact general-sum, often as a result of regularization and well-performing heuristics that augment the cost functions of the generator or the discriminator.
\end{remark}

We utilize these local characterizations in terms of first and second order
conditions to formulate the myopic hierarchical learning algorithms we study.
Following the preceding discussion, consider 
the learning rule for each player to be given by
\begin{equation}
x_{i,k+1}=x_{i,k}-\gamma_{i,k}(\omega_{\mc{S},i}(x_k)+w_{i,k+1}),
\label{eq:noisyupdate}
\end{equation}
recalling that $\omega_{\mc{S}}=(Df_1(x), D_2f_2(x))$ and the notation
$\omega_{\mc{S},i}$ indicates the entry of $\omega_{\mc{S}}$ corresponding to
the $i$--th player. Moreover, $\{\gamma_{i,k}\}$ the sequence of learning rates
and $\{w_{i,k}\}$ is the noise process for player $i$, both of which satisfy the usual
assumptions from theory of stochastic approximation provided in detail in
Section~\ref{sec:results}. 
We note that the component of the update $\omega_{\mc{S},i}(x_k)+w_{i,k+1}$
captures the case in which each agent does not have oracle access to
$\omega_{\mc{S},i}$, but instead has an unbiased estimator for it. The given
update formalizes the class of learning algorithms we study in this paper.

\paragraph{Leader-Follower Timescale Separation.} We require a timescale separation between the leader
and the follower: the leader is assumed to be learning at a slower rate than the follower so that
$\gamma_{1,k}=o(\gamma_{2,k})$. The reason for this
timescale separation is that the leader's update is formulated using the reaction curve of
the follower. In the gradient-based learning setting considered, 
the reaction curve can be characterized by the set of critical points of
$f_2(x_{1,k}, \cdot)$ that have a local positive definite structure in the
direction of $x_2$, which is
\[\{x_2|\ D_2f_2(x_{1,k},x_2)=0, \ D_{2}^2f_2(x_{1,k},x_2)> 0\}.\]
This set can be characterized in terms of an
implicit map $\conj$, defined by the leader's belief that the follower
is playing a best response to its choice at each iteration,
which would imply $D_2f_2(x_{1,k},x_{2,k})=0$. Moreover, under sufficient regularity conditions, the
implicit mapping theorem~\cite{lee:2012aa} gives rise to the implicit map
$\conj:U\rar X_2:x_1\mapsto x_2$ on a neighborhood $U\subset X_1$ of
$x_{1,k}$. Formalized in Section~\ref{sec:results}, we note that
when $\conj$ is defined uniformly in $x_1$ on the domain for which convergence is
being assessed, the update in \eqref{eq:noisyupdate} is
well-defined in the sense that the component of the derivative $Df_1$
corresponding to the implicit dependence of the follower's action on $x_1$ via
$\conj$ is well-defined and locally consistent. In particular, for a given point
$x=(x_1,x_2)$ such that $D_2f_2(x_1,x_2)=0$ with $D_2^{2}f_2(x)$ an isomorphism, the
implicit function theorem implies there exists an open set $U\subset X_1$ such
that there exists a unique continuously differentiable function $r:U\rar X_2$
such that $r(x_1)=x_2$ and $D_2f_2(x_1,r(x_1))=0$ for all $x_1\in U$. Moreover, 
\[Dr(x_1)=-(D_2^2f_2(x_1,r(x_1)))^{-1}D_{21}f_2(x_1,r(x_1))\]
on $U$. Thus, in the limit of the two-timescale setting, the leader sees the
follower as having equilibriated (meaning $D_2f_2\equiv0$) so that
\begin{align}
Df_1(x_1,x_2)&=D_1f_1(x_1,x_2)+D_2f_1(x_1,x_2)Dr(x_1) \label{eq:unrolled}\\
&=D_1f_1(x_1,x_2)-D_2f_1(x_1,x_2)(D_2^2f_2(x_1,x_2))^{-1}D_{21}f_2(x_1,x_2). \notag
\end{align}
The map $r$ is an implicit representation of the follower's reaction curve.

\paragraph{Overview of Analysis Techniques.} The following describes the general
approach to studying the hierarchical learning dynamics in
\eqref{eq:noisyupdate}. The purpose of this overview is to provide the reader
with the high-level architecture of the analysis approach. 

The analysis techniques we employ
combine tools from dynamical systems theory with the theory of stochastic
approximation. In particular, we leverage the limiting continuous time dynamical
systems derived from~\eqref{eq:noisyupdate} to characterize concentration bounds
for iterates or samples generated by~\eqref{eq:noisyupdate}.  
We note that the hierarchical learning update in \eqref{eq:noisyupdate} with timescale
separation $\learnraten{1,k}=o(\learnraten{2,k})$ has a  limiting dynamical system
that takes the form of a
\emph{singularly perturbed} dynamical system given by
\begin{equation}
\begin{array}{lcl} \dot{x}_{1}(t)&=&-\tau Df_1(x_1(t),x_2(t))\\
        \dot{x}_2(t) & =&-D_2f_2(x_1(t),x_2(t))\end{array}
    \label{eq:singperturb}
\end{equation}
which, in the limit as $\tau\rar 0$, approximates
\eqref{eq:noisyupdate}.

The limiting dynamical system has known convergence properties (asymptotic
convergence in a region of attraction for a locally asymptotically stable
attractor). Such convergence properties can be translated in some sense
to the discrete time system by comparing \emph{pseudo-trajectories}---in this
case, linear interpolations between sample points of the update process---generated by sample
points of \eqref{eq:noisyupdate} and the limiting system flow for
initializations containing the set of sample points of \eqref{eq:noisyupdate}.
Indeed, the limiting dynamical system is then used to generate flows initialized from
the sample points generated by \eqref{eq:noisyupdate}. Creating
pseudo-trajectories, we then bound
the probability that the pseudo-trajectories deviate by some small amount from
the limiting dynamical system flow over each continuous time interval between
the sample points. A concentration bound can be constructed by taking a union bound over each time interval after a
finite time; following this we can guarantee the sample path has entered the region
of attraction, on which we can produce a Lyapunov function for the continuous
time dynamical system. 
The analysis in this paper is based on the high-level ideas outlined in this section. 

\subsection{Connections and Implications}
Before presenting convergence analysis of the update in \eqref{eq:noisyupdate},
we draw some connections to application domains---including
adversarial learning, where zero-sum game abstractions have been touted
for finding robust parameter configurations for neural networks, and opponent
shaping in multi-agent learning---and
equilibrium concepts commonly used in these domains. Let us first remind the reader of some
common definitions from dynamical systems theory. 

Given a sufficiently smooth function $f\in C^q(X,
\mb{R})$, a critical point $x^\ast$ of $f$ is said to be \emph{stable} if 
 for all $t_0\geq 0$ and $\vep>0$, there exists $\delta(t_0,\vep)$ such that
\[x_0\in B_{\delta}(x^\ast) \ \implies\ x(t)\in B_\vep(x^\ast), \ \forall t\geq
t_0\]
Further, $x^\ast$ is said to be \emph{asymptotically stable} if $x^\ast$ is
additionally attractive---that is, for all $t_0\geq 0$, there exists 
$\delta(t_0)$ such that 
\[x_0\in B_\delta(x^\ast)\ \implies \ \lim_{t\rar \infty}\|x(t)-x^\ast\|=0.\]
A critical point is said to be \emph{non-degenerate} if the determinant of the
Jacobian of the dynamics at the critical point is non-zero. 
For a non-degenerate critical point, the Hartman-Grobman
theorem~\cite{sastry:1999aa} enables us to
check the eigenvalues of the Jacobian to determine asymptotic
stability. In particular, at a non-degenerate critical point, if the eigenvalues
of the Jacobian are in the \emph{open left-half} complex plane, then the critical point
is asymptotically stable. The dynamical systems we study in this paper are of
the form $\dot{x}=-F(x)$ for some vector field $F$ determined by the gradient
based update rules employed by the agents. Hence, to determine if a critical
point is stable, we simply need to check that the spectrum of the Jacobian of $F$ is in the
\emph{open right-half} complex plane.

For the dynamics $\dot{x}=-\omega(x)$, let $J(x)$ denote the Jacobian of the
vector field $\omega(x)$. Similarly, for the dynamics $\dot{x}=-\omegas(x)$, let
$\Js(x)$ denote the Jacobian of the vector field $\omegas(x)$. Then, we say a
differential Nash equilibrium of a continuous game with corresponding individual
gradient vector field $\omega$ is stable if $\spec(J(x))\subset
\mb{C}_{+}^\circ$ where $\spec(\cdot)$
denotes the spectrum of its argument and $\mb{C}_{+}^\circ$ denotes the open
right-half complex plane. Similarly, we say differential Stackelberg equilibrium
is stable if $\spec(\Js(x))\subset \mb{C}_{+}^\circ$.
\subsubsection{Implications for Zero-Sum Settings}
\label{sec:zsimplications}
Zero-sum games are a very special class since there is a strong connection
between Nash equilibria and Stackelberg equilibria. Let us first show that for zero-sum games, attracting critical points of
$\dot{x}=-\omegas(x)$ are differential Stackelberg equilibria.
\begin{proposition}
Attracting critical points of $\dot{x}=-\omegas(x)$ in continuous zero-sum games are differential Stackelberg equilibria. That is, given a zero-sum game $(f,-f)$ defined by a sufficiently smooth function $f\in C^q(X, \mb{R})$ with $q\geq 2$, any stable critical point $x^\ast$ of the dynamics $\dot{x}=-\omegas(x)$ is a differential Stackelberg equilibrium. 
\label{prop:allstack}
\end{proposition}
\begin{proof}
Consider an arbitrary sufficiently smooth zero-sum game $(f,-f)$ on continuous
strategy spaces.
The Jacobian of the Stackelberg limiting
dynamics $\dot{x}=-\omegas(x)$ at a stable critical point
is $x^{\ast}$
\begin{equation}
    \Js(x^{\ast})=\bmat{D_1(Df)(x^{\ast}) & 0\\
-D_{21}f(x^{\ast}) & -D_{2}^2f(x^{\ast})} > 0.
\label{eq:stackjac}
\end{equation}
 The structure of the Jacobian $\Js(x^{\ast})$ follows from the fact that 
 \[D_2(Df)(x^{\ast}) = D_{12}f(x^{\ast})-D_{12}f(x^{\ast})(D_2^2f(x^{\ast}))^{-1}D_2^2f(x^{\ast})=0.\] 
The eigenvalues of a lower triangular block matrix are the union of the eigenvalues in each of the block diagonal components. This implies that if $\Js(x^{\ast}) > 0$, then necessarily $D_1(Df)(x^{\ast}) > 0$ and $-D_{2}^2f(x^{\ast}) > 0$. Consequently, any stable critical point of the Stackelberg limiting dynamics must be a differential Stackelberg equilibrium by definition.
\end{proof}
The result of Proposition~\ref{prop:allstack} implies that with appropriately
chosen stepsizes the only attracting critical points of the update rule in~\eqref{eq:noisyupdate} will be Stackelberg equilibria and thus, unlike simultaneous play individual gradient descent (known as gradient-play in the game theory literature), will not converge to spurious locally asymptotically stable attractors of the dynamics that are not relevant to the underlying game. 

In a recent work on GANs~\cite{metz:2017aa}, hierarchical learning of
a similar nature proposed in this paper is studied in the context of
zero-sum games. In the author's formulation, the generator is deemed the leader and the discriminator as the follower. The idea is to allow the discriminator to take $k$ individual gradient steps to update its parameters, while the parameters of the generator are held fixed. The effect of `unrolling' the discriminator update for $k$ steps is that a surrogate objective of $f(x_1, r_2(x_1))$ arises for the generator, meaning that the timescale-separation between the discriminator and the follower induces an update reminiscent of that given for the leader in~\eqref{eq:unrolled}. In particular, when $k\rightarrow\infty$ the follower converges to a local optimum as a function of the generator's parameters so that $D_2f(x_1, x_2)\rightarrow 0$. As a result, the critical points coincide with the Stackelberg dynamics we study, indicating that unrolled GANs are converging only to Stackelberg equilibria. Empirically, GANs learned with such timescale separation procedures seem to outperform gradient descent
with uniform stepsizes~\cite{metz:2017aa}, providing evidence Stackelberg equilibria can be sufficient in GANs.

This begs a further question of if attractors of the dynamics
$\dot{x}=-\omega(x)$ are Stackelberg equilibria. We begin to answer this inquiry by showing that
stable differential Nash are differential Stackelberg equilibria.

\begin{proposition}
    Stable differential Nash
    equilibria in continuous zero-sum games are differential
    Stackelberg equilibria. That is, given a zero-sum game $(f,-f)$ defined by a
    sufficiently smooth function $f\in C^q(X, \mb{R})$ with $q\geq 2$, a
    stable differential Nash
    equilibrium $x^\ast$ is a differential Stackelberg equilibrium. 
    \label{prop:DNEareDSE}
\end{proposition}
\begin{proof}
Consider an arbitrary sufficiently smooth zero-sum game $(f,-f)$ on continuous
strategy spaces.  Suppose $x^{\ast}$ is a stable differential Nash equilibrium so that
  by definition  $D_{1}^2f(x^{\ast})>0$,
    $-D_{2}^2f(x^{\ast})>0$, and 
    \[J(x^{\ast})=\bmat{D_{1}^2f(x^{\ast}) & D_{12}f(x^{\ast})\\ -D_{21}f(x^{\ast}) & -D_{2}^2f(x^{\ast})}>0.\]
Then, the Schur
 complement of $J(x^{\ast})$ is also positive definite: 
 \[D_{1}^2f(x^{\ast})-D_{21}f(x^{\ast})^\trans(D_2^2f(x^{\ast}))^{-1}D_{21}f(x^{\ast})>0\]
 Hence, $x^{\ast}$ is a
 differential Stackelberg equilibrium since the Schur complement of $J$ is exactly the
 derivative $D^2f$ at critical points and $-D_2^2f(x^{\ast})>0$ since $x$ is a
 differential Nash equilibrium.
\end{proof}
\begin{remark}
 In the zero-sum setting, the fact that Nash equilibria are a subset of Stackelberg equilibria (or
    minimax equilibria) for finite games is well-known~\cite{basar:1998aa}.
       We show the result for the notion of differential Stackelberg equilibria for
    continuous action space games that we introduce.  Similar to our work and concurrently, Jin et al.~\cite{jin2019minmax} also show that local Nash
    equilibria are local minmax solutions for continuous zero-sum games. 
 It is interesting to point
    out that for a subclass of zero-sum continuous games with a convex-concave 
    structure for the leader's cost the set
    of
    (differential) Nash and
   (differential) Stackelberg equilibria coincide. Indeed, $D_1^2f(x)>0$ at
   critical points for convex-concave games, so that if $x$ is a differential Stackelberg
   equilibrium, it is also a Nash equilibrium. 
\end{remark}

This result indicates that recent works seeking Nash equilibria in GANs are
seeking Stackelberg equilibria concurrently. Given that it is well-known simultaneous gradient play can converge to attracting critical points that do not satisfy the conditions of a Nash equilibria, it remains to determine when such spurious non-Nash attractors of the dynamics $\dot{x}=-\omega(x)$ will be an attractor of the Stackelberg dynamics
$\dot{x}=-\omegas(x)$. 

Let us start with a motivating question: \emph{when are non-Nash
    attractors of $\dot{x}=-\omega(x)$ differential Stackelberg equilibria?}
    It was shown by~\citet{jin2019minmax} that not all
attractors of $\dot{x}=-\omega(x)$ are local min-max or local max-min equilibria
since one can construct a function such that 
$D_1^2f(x)$ and $-D_2^2f(x)$ are \emph{both} not positive definite but $J(x)$ has positive
eigenvalues. 
It appears to be much harder to characterize when a non-Nash attractor of
$\dot{x}=-\omega(x)$ is a differential Stackelberg equilibrium since being a
differential Stackelberg equilibrium requires the follower's individual Hessian
to be positive
definite. Indeed, it reduces to a fundamental problem in linear algebra in which
the relationship between the eigenvalues of the sum of two matrices is largely unknown
without assumptions on the structure of the matrices~\cite{tao:2001aa}. For the class of zero-sum games, in what follows we provide some necessary and sufficient conditions
for non-Nash attractors at which the follower's Hessian is positive definite to be a differential Stackelberg equilibria. Before doing so, we present an illustrative example in which several attracting critical points of the simultaneous gradient play dynamics are not differential Nash equilibria but are differential Stackelberg equilibria---meaning points $x\in X$ at which $-D_2^2f(x)>0$, $\spec(J(x))\subset \mb{C}_+^\circ$, and $D_1^2f(x)-D_{21}f(x)^\top (D_2^2f(x))^{-1}D_{21}f(x)>0$.

\begin{example}[Non-Nash Attractors are Stackelberg.] 
Consider the zero-sum game defined by
\begin{equation}f(x)=-e^{-0.01(x_1^2+x_2^2)}((ax_1^2+x_2)^2+(bx_2^2+x_1)^2).\label{eq:polygame}\end{equation}
Let player $1$ be the leader who aims to minimize $f$ with respect to $x_1$
taking into consideration that player 2 (follower) aims to minimize $-f$ with
respect to $x_2$. In Fig.~\ref{fig:polygame}, we show the trajectories for
various initializations for this game with $(a,b)=(0.15, 0.25)$; it can be
seen that for several initializations, simultaneous gradient play leads to
non-Nash attractors which are differential Stackelberg equilibria. 
\begin{figure}[t]
    \centering
    \includegraphics[width=0.5\columnwidth]{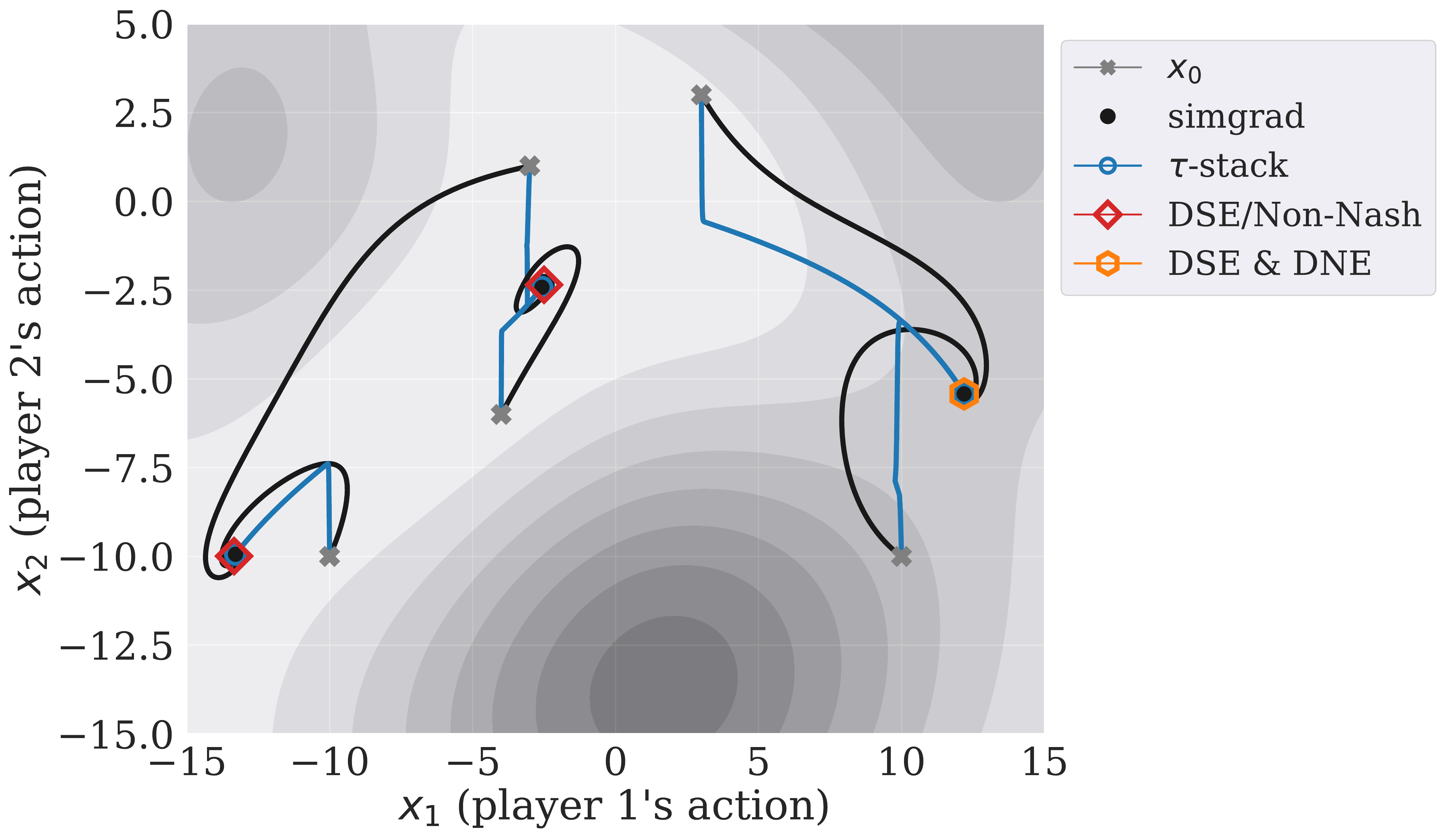}
    \caption{Simultaneous gradient play is attracted to non-Nash differential Stackelberg equilibria: The game is given by the pair of cost functions $(f,-f)$
    where $f$ is defined in \eqref{eq:polygame} with $a=0.15$ and $b=0.25$. There are
two non-Nash attractors of simultaneous gradient play which are also 
differential Stackelberg equilibria. }
    \label{fig:polygame}
\end{figure}
\label{example:nonnash_stack}
\end{example}

We now proceed to provide necessary and sufficient conditions for the phenenom demonstrated in Example~\ref{example:nonnash_stack}. Attracting critical points $x^\ast$ of the dynamics $\dot{x}=-\omega(x)$ that are not
Nash equilibria are such that either $D_{1}^2f(x^\ast)$ or $-D_2^2f(x^\ast)$
are not positive definite. Without loss of generality, considering player 1 to
be the leader, an attractor of the Stackelberg dynamics $\dot{x}=-\omegas(x)$
requires both
$-D_2^2f(x^\ast)$ and $D_1^2f(x^\ast)-D_{21}f(x^\ast)^\trans(D_2^2f(x^\ast))^{-1}D_{21}f(x^\ast)$ to
be positive
definite. 
Hence, if $-D_2^2f(x^\ast)$ is not positive definite at a non-Nash attractor of
$\dot{x}=-\omega(x)$,
then $x^\ast$ will also not be an attractor of $\dot{x}=-\omega_{\mc{S}}(x)$. 
We focus on non-Nash attractors with $-D_2^2f(x^\ast)>0$ and seek to determine when the Schur complement is positive definite, so that $x^\ast$ is an attractor of $\dot{x}=\omegas(x)$.

In the following two propositions, we need some addition notion that is common
across the two results. 
Let $x_1\in \mb{R}^{m}$ and
$x_2\in\mb{R}^{n}$. For a non-Nash attracting critical point $x^\ast$, let
$\spec(D_1^2f(x^\ast))=\{\mu_j, \ j\in \{1, \ldots, m\}\}$
where \[\mu_1\leq \cdots \leq \mu_r<0\leq \mu_{r+1}\leq \cdots \leq \mu_m,\] and
let $\spec(-D_2^2f(x^\ast))=\{\lambda_i, \ i\in \{1, \ldots, n\}\}$ where
\[\lambda_1\geq \cdots \geq \lambda_n>0,\] and define
$p=\dim(\ker(D_1^2f(x^\ast)))$.
\begin{proposition}[Necessary conditions]
Consider a non-Nash attracting critical point $x^\ast$ of the gradient dynamics $\dot{x}=-\omega(x)$ such that $-D_2^2f(x^\ast)>0$. Given $\kappa>0$ such that $\|D_{21}f(x^\ast)\|\leq \kappa$, if $D_1^2f(x^\ast)-D_{21}f(x^\ast)^\trans(D_2^2f(x^\ast))^{-1}D_{21}f(x^\ast)>0$,
then $r\leq n$ and $\kappa^2\lambda_i+\mu_i>0$ for all $i\in\{1, \ldots, r-p\}$.
\label{prop:nec}
\end{proposition}

\begin{proposition}[Sufficient conditions]
    Let $x^\ast$ be a non-Nash attracting critical point of the individual gradient dynamics $\dot{x}=-\omega(x)$ such
    that $D_1^2f(x^\ast)$ and $-D_2^2f(x^\ast)$ are Hermitian, and
    $-D_2^2f(x^\ast)>0$. Suppose that
    there exists a diagonal matrix (not necessarily positive) $\Sigma \in
    \mb{C}^{m \times n}$ with non-zero entries such that
    $D_{12}f(x^\ast)=\poneevecs \Sigma \ptwoevecs^\ast$ where $\poneevecs$ are the orthonormal eigenvectors
    of $D_1^2f(x^\ast)$ and $\ptwoevecs$ are orthonormal eigenvectors of
    $-D_2^2f(x^\ast)$. 
    Given $\kappa>0$ such that $\|D_{21}f(x^\ast)\|\leq \kappa$, if $r\leq
    n$ and $\kappa^2\lambda_i+\mu_i>0$ for each $i\in \{1, \ldots, r-p\}$, then $x^\ast$ is a
    differential Stackelberg equilibrium and an attractor of
    $\dot{x}=-\omegas(x)$. 
     \label{prop:suf}
\end{proposition}
The proofs of the above results follow from some results linear algebra and are
both in Appendix~\ref{app:linalg}.
Essentially, this says that if $D_1^2f(x^\ast)=\poneevecs M \poneevecs^\ast$ with
$\poneevecs \poneevecs^\ast=I_{n\times n}$ and $M$ diagonal, and
$-D_2^2f(x^\ast)=\ptwoevecs \Lambda \ptwoevecs^\ast$ with
$\ptwoevecs \ptwoevecs^\ast=I_{m\times m}$ and $\Lambda $ diagonal, then
 $D_{12}f(x^\ast)$ can be written as $\poneevecs\Sigma \ptwoevecs^\ast$ for some diagonal matrix
 $\Sigma \in
 \mb{R}^{n\times m}$ (not necessarily positive). Note that since $\Sigma$ does
 not necessarily have positive values, $\poneevecs\Sigma \ptwoevecs^\ast$ is not the singular
 value decomposition of $D_{12}f(x^\ast)$. In turn, this means that
 the
each eigenvector of $D_1^2f(x^\ast)$ get mapped onto a single eigenvector of
$-D_2^2f(x^\ast)$ through the transformation $D_{12}f(x^\ast)$ which describes
how player 1's variation $D_1f(x)$ changes as a function of player 2's choice.
With this structure for $D_{12}f(x^\ast)$, we can show that 
  $D_1^2f(x^\ast)-D_{21}f(x^\ast)^\trans(D_2^2f(x^\ast))^{-1}D_{21}f(x^\ast)>0$. If we remove the assumption that $\Sigma$ has non-zero entries,
  then the remaining assumptions are still sufficient to guarantee that 
  \[D_1^2f(x^\ast)-D_{21}f(x^\ast)^\trans(D_2^2f(x^\ast))^{-1}D_{21}f(x^\ast)\geq
  0.\]
  This means that $x^\ast$ does not satisfy the conditions for a differential
  Stackelberg, however, the point does satisfy necessary conditions for a local
  Stackelberg equilibrium and the point is a marginally stable attractor of the
  dynamics. 

While the results depend on conditions that are difficult to
check a priori without knowledge of $x^\ast$, certain classes of games for which these conditions hold
everywhere and not just at the equilibrium can be constructed. 
For instance, alternative conditions can be given: if the function $f$ which
defines the zero-sum game is such that it is concave in $x_2$ and there exists a 
$K$ such that 
\[D_{12}f(x)=KD_2^2f(x)\]
where $\sup_{x}\|D_{12}f(x)\|\leq \kappa<\infty$\footnote{Functions such that
    derivative of $f$
is Lipschitz will satisfy this condition.} and $K=\poneevecs\Sigma \ptwoevecs^\ast$ with $\Sigma$ again a (not necessarily positive)
diagonal matrix, then the results of Proposition~\ref{prop:suf} hold. From a
control point of view, one can think about the leader's update as having a
feedback term with the follower's input. 
On the other hand, the results are useful for the synthesis of games, such as in reward shaping or
incentive design, where the goal is to drive agents to particular desirable
behavior.

We remark that the fact that the eigenvalues of $J(x^\ast)$ are in
the open-right-half complex plane is not used in proving this result. We believe
that further investigation could lead to a less restrictive sufficient
condition. Empirically, by randomly generating the different block matrices, it
is quite difficult to find examples such that $J(x^\ast)$ has positive
eigenvalues, $-D_2^2f(x^\ast)>0$, and the Schur complement
$D_1^2f(x^\ast)-D_{21}f(x^\ast)^\trans (D_2^2f(x^\ast))^{-1}D_{21}f(x^\ast)$ is
not positive definite. In fact, for games on scalar action spaces, it turns out that non-Nash attracting critical points of the simultaneous gradient play dynamics at which $-D_2^2f(x^\ast)>0$ must be differential Stackelberg equilibria and attractors of the Stackelberg limiting dynamics.
\begin{corollary}
Consider a zero-sum game $(f, -f)$ defined by a sufficiently smooth cost function $f:\mb{R}^2\rar \mb{R}$ such that the action space is $X = \mb{R} \times \mb{R}$ and player 1 is deemed the leader and player 2 the follower. Then, any non-Nash attracting critical point of $\dot{x}=-\omega(x)$ at which
$-D_2^2f(x)>0$ is a differential Stackelberg equilibrium and an attractor of $\dot{x}=-\omega_{\mc{S}}(x)$.
\label{cor:2by2case}
\end{corollary}
\begin{proof}
Consider a sufficiently smooth zero-sum game $(f,-f)$ on continuous
strategy spaces defined by the cost function $f:\mb{R}^2\rar
\mb{R}$. Suppose $x^{\ast}$ is an attracting critical point of the dynamics $\dot{x}=-\omega(x)$ at which $-D_2^2f(x^{\ast})>0$ and $D_{1}^2f(x^{\ast})<0$ so that it is not a Nash equilibria. The Jacobian of the dynamics at a stable critical point is
\[J(x^{\ast})=\bmat{D_{1}^2f(x^{\ast}) & D_{12}f(x^{\ast})\\ -D_{21}f(x^{\ast}) & -D_{2}^2f(x^{\ast})}>0.\]
The fact that the real components of the eigenvalues of the Jacobian are positive implies that $D_{12}f(x^{\ast})D_{21}f(x^{\ast})> D_{1}^2f(x^{\ast})D_{2}^2f(x^{\ast})$ and $D_{1}^2f(x^{\ast}) > D_{2}^2f(x^{\ast})$ since the determinant and the trace of the Jacobian must be positive. Using this information, it directly follows that the Schur complement of $J(x^{\ast})$ is positive definite:
 \[D_{1}^2f(x^{\ast})-D_{12}f(x^{\ast})(D_2^2f(x^{\ast}))^{-1}D_{21}f(x^{\ast})>0.\]
As a result, $x^{\ast}$ is a
 differential Stackelberg equilibrium and an attractor of $\dot{x}=-\omega_{\mc{S}}(x)$ since the Schur complement of $J(x^{\ast})$ is the
 derivative $D^2f(x^{\ast})$ and $-D_2^2f(x)>0$ was given.
\end{proof}
We suspect that using the notion of quadratic numerical
range~\cite{tretter:2008aa}, which is a super set of the spectrum of a block operator matrix, along with the fact that the Jacobian of the simultaneous gradient play dynamics has its spectrum in the open right-half complex plane, may lead to an extension of the result to arbitrary dimensions.

The results of Propositions~\ref{prop:nec} and~\ref{prop:suf}, Corollary~\ref{cor:2by2case}, and Example~\ref{example:nonnash_stack} imply that some of the non-Nash attractors of $\dot{x}=-\omega(x)$ are in fact Stackelberg equilibria. This is a meaningful insight since recent works have proposed schemes to avoid non-Nash attractors of the dynamics as they have been classified or viewed as lacking game-theoretic meaning~\cite{mazumdar:2019aa}. 
Moreover, some recent empirical results show that a number of successful approaches to training GANs are not converging to Nash equilibria, but rather to non-Nash attractors of the dynamics~\cite{berard:2019aa}. It would be interesting to characterize whether
or not the attractors satisfy the conditions we propose, and if such conditions
could provide insights into how to improve GAN training. 
It also further suggests that the Stackelberg equilibria may be a suitable solution concept for GANs.

One of the common assumptions in some of the recent GANs literature is that the
discriminator network is zero in a neighborhood of an equilibrium parameter
configuration (see,
e.g.,~\cite{nagarajan2017gradient,nie:2019aa,mescheder2018training}). This assumption limits
the theory to the  `realizable' case; the work by \cite{nagarajan2017gradient}
provides relaxed assumptions for the non-realizable case. In both cases, the
Jacobian for the dynamics $\dot{x}=-\omega(x)$ is such that $D_1^2f(x^\ast)=0$. 

\begin{proposition}
    Consider a GAN satisfying the realizable assumption---that is, the discriminator
network is zero in a neighborhood of any equilibrium. Then, an attracting critical point for the
simultaneous gradient dynamics $\dot{x}=-\omega(x)$ at which $-D_2^2f$ is
positive semi-definite satisfies necessary
conditions for a local Stackelberg equilibrium, and it will be a marginally
stable point of the Stackelberg dynamics $\dot{x}=-\omegas(x)$. 
\label{prop:gans}
\end{proposition}
\begin{proof}
    Consider an attracting critical point $x$ of $\dot{x}=-\omega(x)$ such that
    $-D_2^2f(x^\ast)\geq 0$. Note that the realizable assumption implies that
the Jacobian of $\omega$ is
    \[J(x^{\ast})=\bmat{0 & D_{12}f(x^{\ast})\\ -D_{21} f(x^{\ast}) & -D_2^2f(x^{\ast})}\]
    (see, e.g.,~\cite{nagarajan2017gradient}). Hence, since $-D_2^2f(x^{\ast})\geq 0$, 
    \[-D_{21}^\trans f(x^\ast)(D_2^2 f)^{-1}(x^\ast)D_{21}f(x^\ast)\geq 0.\]
    Since $x^\ast$ is an attractor, $D_1f(x^\ast)=0$ and $D_2f(x^\ast)=0$ so
    that
    \[Df(x^\ast)=D_1f(x^\ast)+D_2f(x^\ast)(D_2^2f)^{-1}(x^\ast)D_{21}f(x^\ast)=0\]
  Consequently, the necessary conditions for a local Stackelberg equilibrium are
  satisfied. Moreover, since both $-D_2^2f(x^\ast)\geq 0$ and the Schur complement
  $-D_{21}^\trans f(x^\ast)(D_2^2)^{-1}f(x^\ast)D_{21}f(x^\ast)\geq 0$, the
  Jacobian of $\omegas$ is positive semi-definite so
  that the
  point $x^\ast$ is marginally stable. 
\end{proof}

Now, simply satisfying the necessary conditions is not enough to guarantee that
attractors of the simultaneous play gradient dynamics will be a local
Stackelberg equilibrium.
We can state sufficient conditions by examining Proposition~\ref{prop:suf}. 
\begin{proposition}
    Consider a GAN satisfying the realizable assumption---that is, the discriminator
    network is zero in a neighborhood of any equilibrium---and an attractor for the
simultaneous gradient dynamics $\dot{x}=-\omega(x)$ at which $-D_2^2f$ is
positive definite. Suppose that  there
    exists a  diagonal matrix $\Sigma$ with non-zero entries such that
    $D_{12}f(x^\ast)=\Sigma W$ where $W$ are the orthonormal eigenvectors of
    $-D_2^2f(x^\ast)$. Then, $x^\ast$ is a differential Stackelberg equilibrium
    and an attractor of $\dot{x}=-\omegas(x)$.
\label{prop:gans2}
\end{proposition}
The proof follows directly from Proposition~\ref{prop:suf} and
Proposition~\ref{prop:gans}. It is not directly clear how restrictive these
sufficient conditions are for GANs. We leave this for future inquiry.

\subsubsection{Connections to Opponent Shaping}
    Beyond the work in zero-sum games and applications to GANs, there has also
    been recent work, which we will refer to as `opponent shaping', where one or more players takes into
    account its opponents' response to their
    action~\cite{letcher:2018aa,foerster:2018aa, zhang:2010aa}. The
    initial work of \citet{foerster:2018aa} 
bears the most resemblance to the learning algorithms studied in this paper. The update rule (LOLA) considered there (in the deterministic setting with
constant stepsizes) takes the following form:
\begin{align*}
    x_1^+&=x_1-\gamma_1(D_1f_1(x)-\gamma_2D_2f_1(x) D_{21}f_2(x))\\
    x_2^+&=x_2-\gamma_2D_2f_2(x)
\end{align*}
The attractors of these dynamics are not necessarily Nash equilibria nor are
they Stackelberg equilibria as can be seen by looking at the critical points of
the dynamics.  Indeed, the LOLA dynamics lead only to Nash
or non-Nash stable attractors of the limiting dynamics. The effect of the
additional `look-ahead' term is simply that it changes the vector field and
region of attraction for stable critical points.  In the zero-sum case, however, the critical points of the above
are the same as those of simultaneous play individual gradient updates, yet the
Jacobian is not the same and it is still possible to converge to a non-Nash
attractor. 

With a few modifications, the above update rule can be massaged into a form
which more closely resembles the hierarchical learning rules we study in this paper. In
particular, if instead of $\learnraten{2}$, player 2 employed a Newton
stepsize of $(D_2^{2}f_2)^{-1}$, then the update would look like
\begin{align*}
    x_1^+&=x_1-\gamma_1(D_1f_1(x)-D_2f_1(x) (D_2^{2}f_2(x))^{-1}D_{21}f_2(x))\\
    x_2^+&=x_2-\gamma_2 D_2f_2(x)
\end{align*}
which resembles a deterministic version of \eqref{eq:noisyupdate}. The critical
points of this update coincide with the critical points of a Stackelberg game
$(f_1,f_2)$. With appropriately chosen stepsizes and with an initialization in a region on which the implicit map, which defines the
$-(D_2^2f_2(x))^{-1}D_{21}f_2(x)$ component of the update, is well-defined
uniformly in $x_1$, the above dynamics will converge to Stackelberg equilibria. 
In
this paper, we provide an in-depth convergence analysis and for
the stochastic setting\footnote{In~\cite{foerster:2018aa}, the authors do not
provide convergence analysis; they do in their extension, yet only for constant
and uniform stepsizes and for a learning rule that is different than the one
studied in this paper as all players are \emph{conjecturing} about the behavior
of their opponents. This distinguishes the present work from their setting.} of
the above update.

\subsubsection{Comparing Nash and Stackelberg Equilibrium Cost}
We have alluded to the idea that the ability to act first gives the leader a distinct advantage over the follower in a hierarchical game. We now formalize this statement with a known result that compares the cost of the leader at Nash and Stackelberg equilibrium.

\begin{proposition}(\citep[Proposition 4.4]{basar:1998aa}).
Consider an arbitrary sufficiently smooth two-player general-sum game $(f_1, f_2)$ on continuous strategy spaces. Let $f_{1}^{\mc{N}}$ denote the infimum of all Nash equilibrium costs for player 1 and $f_1^{\mc{S}}$ denote an arbitrary Stackelberg equilibrium cost for player 1. Then, if $\reac(x_1)$ is a singleton for every $x_1 \in X_1$, $f_1^{\mc{S}}\leq f_1^{\mc{N}}$.
\end{proposition}

This result says that the leader never favors the simultaneous play game over the hierarchical play game in two-player general-sum games with unique follower responses. On the other hand, the follower may or may not prefer the simultaneous play game over the hierarchical play game. 

The fact that under certain conditions the leader can obtain lower cost under a Stackelberg equilibrium compared to any of the Nash equilibrium may provide further explanation for the success of the methods in~\cite{metz:2017aa}. Commonly, the discriminator can overpower the generator when training a GAN~\cite{metz:2017aa} and giving the generator an advantage may mitigate this problem. In the context of multi-agent learning, the advantage of the leader in hierarchical games leads to the question of how the roles of each player in a game are decided. While we do not focus on this question, it is worth noting that when each player mutually benefits from the leadership of a player the solution is called concurrent and when each player prefers to be the leader the solution is called non-concurrent. We believe that exploring classes of games in which each solution concept arises is an interesting direction of future work.

\section{Convergence Analysis}
\label{sec:results}
Following the preceding discussion, consider 
the learning rule for each player to be given by
\begin{equation}
    x_{i,k+1}=x_{i,k}-\gamma_{i,k}(\omegasn{i}(x_k)+w_{i,k+1}),
    \label{eq:noisyupdate2}
\end{equation}
where recall that $\omegas=(Df_1(x), D_2f_2(x))$. Moreover, for each $i\in
\mc{I}$, $\{\gamma_{i,k}\}$ is the sequence of learning rates  and $\{w_{i,k}\}$
is the noise process for player $i$. As before, suppose player 1 is the leader
and conjectures that player 2 updates its action $x_{2}$ in each round via
$\conj(x_1)$. This setting captures the scenario in which players do not have
oracle access to their gradients, but do have an unbiased estimator. As an
example, players could be performing policy gradient reinforcement learning or
alternative gradient-based learning schemes. Let $\dim(X_i)=d_i$ for each
$i\in\mc{I}$ and $d=d_1+d_2$.
\begin{assumption}The following hold:
    \begin{enumerate}[itemsep=-2pt, topsep=0pt,
            label=\textbf{A\arabic{assumption}\alph*.}, leftmargin=25pt]
        \item     The maps $Df_1:\mb{R}^{d}\rar \mb{R}^{d_1}$, 
            $D_2f_2:\mb{R}^d\rar\mb{R}^{d_2}$ are $L_1$, $L_2$
            Lipschitz, and $\|Df_1\|\leq M_1<\infty$.
            \label{ass:lip}
        \item For each $i\in \mc{I}$, the learning rates satisfy
            $\sum_{k}\learnraten{i,k}=\infty$,
    $\sum_{k}\learnraten{i,k}^2<\infty$.
    \label{ass:learnrate}
\item  The noise processes $\{w_{i,k}\}$ are zero mean, martingale difference
    sequences. That is, 
    given the filtration $\mc{F}_k=\sigma(x_s, w_{1,s}, w_{2,s},\ s\leq k)$,
    $\{w_{i,k}\}_{i\in \mc{I}}$ are conditionally independent,
    $\mb{E}[w_{i,k+1}|\ \mc{F}_k]=0$ a.s., and $\mb{E}[\|w_{i,k+1}\||\
    \mc{F}_{k}]\leq c_i(1+\|x_{k}\|)$ a.s.~for some constants
    $c_i\geq 0$, $i\in \mc{I}$.
\label{ass:noise}
    \end{enumerate}
    \label{ass:all}
    \end{assumption}

Before diving into the convergence analysis, we need some machinery from dynamical systems theory. Consider the dynamics from~\eqref{eq:noisyupdate2} written as a continuous time combined system $\dot{\xi_t}=F(\xi_t)$ where $\xi_t(z)=\xi(t,z)$ is a continuous
map and $\xi=\{\xi_t\}_{t\in \mb{R}}$ is the flow of $F$. A set $A$ is said to be
\emph{invariant} under the flow $\xi$ if for all $t\in \mb{R}$, $\xi_t(A)\subset
A$, in which case $\xi|A$ denotes the semi-flow.
A point $x$ is an equilibrium if $\xi_t(x)=x$ for all $t$ and, of course, when
$\xi$ is induced by $F$, equilibria coincide with critical points of $F$. 
Let $X$ be a topological metric space with metric $\rho$, an example being $X=\mb{R}^d$
endowed with
the Euclidean distance.
\begin{definition}
A nonempty invariant set $A\subset X$ for $\xi$ is said to be internally
chain transitive   if for any $a,b\in A$ and $\delta>0$, $T>0$, there exists a
finite sequence $\{x_1=a, x_2, \ldots, x_{k-1},x_k=b; t_1, \ldots, t_{k-1}\}$
with $x_i\in A$ and $t_i\geq T$, $1\leq i\leq k-1$, such that
$\rho(\xi_{t_i}(x_i), x_{i+1})<\delta$, $\forall 1\leq i\leq k-1$. \label{def:chaintrans}
\end{definition}

\subsection{Learning Stackelberg Solutions for the Leader}\label{sec:leader_solutions}
Suppose that the leader (player 1) operates under the assumption that the
follower (player 2) is playing a local optimum in
each round. That is, given $x_{1,k}$, $x_{2,k+1}\in
\arg\min_{x_2}f_2(x_{1,k},x_2)$ for which $D_2f_2(x_{1,k},x_2)=0$ is a
first-order local optimality condition. 
If, for a given $(x_1,x_2)\in X_1\times X_2$, $D_{2}^2f_2(x_1,x_2)$ is invertible and
    $D_2f_2(x_1,x_2)=0$, 
then  the implicit function theorem implies that there
exists  neighborhoods $U\subset X_1$ and $V\subset X_2$ and a smooth map
$\conj:U\rar V$ such that $\conj(x_1)=x_2$.
\begin{assumption}
    For every $x_1$, $\dot{x}_2=-D_2f_2(x_1,x_2)$ has a globally asymptotically
    stable equilibrium $\conj(x_1)$ uniformly in $x_1$ and 
            $\conj:\mb{R}^{d_1}\rar\mb{R}^{d_2}$ is $L_r$--Lipschitz.
    \label{ass:gas}
\end{assumption}
Consider the leader's learning rule
\begin{equation}x_{1,k+1}=x_{1,k}-\gamma_{1,k}(Df_1(x_{1,k},x_{2,k})+w_{1,k+1})\label{eq:leader1}
\end{equation}
where $x_{2,k}$ is defined via the  map $r_2$ defined implicitly in a
neighborhood of $(x_{1,k},x_{2,k})$.

\begin{proposition}
    Suppose that for each $x\in X$, $D_{2}^2f_2$ is non-degenerate and
    Assumption~\ref{ass:all} holds for $i=1$. 
    Then, $x_{1,k}$ converges almost surely to an (possibly sample path
    dependent) equilibrium point $x_1^\ast$ which is a local Stackelberg solution
    for the leader. 
    Moreover, if Assumption~\ref{ass:all} holds for $i=2$ and
    Assumption~\ref{ass:gas} holds, 
    then $x_{2,k}\rar
    x_{2}^\ast=\conj(x_1^\ast)$ so that
    $(x_1^\ast,x_2^\ast)$ is a differential Stackelberg equilibrium.
    \label{prop:stackone}
\end{proposition}
\begin{proof}
This proof follows primarily from using known stochastic approximation results. The update rule in~\eqref{eq:leader1} is a stochastic approximation of $\dot{x}_1 = -Df_1(x_1, x_2)$ and consequently is expected to track this ODE asymptotically.  The main idea behind the analysis is to construct a continuous interpolated trajectory $\bar{x}(t)$ for $t \geq 0$ and show it asymptotically almost surely approaches the solution set to the ODE. Under Assumptions~\ref{ass:all}--\ref{ass:gastwo}, results from~\cite[\S2.1]{borkar:2008aa} imply that the sequence generated from~\eqref{eq:leader1} converges almost surely to a compact internally chain transitive set of~$\dot{x}_1 = -Df_1(x_1, x_2)$.
Furthermore, it can be observed that the only internally chain transitive
invariant sets of the dynamics are differential Stackelberg equilibria since at
any stable attractor of the dynamics $D^2f_1(x_1, \conj(x_1))>0$ and from
assumption $D_2^2f_2(x_1, \conj(x_1))>0$. Finally, from~\cite[\S2.2]{borkar:2008aa}, we can conclude that the update from~\eqref{eq:leader1} almost surely converges to a possibly sample path dependent equilibrium point since the only internally chain transitive invariant sets for~$\dot{x}_1 = -Df_1(x_1, x_2)$ are equilibria. 
The final claim that $x_{2,k}\to \conj(x_{1}^\ast)$ is guaranteed since $\conj$ is Lipschitz
and $x_{1,k}\to x_1^\ast$.
\end{proof}

   The above result can be stated with a relaxed version of
   Assumption~\ref{ass:gas}. 
   \begin{corollary}
       Given a differential Stackelberg
   equilibrium $x^\ast=(x_1^\ast,x_2^\ast)$, let
    $B_\radius(x^\ast)=B_{\radius_1}(x_1^\ast)\times B_{\radius_2}(x_2^\ast)$
    for some $\radius_1, \radius_2>0$ on which $D_2^2f_2$ is non-degenerate. Suppose that
    Assumption~\ref{ass:all} holds for $i=1$ and that $x_{1,0}\in
    B_{\radius_1}(x_1^\ast)$. Then, $x_{1,k}$ converges almost surely to $x_1^\ast$.
    Moreover, if  Assumption~\ref{ass:all} holds for $i=2$, $\conj(x_1)$ is a locally
    asymptotically stable equilibrium uniformly in $x_1$ on the ball
    $B_{\radius_2}(x_2^\ast)$, and $x_{2,0}\in B_{\radius_2}(x_2^\ast)$, then $x_{2,k}\rar
    x_2^\ast=\conj(x_1^\ast)$.
    \label{cor:local}
   \end{corollary}
   The proof follows the same arguments as the proof of
   Proposition~\ref{prop:stackone}.

\subsection{Learning Stackelberg Equilibria: Two-Timescale Analysis}\label{sec:tt_results}
Now, let us consider the case where the leader again operates under the
assumption 
that the follower is
playing (locally) optimally at each round so that the belief is $D_2f_2(x_{1,k},x_{2,k})=0$, but the
follower is actually performing the update
$x_{2,k+1}=x_{2,k}+g_2(x_{1,k}, x_{2,k})$ where $g_2\equiv -\gamma_{2,k}\mb{E}[D_2f_2]$. 
The learning dynamics in this setting are then
\begin{align}
 \label{eq:learnstacka}
       x_{1,k+1}&=x_{1,k}-\gamma_{1,k}(Df_1(x_k)+w_{1,k+1})\\
        x_{2,k+1}&=x_{2,k}-\gamma_{2,k}(D_2f_2(x_k)+w_{2,k+1})
    \label{eq:learnstackb}
\end{align}
where
$Df_1(x)=D_1f_1(x)+D_2f_1(x)D\conj(x_1)$.
Suppose that $\gamma_{1,k}\rar 0$ faster than $\gamma_{2,k}$ so that in the limit $\tau \rar
0$, the above approximates 
the singularly perturbed system defined by
\begin{equation}
   \begin{array}{lcl} \dot{x}_{1}(t)&=&-\tau Df_1(x_1(t),x_2(t))\\
        \dot{x}_2(t) & =&-D_2f_2(x_1(t),x_2(t))\end{array}
    \label{eq:singperturb}
\end{equation}
The learning rates can be seen as
stepsizes in a discretization scheme for solving the above dynamics. The
condition that $\gamma_{1,k}=o(\gamma_{2,k})$ induces a \emph{timescale
separation} in which $x_2$ evolves on a faster timescale than $x_1$. 
That is, the fast transient player is the \emph{follower} and the slow component is the \emph{leader} since
$\lim_{k\rar\infty}\gamma_{1,k}/\gamma_{2,k}=0$ implies that from the
perspective of the follower, $x_1$ appears quasi-static and from the perspective
of the leader, $x_2$ appears to have equilibriated, meaning
$D_2f_2(x_1,x_2)=0$ given $x_1$. 
From this point of view, the
learning dynamics \eqref{eq:learnstacka}--\eqref{eq:learnstackb} approximate the dynamics in the
preceding section. Moreover, stable attractors of the dynamics are such that
the leader is at a local optima for $f_1$, not just along its coordinate axis
but in both coordinates $(x_1,x_2)$ constrained to the manifold $\conj(x_1)$; this is to make a distinction between
differential Nash
equilibria in agents are at local optima aligned with their individual
coordinate axes.

\subsubsection{Asymptotic Almost Sure Convergence}
\label{subsubsec:asymp}
The following two results are fairly classical results in stochastic
approximation. They are leveraged here to making conclusions about convergence to
Stackelberg equilibria in hierarchical learning settings.

While we do not need the following assumption for all the results in this
section, it is required for asymptotic convergence of the two-timescale process
in \eqref{eq:learnstacka}--\eqref{eq:learnstackb}.
\begin{assumption}
    The dynamics $\dot{x}_1=-Df_1(x_1,\conj(x_1))$ have a globally asymptotically
    stable equilibrium.
    \label{ass:gastwo}
\end{assumption}

Under Assumption~\ref{ass:all}--\ref{ass:gastwo},
   and the assumption that $\gamma_{1,k}=o(\gamma_{2,k})$, classical results
   imply that
   the dynamics~\eqref{eq:learnstacka}--\eqref{eq:learnstackb} converge almost surely to a compact 
    internally chain transitive set $\mc{T}$ of \eqref{eq:singperturb}; see,
    e.g.,~\cite[\S6.1-2]{borkar:2008aa}, \cite[\S3.3]{bhatnagar:2013aa}.
Furthermore, it is straightforward to see that
    stable differential Nash equilibria are  internally chain transitive sets
since they are stable attractors of the dynamics $\dot{\xi}_t=F(\xi_t)$ from~\eqref{eq:singperturb}.

\begin{remark}
  There are two important points to remark on at this juncture. First,  the flow of the dynamics \eqref{eq:singperturb} is not    necessarily a
    gradient flow, meaning that the dynamics may admit non-equilibrium
    attractors such as periodic orbits.
The dynamics correspond to a gradient vector field if and only if  $D_{2}(Df_1)\equiv D_{12}f_2$, meaning
    when the dynamics admit a potential function. Equilibria may also not be
    isolated unless the Jacobian of $\omegas$, say $\Js$, is
    non-degenerate at the points. 
    Second, except in the case of zero-sum settings in which $(f_1,f_2)=(f,-f)$, non-Stackelberg locally asymptotically
    stable equilibria are attractors.  That is, convergence {does not} imply that
    the players have settled on a Stackelberg equilibrium, and this can occur
    even if the dynamics admit a potential. 
    \label{rem:gradflow}
\end{remark}

Let $t_k=\sum_{l=0}^{k-1}\gamma_{1,l}$ be the (continuous) time accumulated
after $k$ samples of the slow component $x_1$. Define $\xi_{1,s}(t)$ to be the
flow of $\dot{x}_1=-Df_1(x_1(t),\conj(x_1(t)))$ starting at time $s$ from
intialization $x_s$. 
\begin{proposition}
    Suppose that Assumptions~\ref{ass:all} and \ref{ass:gas} hold. Then,
    conditioning on the event $\{\sup_k\sum_i\|x_{i,k}\|^2<\infty\}$, for any integer
    $K>0$, $\lim_{k\rar \infty}\sup_{0\leq h\leq
    K}\|x_{1,k+h}-\xi_{1,t_k}(t_{k+h})\|_2=0$ almost surely.
    \label{prop:converge}
\end{proposition}
\begin{proof}
    The proof follows standard arguments in stochastic approximation. We simply
    provide a sketch here to give some intuition. First, we
    show that  conditioned on the event 
    $\{\sup_k\sum_i\|x_{1,k}\|^2<\infty\}$,
    $(x_{1,k},x_{2,k})\rar\{(x_1,\conj(x_1))| \ x_1\in \mb{R}^{d_1}\}$ almost
    surely.
    Let $\zeta_k=\frac{\gamma_{1,k}}{\gamma_{2,k}}(Df_1(x_{k})+w_{1,k+1})$.
    Hence the leader's
    sample path is generated by $x_{1,k+1}=x_{1,k}-\gamma_{2,k}\zeta_k$ which
    tracks $\dot{x}_1=0$ since $\zeta_k=o(1)$ so that it is asymptotically
    negligible. In particular, $(x_{1,k},x_{2,k})$ tracks
    $(\dot{x}_{1}=0,\dot{x}_2=-D_2f_2(x_1,x_2))$. 
    That is, on intervals $[\hat{t}_j,\hat{t}_{j+1}]$ where
    $\hat{t}_j=\sum_{l=0}^{j-1}\gamma_{2,l}$, the norm difference between interpolated trajectories of
    the sample paths and the trajectories of
    $(\dot{x}_{1}=0,\dot{x}_2=-D_2f_2(x_1,x_2))$ vanishes a.s.~as
    $k\rar \infty$. 
    Since the leader is tracking $\dot{x}_1=0$, the follower can
    be viewed as tracking $\dot{x}_2(t)=-D_2f_2(x_1,x_2(t))$. 
    Then applying
    Lemma~\ref{lem:hirsch} provided in Appendix~\ref{app:prelims}, $\lim_{k\rar
    0}\|x_{2,k}-\conj(x_{1,k})\|\rar 0$
    almost surely.    
    
    Now, by Assumption~\ref{ass:all}, $Df_1$ is Lipschitz and bounded (in fact,
    independent of \ref{ass:lip},
since $Df_1\in C^{q}$, $q\geq 2$, it is locally Lipschtiz
and, on the event $\{\sup_k\sum_i
        \|x_{i,k}\|_2<\infty\}$, it is bounded).  In turn, it induces a continuous globally integrable
        vector field, and therefore satisfies the assumptions
        of~\citet[Prop.~4.1]{benaim:1999aa}. Moreover, under Assumptions~\ref{ass:learnrate} and
\ref{ass:noise}, the
assumptions of~\citet[Prop.~4.2]{benaim:1999aa} are satisfied, which gives the
desired result.
\end{proof}
\begin{corollary}
    Under Assumption~\ref{ass:gastwo} and the assumptions of Proposition~\ref{prop:converge}, 
    $(x_{1,k},x_{2,k})\rar(x_1^\ast, \conj(x_1^\ast))$ almost surely conditioned on the
    event  
    $\textstyle\{\sup_k\sum_i\|x_{i,k}\|^2<\infty\}$. That is, the
    learning dynamics \eqref{eq:learnstacka}--\eqref{eq:learnstackb} converge to
    stable attractors of \eqref{eq:singperturb}, the set of which
    includes the stable differential Stackelberg equilibria.
\end{corollary}
\begin{proof}
Continuing with the conclusion of the proof of Proposition~\ref{prop:converge}, on intervals   $[t_k,t_{k+1}]$ the norm difference between
    interpolates of the sample path and the trajectories of
    $\dot{x}_1=-Df_1(x_1,\conj(x_1))$ vanish asymptotically; applying
    Lemma~\ref{lem:hirsch} (Appendix~\ref{app:prelims}) gives the result.
\end{proof}
Leveraging the results in Section~\ref{sec:zsimplications}, the convergence
guarantees are stronger since in zero-sum settings all attractors are
Stackelberg; this contrasts with the Nash equilibrium concept. 
\begin{corollary}
    Consider a zero-sum setting $(f,-f)$. Under the assumptions of Proposition~\ref{prop:converge} and
    Assumption~\ref{ass:gastwo}, conditioning on the event 
    $\{\sup_k\sum_i\|x_{i,k}\|^2<\infty\}$, 
     the
    learning dynamics \eqref{eq:learnstacka}--\eqref{eq:learnstackb} converge to
    a differential Stackelberg equilibria almost surely.
\end{corollary}
The proof of this corollary follows the above analysis and invokes
Proposition~\ref{prop:allstack}. As with Corollary~\ref{cor:local}, we can relax Assumption~\ref{ass:gas} and
\ref{ass:gastwo} to local asymptomatic stability assumptions and obtain similarity convergence guarantees.
\begin{corollary}
Given a differential Stackelberg equilibrium $x^\ast=(x_1^\ast,x_2^\ast)$ where $x_2^{\ast}=\conj(x_1^\ast)$, let
$B_\radius(x^\ast)=B_{\radius_1}(x_1^\ast)\times B_{\radius_2}(x_2^\ast)$
for some $\radius_1, \radius_2>0$ on which $D_2^2f_2$ is non-degenerate. 
Suppose that Assumption~\ref{ass:all} holds for each player, $r(x_1)$ is a locally asymptotically stable attractor uniformly in $x_1$ on the ball $B_{\radius_2}(x_2^\ast)$ for the dynamics $\dot{x}_2=-D_2f_2(x)$, and there exists a locally asymptotically stable attractor on $B_{\radius_1}(x_1)$ for the dynamics $\dot{x}_1=-Df_1(x_1,\conj(x_1))$. Then, given an initialization $x_{1, 0} \in B_{\radius_1}(x_1^\ast)$ and $x_{2, 0}\in B_{\radius_2}(x_2^\ast)$, it follows that $(x_{1, k}, x_{2, k})\rightarrow (x_1^{\ast}, x_2^{\ast})$ almost surely. 
\end{corollary}

\subsubsection{Finite-Time High-Probability Guarantees}
\label{sec:finitetime}
While asymptotic guarantees of the proceeding section are useful, high-probability finite-time
guarantees can be leveraged more directly in analysis and synthesis, e.g., of
mechanisms to coordinate otherwise autonomous agents. 
In this section, we aim to provide concentration bounds for the purpose of
deriving convergence rate and error bounds in support of this objective. 
The results in this section follow the very recent work by
\citet{borkar:2018aa}. We highlight key differences and, in particular, where
the analysis may lead to insights relevant for learning in hierarchical decision
problems between non-cooperative agents. 

Consider a locally asymptotically stable differential Stackelberg equilibrium
$x^\ast=(x_1^\ast,
\conj(x_1^\ast))\in X$ and let $B_{q_0}(x^\ast)$ be an ${q}_0>0$ radius ball around $x^\ast$ contained in the
region of attraction.
Stability implies that the Jacobian $\Js(x_1^\ast,
\conj(x_1^\ast))$ is positive definite and by the converse Lyapunov
theorem~\cite[Chap.~5]{sastry:1999aa} there
exists local Lyapunov functions for the dynamics $\dot{x}_1(t)=-\tau
Df_1(x_1(t),\conj(x_1(t)))$ and for the dynamics $\dot{x}_2(t)=-D_2f_2(x_1,
x_2(t))$, for each fixed $x_1$. In particular, there exists a local Lyapunov
function $V\in C^1(\mb{R}^{d_1})$ with $\lim_{\|x_1\|\uparrow
\infty}V(x_1)=\infty$, and $\la \nabla V(x_1), Df_1(x_1,\conj(x_1))\ra<0$ for $x_1\neq
x_1^\ast$. 
For $q>0$, let $V^q=\{x\in \text{dom}(V):\ V(x)\leq q\}$. Then, there is also 
$q>q_0>0$ and $\epsilon_0>0$ such that for $\epsilon<\epsilon_0$,
$\{x_1\in \mb{R}^{d_1}|\ \|x_1-x_1^\ast\|\leq \epsilon\}\subseteq V^{q_0}\subset
\mc{N}_{\epsilon_0}(V^{q_0})\subseteq V^q\subset \text{dom}(V)$ where
$\mc{N}_{\epsilon_0}(V^{q_0})=\{x\in \mb{R}^{d_1}|\ \exists x'\in V^{q_0}\
\text{s.t.} \|x'-x\|\leq \epsilon_0\}$. An analogously defined $\tilde{V}$
exists for the dynamics $\dot{x}_2$ for each fixed $x_1$. 

For now, fix $n_0$ sufficiently large; we specify the values of $n_0$ for which
the theory holds before the statement of Theorem~\ref{thm:conjecturetrack}. Define the
event $\mc{E}_n=\{\bar{x}_2(t)\in V^q\ \forall t\in [\tilde{t}_{n_0},
\tilde{t}_n]\}$ where
$\bar{x}_2(t)=x_{2,k}+\frac{t-\tilde{t}_k}{\gamma_{2,k}}(x_{2,k+1}-x_{2,k})$ are
linear interpolates---i.e., \emph{asymptotic pseudo-trajectories}---defined for $t\in (\tilde{t}_k, \tilde{t}_{k+1})$ with
$\tilde{t}_{k+1}=\tilde{t}_k+\gamma_{2,k}$ and $\tilde{t}_0=0$. 

The basic idea
of the proof is to leverage Alekseev's formula (Thm.~\ref{thm:alekseev},
    Appendix~\ref{app:prelims}) to bound the difference between the asymptotic
    pseudo-trajectories and
        the flow of the corresponding limiting differential equation on each
        continuous time interval between each of the successive iterates $k$ and
        $k+1$ by sequences of constants that decay asymptotically. Then, a union
        bound is used over all time intervals after defined for $n\geq n_0$ in
        order
        to construct a concentration bound. This is done first for the follower,
     showing that $x_{2,k}$ tracks the leader's 'conjecture' or belief
     $\conj(x_{1,k})$ about the follower's reaction, and
    then for the leader. 

Following~\citet{borkar:2018aa}, we can express the linear interpolates for any
$n\geq n_0$ as
    $\bar{x}_2(\tilde{t}_{n+1})\textstyle=\bar{x}_2(\tilde{t}_{n_0})-\sum_{k=n_0}^n\gamma_{2,k}(D_2f_2(x_{k})+w_{2,k+1})$
where
$\textstyle\gamma_{2,k}D_2f_2(x_{k})=\int_{\tilde{t}_k}^{\tilde{t}_{k+1}}D_2f_2(x_{1,k},\bar{x}_2(\tilde{t}_k))\
ds$
and similarly for the $w_{2,k+1}$ term.
Adding and subtracting $\int_{\tilde{t}_{n_0}}^{\tilde{t}_{n+1}}D_2f_2(x_{1}(s),
\bar{x}_2(s))\ ds$, Alekseev's formula 
can be applied to get 
\begin{align*}
    \bar{x}_2(t)&=x_2(t)+\Phi_2(t,s,x_1(\tilde{t}_{n_0}),\bar{x}_2(\tilde{t}_{n_0}))(\bar{x}_2(\tilde{t}_{n_0})\textstyle-x_2(\tilde{t}_{n_0}))+\int_{\tilde{t}_{n_0}}^t\Phi_2(t,s,x_1(s),\bar{x}_2(s))\zeta_2(s)\ ds
\end{align*}
where $x_1(t)\equiv x_1$ is constant (since $\dot{x}_1=0$), $x_2(t)=\conj(x_1)$,
and
\[\zeta_2(s)=-D_2f_2(x_1(\tilde{t}_k),\bar{x}_2(\tilde{t}_k))+D_2f_2(x_1(s),\bar{x}_2(s))+w_{2,k+1}.\]
In addition, for $t\geq s$, $\Phi_2(\cdot)$ satisfies linear
 system 
 \[\dot{\Phi}_2(t,s,x_{0})=J_2(x_1(t),x_2(t))\Phi_2(t,s,x_{0}),\]
 with $\Phi_2(t,s,x_{0})=I$ and $x_0=(x_{1,0},x_{2,0})$ and where $J_2$ the Jacobian of
$-D_2f_2(x_1,\cdot)$. We provide more detail on this derivation in Appendix~\ref{app:proofs}.

Given that $x^\ast=(x_1^\ast, \conj(x_1^\ast))$ is a stable differential
Stackelberg equilibrium, $J_2(x^\ast)$ is positive definite. Hence,  as in
\cite[Lem.~5.3]{thoppe:2018aa}, we can find
$M$,
$\kappa_2>0$ such that for $t\geq s$, $x_{2,0}\in V^q$,
$\|\Phi_2(t,s,x_{1,0},x_{2,0})\|\leq Me^{-\kappa_2(t-s)}$; this result follows
from 
standard results on stability of linear systems (see, e.g., \citet[\S7.2,
 Thm.~33]{callier:1991aa})  along with 
 a bound on
 \[\textstyle\int_{s}^t\big\|D^2_2f_2(x_{1},x_{2}(\tau,s,\tilde{x}_0))-D_2^2f_2(x^\ast)\big\|d\tau\]
 for $\tilde{x}_0\in V^q$ (see, e.g.,~\citet[Lem~5.2]{thoppe:2018aa}).

Now, an interesting point worth making is that this analysis
 leads to a very nice result for the leader-follower setting. In
particular, through the use of the auxiliary variable $z$, we can show that the
follower's sample path `tracks' the leader's conjectured sample path. Indeed,
consider $z_k=\conj(x_{1,k})$, that is, where $D_2f_2(x_{1,k},x_{2,k})=0$. Then,
using a Taylor expansion of the implicitly defined conjecture $\conj$, we get
    $z_{k+1}=z_k+D\conj(x_{1,k})(x_{1,k+1}-x_{1,k})+\delta_{k+1}$
where $\|\delta_{k+1}\|\leq L_{r}\|x_{1,k+1}-x_{1,k}\|^2$ is the error from
the remainder terms. Plugging in $x_{1,k+1}$, 
\begin{align*}
    z_{k+1}&=z_k+\gamma_{2,k}(-D_2f_2(x_{1,k},z_k)+\tau_kDr_2(x_{1,k})(w_{1,k+1}-Df_1(x_{1,k},x_{2,k}))+\gamma_{2,k}^{-1}\delta_{k+1}).
\end{align*}
The terms after $-D_2f_2$ are $o(1)$, and hence asymptotically negligible, so
that this $z$ sequence tracks dynamics as $x_{2,k}$. We show that with high
probability, they
asymptotically contract, leading to the conclusion that the
follower's dynamics track the leader's conjecture.

Towards this end, we first bound the normed difference between $x_{2,k}$ and
$z_{k}$. Define constants
\[H_{n_0}=\textstyle(\|\bar{x}_2(\tilde{t}_{n_0}-x_2(\tilde{t}_{n_0})\|+\|\bar{z}(\tilde{t}_{n_0})-x_2(\tilde{t}_{n_0})\|),\]
and \[S_{2,n}=\textstyle\sum_{k=n_0}^{n-1}\big(\int_{\tilde{t}_k}^{\tilde{t}_{k+1}}\Phi_2(\tilde{t}_n,s,x_{1}(\tilde{t}_k),\bar{x}_2(\tilde{t}_k))
ds)w_{2,k+1},
\]
and let $\tau_k=\gamma_{1,k}/\gamma_{2,k}$. 
\begin{lemma}
    For any $n\geq n_0$, there exists $K>0$ such that conditioned on
    ${\mc{E}}_n$,
    \begin{align*}
       \|x_{2,n}-z_n\|&\leq
       K\big(\|S_{2,n}\|+e^{-\kappa_2(\tilde{t}_n-\tilde{t}_{n_0})}H_{n_0}\textstyle+\sup_{n_0\leq k\leq n-1}\gamma_{2,k}+\sup_{n_0\leq k\leq
        n-1}\gamma_{2,k}\|w_{2,k+1}\|^2\notag\\
        &\qquad\textstyle+\sup_{n_0\leq k\leq n-1}\tau_k+\sup_{n_0\leq k\leq
        n-1}\tau_k\|w_{1,k+1}\|^2\big).
    \end{align*}
    
        \label{lem:defK}
\end{lemma}
Using this bound, we can provide an asymptotic guarantee that $x_{2,k}$ tracks
 $\conj(x_{1,k})$ and a high-probability guarantee that $x_{2,k}$
gets locked in to a ball around $\conj(x_1^\ast)$. Fix $\vep\in[0,1)$ and let $N$ be such that $\gamma_{2,n}\leq
\vep/(8K)$, $\tau_n\leq \vep/(8K)$ for all $n\geq N$. Let $n_0\geq N$ and with
$K$ as in Lemma~\ref{lem:defK}, let $T$ be such that
$e^{-\kappa_2(\tilde{t}_n-\tilde{t}_{n_0})}H_{n_0}\leq \vep/(8K)$ for all $n\geq
n_0+T$. 
\begin{theorem}
    Suppose that Assumptions~\ref{ass:all}, \ref{ass:gas}, and \ref{ass:gastwo}
    hold and let $\gamma_{1,k}=o(\gamma_{2,k})$. 
      Given a stable
    differential Stackelberg equilibrium $x^\ast=(x_1^\ast, \conj(x_1^\ast))$, the
    follower's sample path generated by \eqref{eq:learnstackb} with
    asymptotically track the leader's conjecture $z_k=\conj(x_{1,k})$ and, given
    $\vep\in[0,1)$,  will get `locked in' to a $\vep$--neighborhood with
    high probability conditioned on reaching  $B_{q_0}(x^\ast)$
    by iteration $n_0$.
    That is, letting $\bar{n}= n_0+T+1$, for some $C_1, C_2, C_3, C_4>0$,
    \begin{align}
        \mathrm{P}(\|x_{2,n}-z_n\|\leq \vep,& \forall n\geq \bar{n}| x_{2,n_0},z_{n_0}\in
        B_{q_0})\notag\\
        &\geq\textstyle
        1-\sum_{n=n_0}^{\infty}C_1e^{-C_2\sqrt{\vep/\gamma_{2,n}}}-\sum_{n=n_0}^\infty
        C_2e^{-C_2\sqrt{\vep/\tau_n}} -\sum_{n=n_0}^{\infty} C_3e^{-C_4\vep^2/\beta_n}.
        \label{eq:pxz}
    \end{align}
    with $\beta_n=\textstyle\max_{n_0\leq k\leq
    n-1}e^{-\kappa_2(\sum_{i=k+1}^{n-1}\gamma_{2,i})}\gamma_{2,k}$.
    \label{thm:conjecturetrack}
\end{theorem}
The key technique in proving the above theorem (which is done in detail
in~\citet{borkar:2018aa} using results from \citet{thoppe:2018aa}), is taking a union bound of the errors over all the continuous time
intervals defined for  $n\geq n_0$. 

The above theorem can be restated to give a
guarantee on getting locked-in to an $\vep$-neighborhood of a stable
differenital Stackelberg equilibria $x^\ast$ if the learning processes are initialized in
$B_{q_0}(x^\ast)$.
\begin{corollary}
    Fix $\vep\in[0,1)$ and suppose that $\gamma_{2,n}\leq \vep/(8K)$ for all
        $n\geq 0$. With $K$ as in Lemma~\ref{lem:defK}, let $T$ be such that
$e^{-\kappa_2(\tilde{t}_n-\tilde{t}_{0})}H_{0}\leq \vep/(8K)$ for all $n\geq T$. 
Under the assumptions of Theorem~\ref{thm:conjecturetrack}, $x_{2,k}$ will will get `locked in' to a $\vep$--neighborhood with
high probability conditioned on $x_0\in B_{q_0}(x^\ast)$ where the
high-probability bound is given in \eqref{eq:pxz} with $n_0=0$.
\label{cor:zerolock}
\end{corollary}

Given that the follower's action $x_{2,k}$ tracks $\conj(x_{1,k})$,
we can also show that $x_{1,k}$ gets locked into an $\vep$--neighborhood of
$x_1^\ast$ after a finite time with high probability.
First, a similar bound as in Lemma~\ref{lem:defK} can be
constructed for $x_{1,k}$. 

Define the event
$\hat{\mc{E}}_n=\{\bar{x}_1(t)\in V^{q}\ \forall t\in [\hat{t}_{n_0},
\hat{t}_n]\}$ where for each $t$,
$\bar{x}_1(t)=x_{1,k}+\frac{t-\hat{t}_k}{\gamma_{1,k}}(x_{1,k+1}-x_{1,k})$ is
a linear interpolates  between the samples $\{x_{1,k}\}$,
$\hat{t}_{k+1}=\hat{t}_k+\gamma_{1,k}$, and $\hat{t}_0=0$. Then as above,
Alekseev's formula can again be applied to get 
\begin{align*}
    \bar{x}_1
    &(t)=x_1(t,\hat{t}_{n_0},y(\hat{t}_{n_0}))+\Phi_1(t,\hat{t}_{n_0},
    \bar{x}_1(\hat{t}_{n_0}))\textstyle(\bar{x}_1(\hat{t}_{n_0})-x_1(\hat{t}_{n_0}))+\int_{\hat{t}_{n_0}}^t\Phi_1(t,s,\bar{x}_1(s))\zeta_1(s)\
    ds
\end{align*}
where $x_1(t)\equiv x_1^\ast$,
\begin{align*}
    \zeta_1(s)&=Df_1(x_{1,k},\conj(x_{1,k}))-Df_1(\bar{x}_1(s),\conj(\bar{x}_1(s)))+Df_1(x_k)-Df_1(x_{1,k},\conj(x_{1,k}))+w_{1,k+1},
\end{align*}
and $\Phi_1$ is the solution to a linear system
with dynamics $J_1(x_1^\ast, \conj(x_1^\ast))$, the Jacobian of
$-Df_1(\cdot,\conj(\cdot))$, and with initial data $\Phi_1(s,s,x_{1,0})=I$. This
linear system, as above, has bound $\|\Phi_1(t,s,x_{1,0})\|\leq
M_1e^{\kappa_1(t-1)}$ for some $M_1,\kappa_1>0$. Define $S_{1,n}=\sum_{k=n_0}^{n-1}\int_{\hat{t}_{k}}^{\hat{t}_{k+1}}\Phi_1(\hat{t}_n,s,\bar{x}_1(\hat{t}_k))ds
\cdot w_{1,k+1}$.
\begin{lemma}
    For any $n\geq n_0$, there exists $\bar{K}>0$ such that   conditioned on
    $\tilde{\mc{E}}_n$,
\setlength{\belowdisplayskip}{-8pt}\setlength{\belowdisplayshortskip}{-8pt}
    \begin{align*}
        \|\bar{x}_1(\hat{t}_n)-x_1(\hat{t}_n)\|\leq&\textstyle
        \bar{K}\big(\|S_{1,n}\|+\sup_{n_0\leq k\leq n-1}\|S_{2,k}\|
        \textstyle+\sup_{n_0\leq k\leq n-1}\gamma_{2,k}+\sup_{n_0\leq k\leq n-1}\tau_k\notag\\
        &\textstyle+\sup_{n_0\leq k\leq
        n-1}\gamma_{2,k}\|w_{2,k+1}\|^2+\sup_{n_0\leq k\leq
        n-1}\tau_k\|w_{1,k+1}\|^2+\sup_{n_0\leq
        k\leq n-1}\tau_kH_{n_0}\notag\\
        &\textstyle\quad+e^{\kappa_1(\hat{t}_n-\hat{t}_{n_0})}\|\bar{x}_1(\hat{t}_{n_0})-x_1(\hat{t}_{n_0})\|\big).
    \end{align*}
    \label{lem:defbarK}
\end{lemma}
Using this lemma, we can get the desired guarantees on $x_{1,k}$. Indeed, as
above,
 fix $\vep\in(0,1]$ and let $N$ be such that $\gamma_{2,n}\leq
\vep/(8K)$, $\tau_n\leq \vep/(8K)$, $\forall\ n\geq N$. Then, for any $n_0\geq N$ and
$K$ as in Lemma~\ref{lem:defK}, let $T$ be such that
$e^{-\kappa_2(\tilde{t}_n-\tilde{t}_{n_0})}H_{n_0}\leq \vep/(8K)$, $\forall \ n\geq
n_0+T$. Moreover, with $\bar{K}$ as in Lemma~\ref{lem:defbarK},  let
$e^{-\kappa_1(\hat{t}_n-\hat{t}_{n_0})}(\|\bar{x}_1(\hat{t}_{n_0})-{x}_1(\hat{t}_{n_0})\|\leq
\vep/(8\bar{K})$, $\forall n\geq n_0+T$.
\begin{theorem}
    Suppose that Assumptions~\ref{ass:all}--\ref{ass:gastwo}
    hold and that $\gamma_{1,k}=o(\gamma_{2,k})$.
 Given a stable
    differential Stackelberg equilibrium $x^\ast$ and $\vep\in[0,1)$, $x_k$ will
        get `locked in' to a $\vep$-neighborhood of $x^\ast$
        with high probability conditioned reaching $B_{q_0}(x^\ast)$ by
        iteration $n_0$.  That is, letting $\bar{n}= n_0+T+1$, 
        for some constants $\tilde{C}_j>0$, $j\in \{1,\ldots,6\}$, 
    \begin{align}
        \mathrm{P}(\|x_{1,n}-x_{1}(\hat{t}_n)\|\leq \vep, &\forall n\geq\bar{n}
        |x_{n_0},x_{n_0}\in B_{q_0}) \notag\\
        &\textstyle \geq 1+\sum_{n=n_0}^\infty
        \tilde{C}_1e^{-\tilde{C}_2\sqrt{\vep}/\sqrt{\gamma_{2,n}}}\textstyle -\sum_{n=n_0}^\infty
        \tilde{C}_1e^{-\tilde{C}_2\sqrt{\vep}/\sqrt{\tau_n}}\notag\\
        &\textstyle\quad - \sum_{n=n_0}^\infty
        \tilde{C}_3e^{-\tilde{C}_4\vep^2/\beta_n}\textstyle - \sum_{n=n_0}^\infty
        \tilde{C}_5e^{-\tilde{C}_6\vep^2/\eta_n}
        \label{eq:py}
    \end{align}
    with $\eta_n=\max_{n_0\leq k\leq
    n-1}\big(e^{-\kappa_1(\sum_{i=k+1}^{n-1}\gamma_{1,i})}\gamma_{1,k}\big)$.
    \label{thm:lockin}
\end{theorem}
An analogous corollary to Corollary~\ref{cor:zerolock}
can be stated for $x_{1,k}$ with $n_0=0$.

\section{Numerical Examples}
\label{sec:examples}
In this section, we present and extensive set of numerical examples to validate our theory and demonstrate that the learning dynamics in this paper can effectively train GANs\footnote{Code is available at \texttt{github.com/fiezt/Stackelberg-Code}.}. 

\subsection{Stackelberg Duopoly}
In Cournot's duopoly model a single good is produced by two firms so that the industry is a duopoly. The cost for firm $i=1, 2$ for producing $q_i$ units of the good is given by $c_iq_i$ where $c_i >0$ is the unit cost. The total output of the firms is $Q= q_1+q_2$. The market price is $P=A-Q$ when $A \geq Q$ and $P=0$ when $A < Q$. We can assume that $A > c_i$ for $i=1,2$.  
The profit of each firm is 
$\pi_i = Pq_i - c_iq_i = (A-q_i -q_{-i} -c_i) q_i$.
Moreover, the unique Nash equilibrium in the game is
$q_i^{\ast} = \frac{1}{3}(A + c_{-i} - 2c_i)$
so that the market price is 
$P^{\ast} = \frac{1}{3}(A + c_i + c_{-i})$
and each firm obtains a profit of
$\pi_i^{\ast} = \frac{1}{9}(A-2c_i + c_{-i})^2$.

\begin{figure}[t!]
    \centering
    \hspace*{0.75in}\subfloat[][]{\includegraphics[width=.3\textwidth]{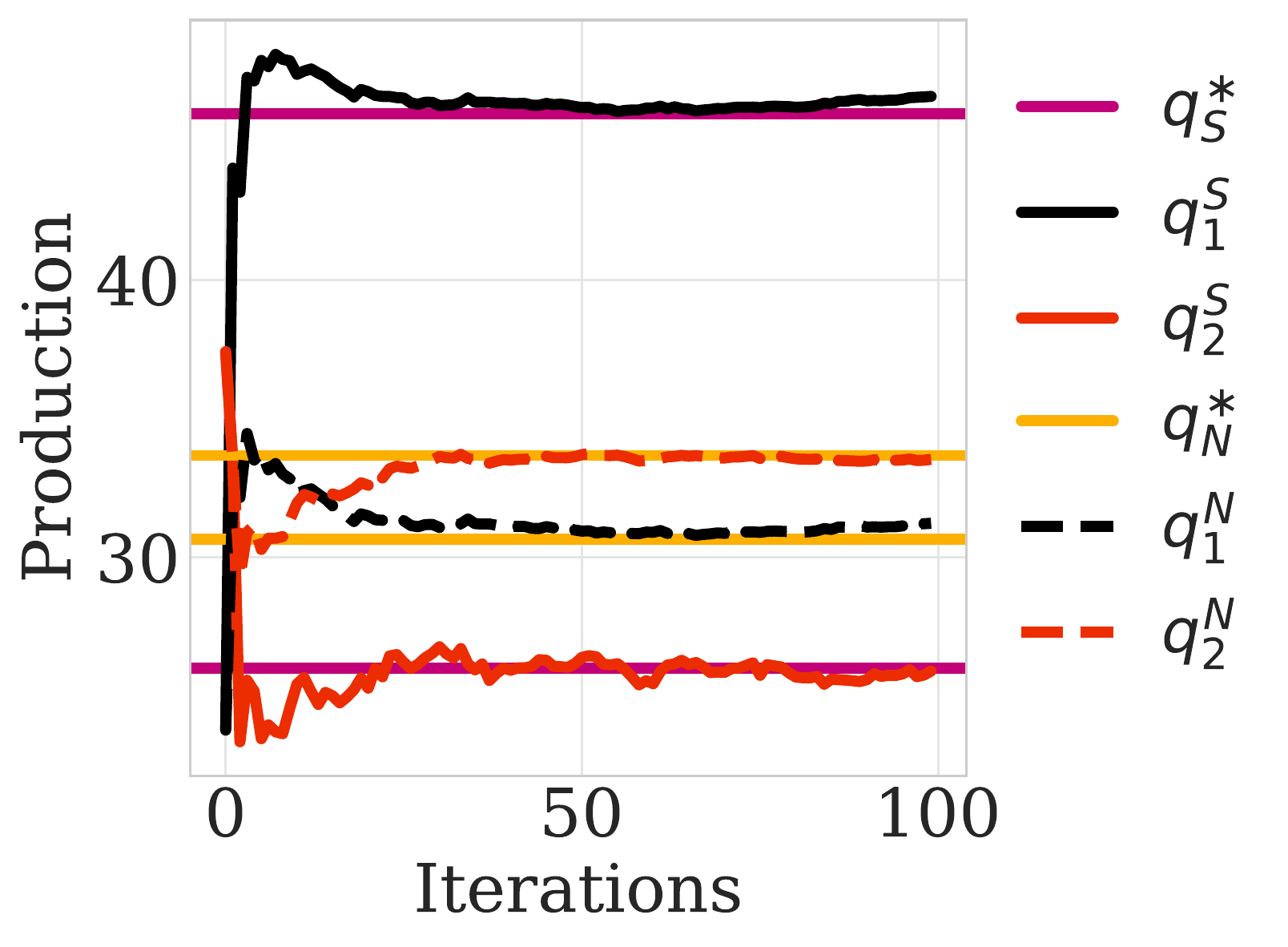}\label{fig:production}}\hfill
    \subfloat[][]{\includegraphics[width=.32\textwidth]{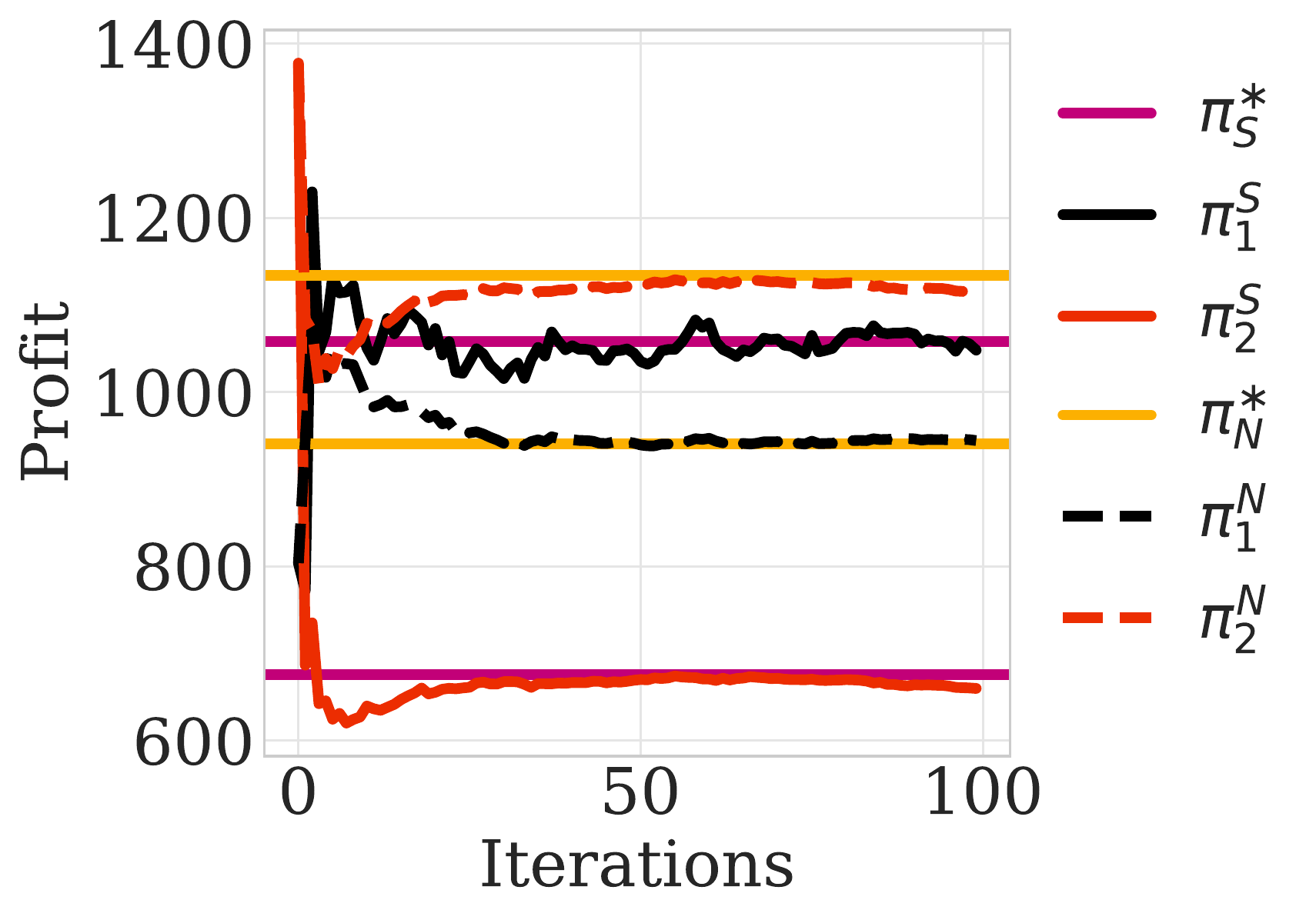}\label{fig:profit}}\hspace*{0.75in}
    \caption{(a) \emph{Firms' Production.} Sample learning paths for each firm showing the production
    evolution and convergence to the Nash equilibrium under the Nash dynamics
(i.e., simultaneous gradient-based learning using players' individual gradients
with respect to their own choice variable) and convergence to the Stackelberg
equilibrium under the Stackelberg dynamics. (b) \emph{Firms' Profit.} Evolution
of each firm's profit under the learning dynamics for both Nash and Stackelberg.
Similar convergence characteristics can be observed in (a) and (b). Of note is the improved profit obtained by the leader in the Stackelberg equilibrium compared to the Nash equilibrium.}
    \label{fig:duopoly}
\end{figure}

In the Stackelberg duopoly model with two firms, there is a leader and a follower. The leader moves and then the follower produces a best response to the action of the leader. Knowing this, the leader seeks to maximize profit taking advantage of the power to move before the follower. 
The unique Stackelberg equilibrium in the game is 
$q_1^{\ast} = \frac{1}{2}(A + c_2 - 2c_1)$, $q_2^{\ast} = \frac{1}{4}(A + 2c_1 -3c_2)$. In equilibrium the market price is 
$P^{\ast} = \frac{1}{4}(A + 2c_1 + c_2)$,
the profit of the leader is 
$\pi_1^{\ast} = \frac{1}{8}(A-2c_1 + c_2)^2$,
and the profit of the follower is 
$\pi_2^{\ast} = \frac{1}{16}(A+2c_1 -3c_2)^2$.

The key point we want to highlight is that in this game, firm 1's (leader) profit is always higher in the hierarchical play game than the simultaneous play game. We also use it as a simple validation example for our theory. For this problem, we simulate the Nash gradient dynamics and our two-timescale algorithm for learning Stackelberg equilibria to illustrate the distinctions between the Cournot and Stackelberg duopoly models. In this simulation, we select a decaying step-size of $\gamma_{i, k} = 1/k$  for each player in the Nash gradient dynamics. The decaying step-size is chosen to be $\gamma_{1,k} = 1/k$ for the leader and $\gamma_{2,k} = 1/k^{2/3}$ for the follower in the Stackelberg two-timescale algorithm so that the leader moves on a slower timescale than the follower as required. The noise at each update step is drawn as $w_{i, k}\sim \mathcal{N}(0, 10)$ for each firm. The parameters of the example are selected to be $A = 100, c_1=5, c_2 = 2$. In Figure~\ref{fig:duopoly} we show the results of the simulation. Figure~\ref{fig:production} shows the production path of each firm and Figure~\ref{fig:profit} shows the profit path of each firm. Under the Nash gradient dynamics, the firms converge to the unique Nash equilibrium of $q_N^{\ast} = (30.67, 33.67)$ that gives profit of $\pi_N^{\ast} = (944.4, 1114.7)$. The Stacklberg procedure converges to the unique Stackelberg equilibrium of $q_S^{\ast} = (46, 26)$ that gives profit of $\pi_S^{\ast} = (1048.2, 659.9)$. Hence as expected the two-timescale procedure converges to the Stackelberg equilibrium and gives the leader higher profit than under the Nash equilibrium.

\subsection{Location Game on Torus}
In this section, we examine a two-player game in which each player is selecting a position on a torus. Precisely, each player has a choice variable $\theta_i$ that can be chosen in the interval $[-\pi, \pi]$. The cost for each player is defined as $f_i(\theta_{i}, \theta_{-i}) = -\alpha_i \cos(\theta_i - \phi_i) + \cos(\theta_i - \theta_{-i})$, where each $\phi_i$ and $\alpha_i$ are constants. The cost function is such that each player must trade-off being close to $\phi_i$ and far from $\theta_{-i}$. 
For the simulation of this game, we select the parameters $\alpha = (1.0, 1.3)$ and $\phi = (\pi/8, \pi/8)$. There are multiple Nash and Stackelberg equilibria under these parameters. Each equilibrium is a stable equilibrium in this example. The Nash equilbria are $\theta_N^{\ast} = (-0.78, 1.18)$ and $\theta_N^{\ast} = (1.57, -0.4)$, and the costs are each $f(\theta_{N}^{\ast}) = (-0.77, -1.3)$ and $f(\theta_{N}^{\ast}) = (-0.77, -1.3)$. The Stackelberg equilbria are $\theta_S^{\ast} = (-0.53, 1.25)$ and $\theta_S^{\ast} = (1.31, -0.46)$, and the costs are each $f(\theta_{S}^{\ast}) = (-0.81, -1.05)$. Hence, the ability to play before the follower gives the leader a smaller cost at any equilibrium. The equilibrium the dynamics will converge to depends on the initialization as we demonstrate. For this simulation, we select a decaying step-size of $\gamma_{i, k} = 1/k^{1/2}$  for each player in the Nash gradient dynamics. The decaying step-size is chosen to be $\gamma_{1,k} = 1/k$ for the leader and $\gamma_{2,k} = 1/k^{1/2}$ for the follower in the Stackelberg two-timescale dynamics. The noise at each update step is drawn as $w_{i, k}\sim \mathcal{N}(0, 0.01)$ for each player. In Figure~\ref{fig:torus} we show the results of our simulation. The Nash and Stackelberg dynamics converge to an equilibrium as expected. In Figures~\ref{fig:path_nash} and~\ref{fig:path_stack}, we visualize multiple sample learning paths for the Nash and Stackelberg dynamics, respectively. The black lines depict $D_1f_1$ for Nash and $Df_1$ for Stackelberg and demonstrate how the order of play warps the first-order conditions for the leader and consequently produces equilibria which move away from the Nash equilibria. In Figure~\ref{fig:choice} we give a detailed look at the convergence to an equilibrium for a sample path. Finally, in Figure~\ref{fig:cost}, we present the evolution of the cost while learning and demonstrate the benefit of being the leader and the disadvantage of being the follower.

\begin{figure*}[t!]
    \centering
    \subfloat[][]{\includegraphics[width=.251\textwidth]{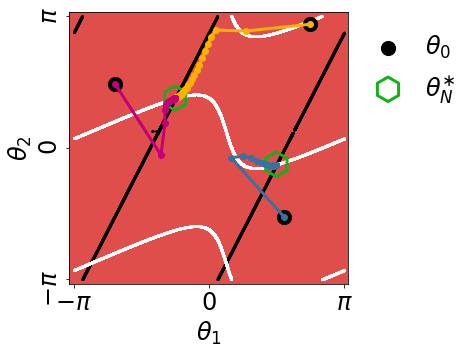}\label{fig:path_nash}}
    \subfloat[][]{\includegraphics[width=.251\textwidth]{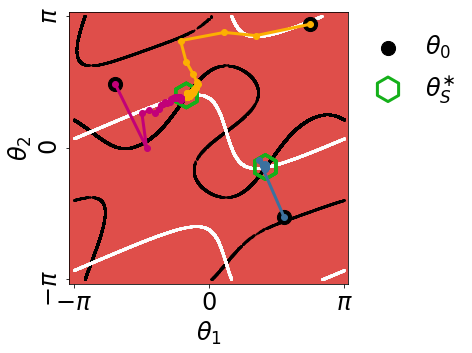}\label{fig:path_stack}}
    \subfloat[][]{\includegraphics[width=.245\textwidth]{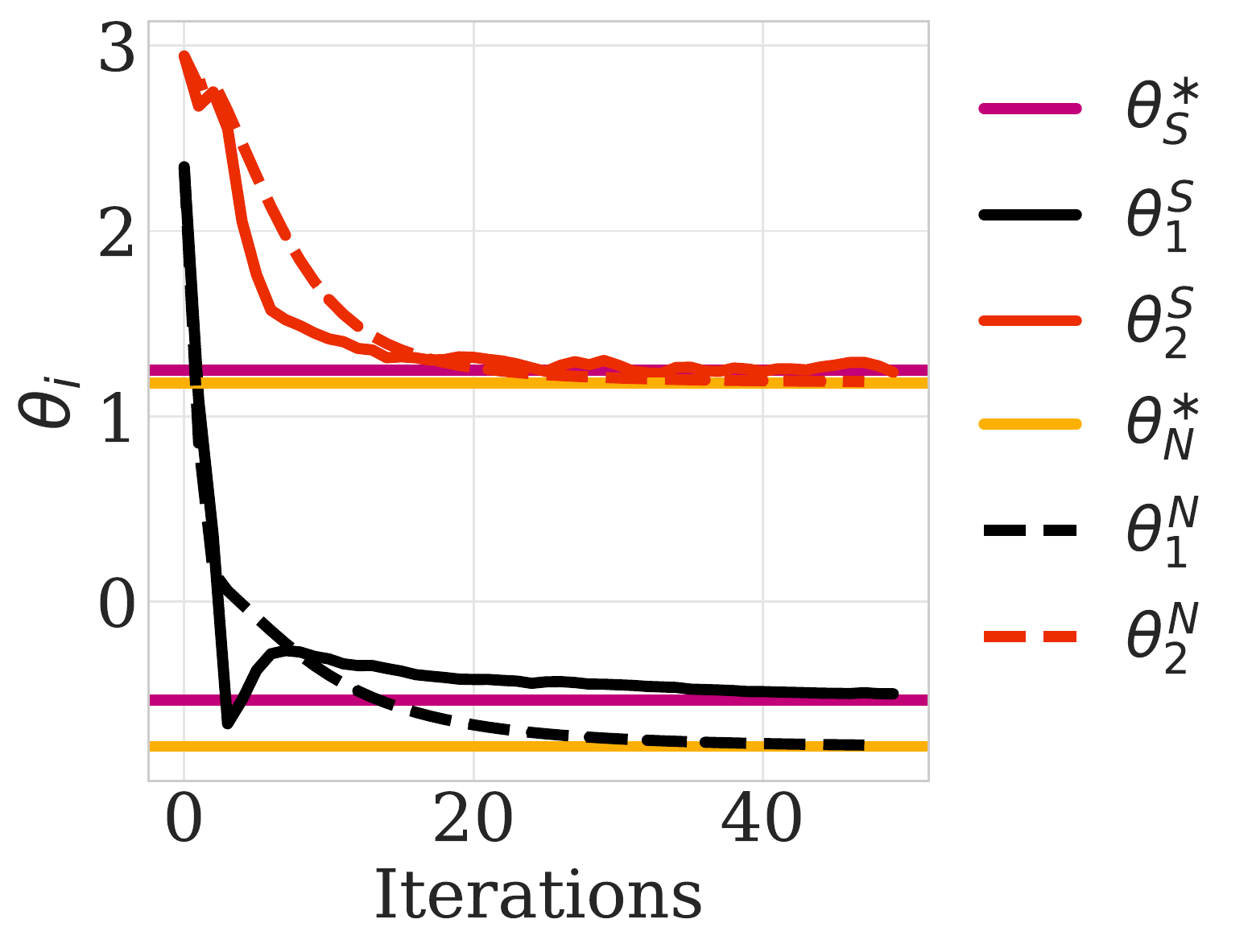}\label{fig:choice}}
    \subfloat[][]{\includegraphics[width=.255\textwidth]{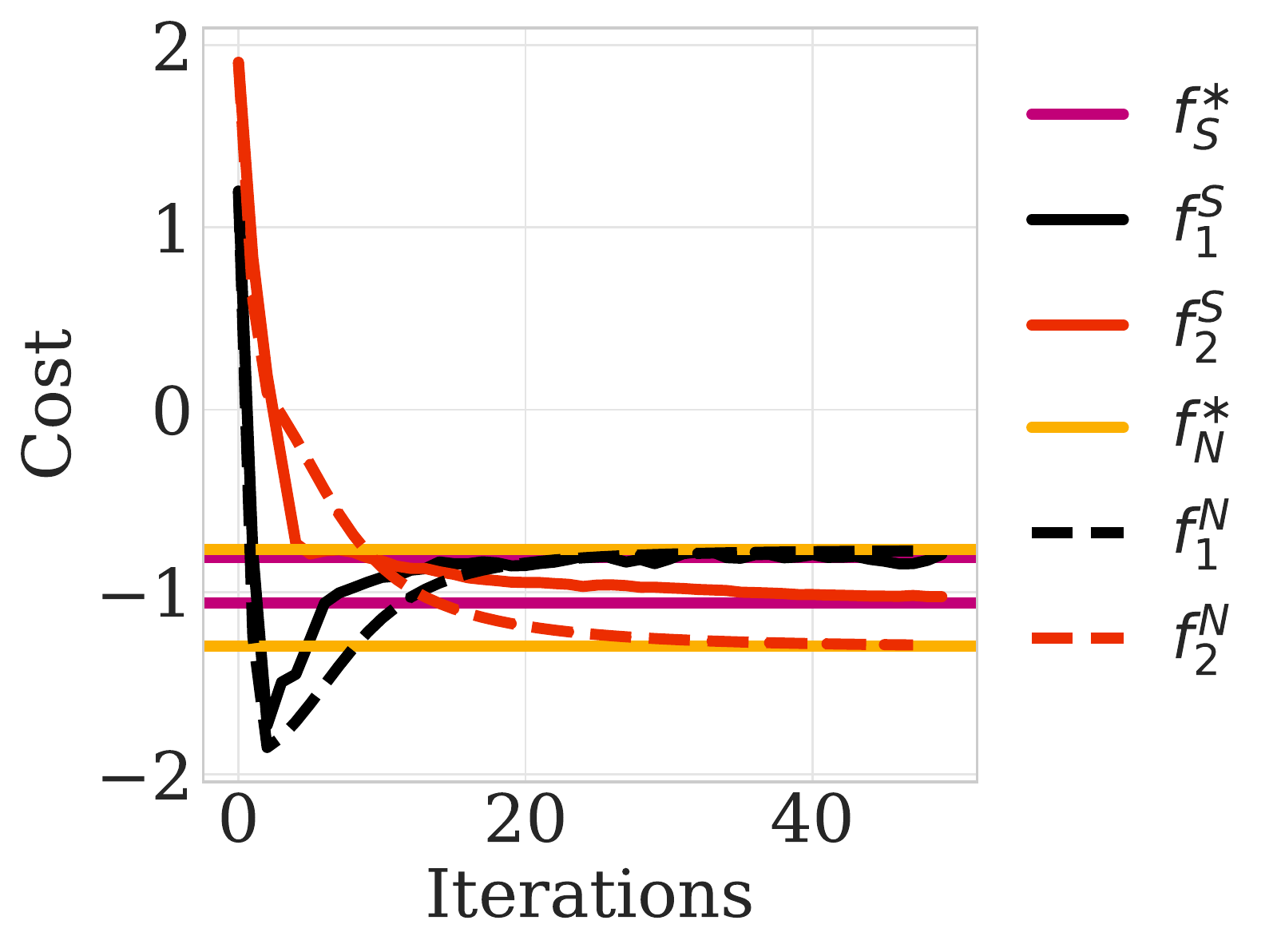}\label{fig:cost}}
    \caption{(a-b) Sample learning paths for each player showing the positions
    and convergence to local Nash equilibria under the Nash dynamics and
convergence to local Stackelberg equilibria under the Stackelberg dynamics.
The value of player 1's choice variable $\theta_1$ is shown on the horizontal
axis and the value of player 2's choice variable $\theta_2$ is shown on the
vertical axis. Note that the square depicts the unfolded torus where horizontal
edges are equivalent, vertical edges are equivalent, and the corners are all
equivalent. The black lines show $D_1f_1$ in (a) and $Df_1$ in (b) where the white lines show $D_2f_2$ in both (a) and (b). (c-d) Position and cost paths for each player for a sampled initial condition under the Nash and Stackelberg dynamics.}
    \label{fig:torus}
\end{figure*}

\subsection{Generative Adversarial Networks}
We now present a set of illustrative experiments showing the role of Stackelberg equilibria in the optimization landscape of GANs and the empirical benefits of training GANs using the Stackelberg learning dynamics compared to the simultaneous gradient descent dynamics. We find that the leader update empirically cancels out rotational dynamics and prevents cycling behavior. Moreover, we discover that the simultaneous gradient dynamics can empirically converge to non-Nash stable attractors that are Stackelberg equilibria in GANs. The generator and the discriminator exhibit desirable performance at such points, indicating that Stackelberg equilibria can be as desirable as Nash equilibria. We also find that the Stackelberg learning dynamics often converge to non-Nash stable attractors and reach a satisfying solution quickly using learning rates that can cause the simultaneous gradient descent dynamics to cycle. We provide details on our implementation of the Stackelberg leader update and the techniques to compute relevant eigenvalues of games in Appendix~\ref{appendix_secs:leaderupdate}. More details for specific hyperparameters can be found in Appendix~\ref{app:experiment_details}.

\textbf{Example 1: Learning a Covariance Matrix.}
We consider a data generating process of $x \sim \mathcal{N}(0, \Sigma)$, where the covariance $\Sigma$ is unknown and the objective is to learn it using a Wasserstein GAN. The discriminator is configured to be the set of quadratic functions defined as $D_W(x) = x^{\top}Wx$ and the generator is a linear function of random input noise $z \sim \mathcal{N}(0, I)$ defined by $G_V(z) = Vz$. The matrices $W \in \mathbb{R}^{m\times m}$ and $V \in \mathbb{R}^{m\times m}$ are the parameters of the discriminator and the generator, respectively. The Wasserstein GAN cost for the problem is  
$f(V, W) = \sum_{i=1}^m\sum_{j=1}^m W_{ij}(\Sigma_{ij}-\sum_{k=1}^m V_{ik}V_{jk})$.
We consider the generator to be the leader minimizing $f(V, W)$. The discriminator is the follower and it minimizes a regularized cost function defined by $-f(V, W) + \tfrac{\eta}{2}\Tr(W^{\top}W)$, where $\eta \geq 0$ is a tunable regularization parameter. 
The game is formally defined by the costs $(f_1, f_2)=(f(V, W), -f(V, W) + \tfrac{\eta}{2}\Tr(W^{\top}W))$, where player 1 is the leader and player 2 is the follower. In equilibrium, the generator picks $V^{\ast}$ such that $V^{\ast}(V^{\ast})^{\top}=\Sigma$ and the discriminator selects $W^{\ast}=0$. 

\begin{figure*}[t!]
  \subfloat[][$m=3$]{\includegraphics[height=.11\textwidth]{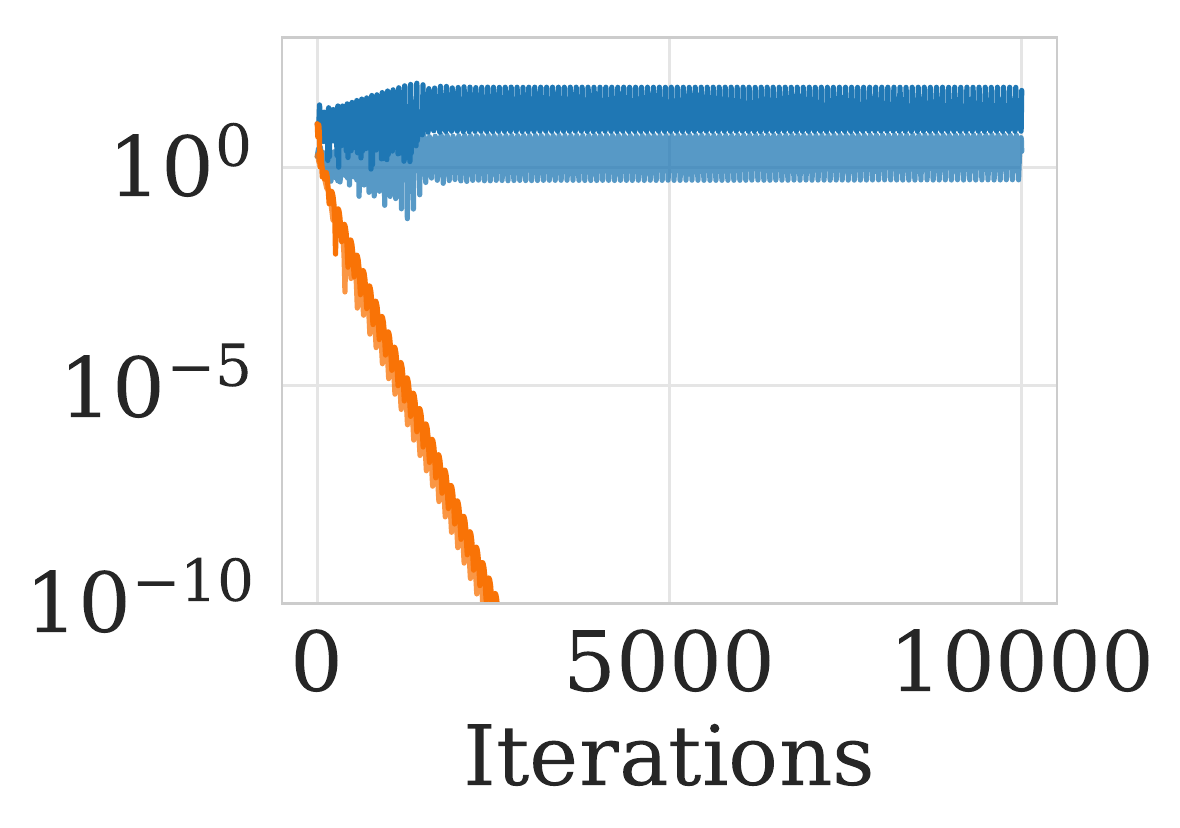}\label{fig:cov1}} \hfill
  \subfloat[][$m=9$]{\includegraphics[height=.11\textwidth]{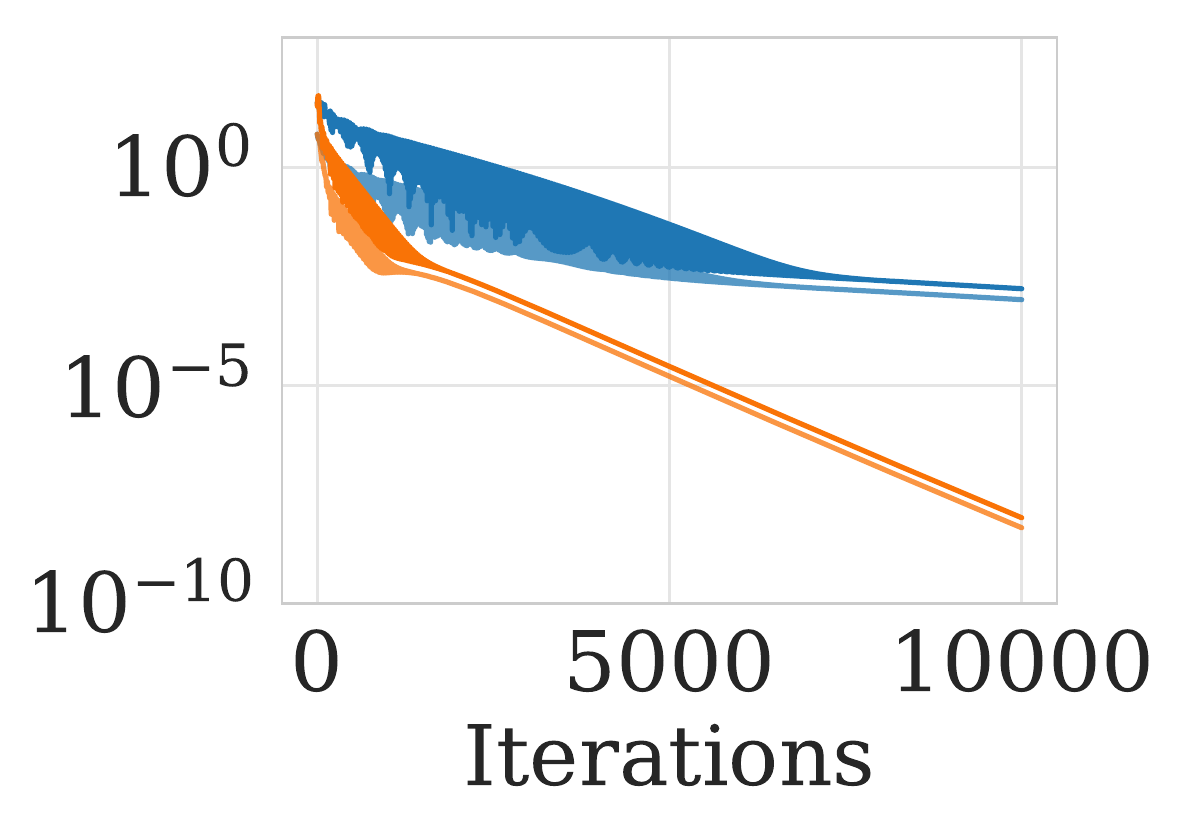}\label{fig:cov2}} \hfill
  \subfloat[][$m=25$]{\includegraphics[height=.11\textwidth]{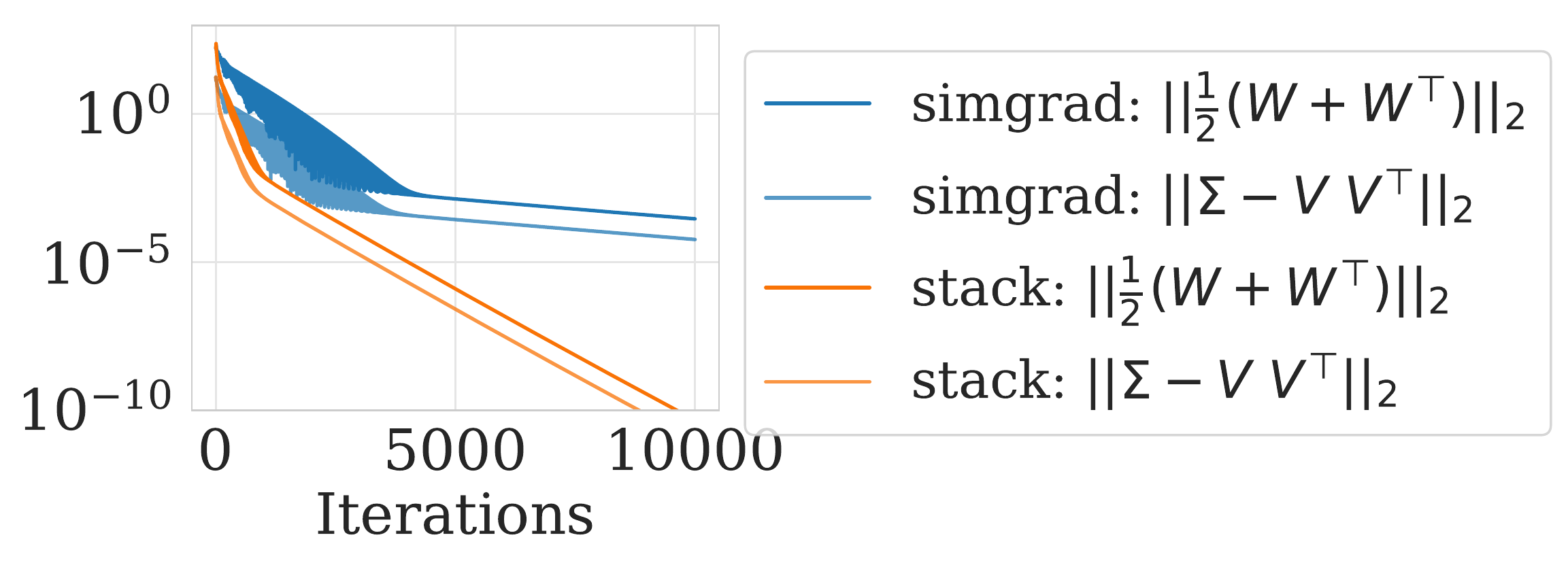}\label{fig:cov3}} \hfill
  \subfloat[][$m=3$]{\includegraphics[height=.11\textwidth]{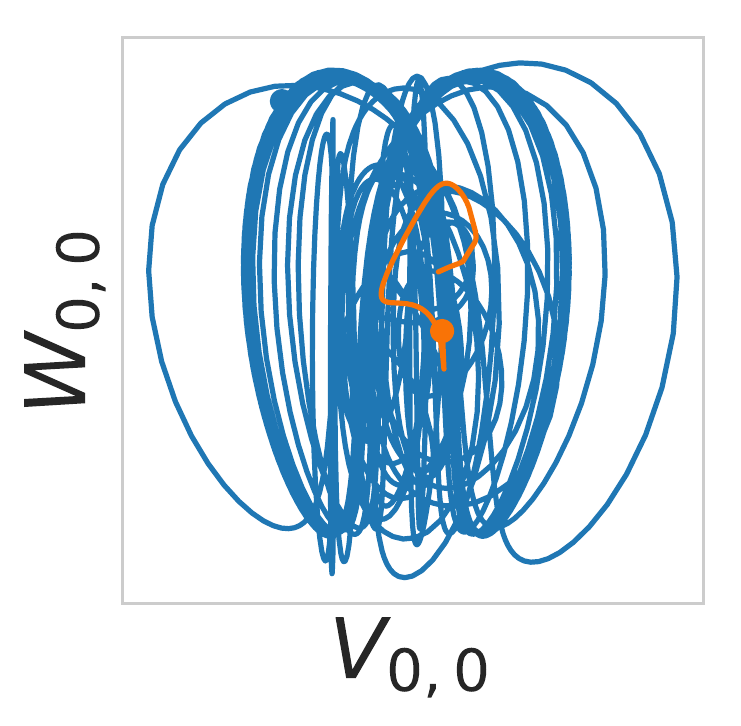}\label{fig:cov1_traj}} \hfill
  \subfloat[][$m=9$]{\includegraphics[height=.11\textwidth]{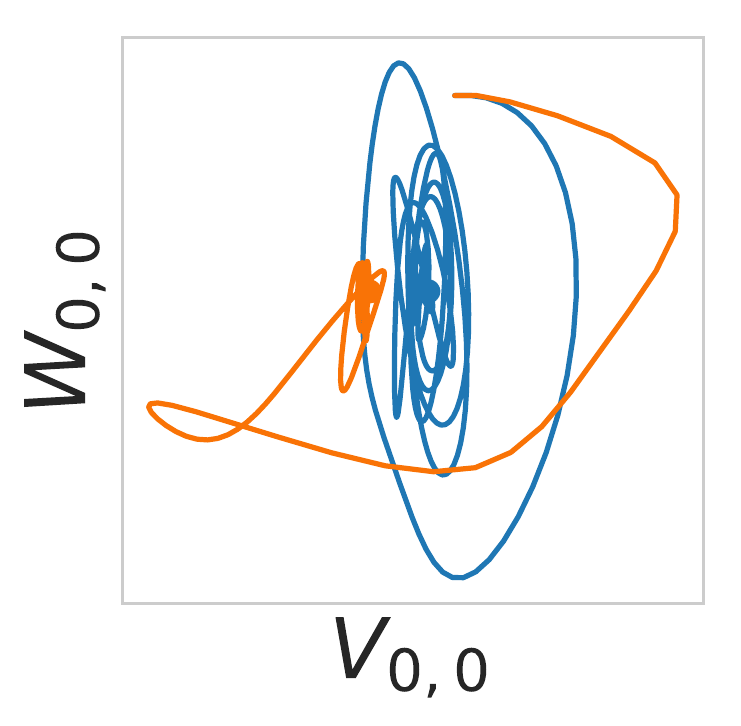}\label{fig:cov2_traj}} \hfill
  \subfloat[][$m=25$]{\includegraphics[height=.11\textwidth]{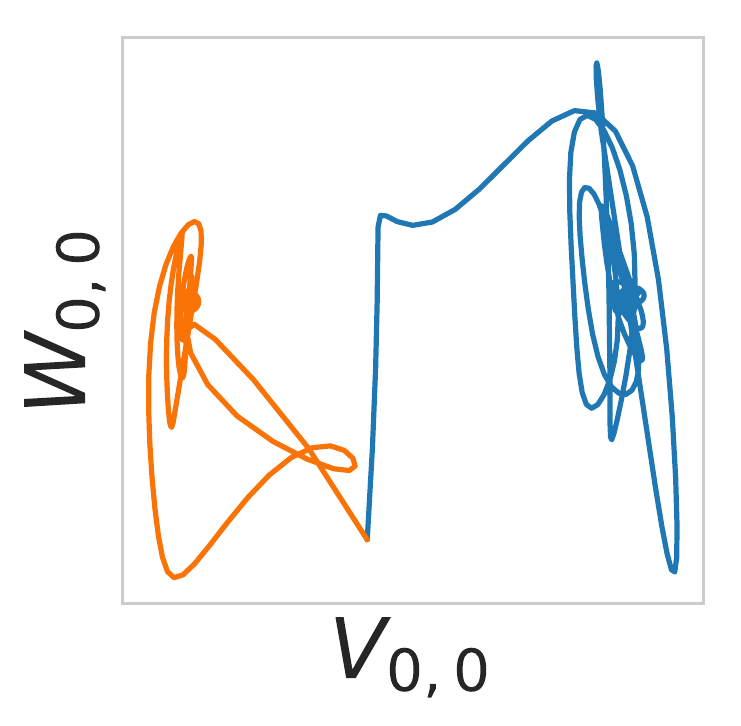}\label{fig:cov3_traj}} \hfill
  \caption{We estimate the covariance matrix $\Sigma$ with the Stackelberg learning dynamics, where the generator is the leader with choice variable $V\in \mathbb{R}^{m\times m}$ and discriminator is the follower with choice variable $W\in \mathbb{R}^{m\times m}$. Stackelberg learning can more effectively estimate the covariance matrix when compared with simultaneous gradient descent. We demonstrate the convergence for dimensions 3, 9, 25 in (a)--(c), with learning rates $\gamma_{1,k}=0.015(1-10^{-5})^k$, $\gamma_{2, k}=0.015(1-10^{-7})^k$ and regularization $\eta =m/5$. The trajectories of the first element of $W$ and $V$ are plotted over time in (d)--(f). Observe the cycling behavior of simultaneous gradient descent.}
  \label{fig:cov}
\end{figure*}

\begin{figure*}[t!]
  \subfloat[][Gen.]{\includegraphics[trim={0 -4ex 0 -5ex},clip,height=.10\textwidth]{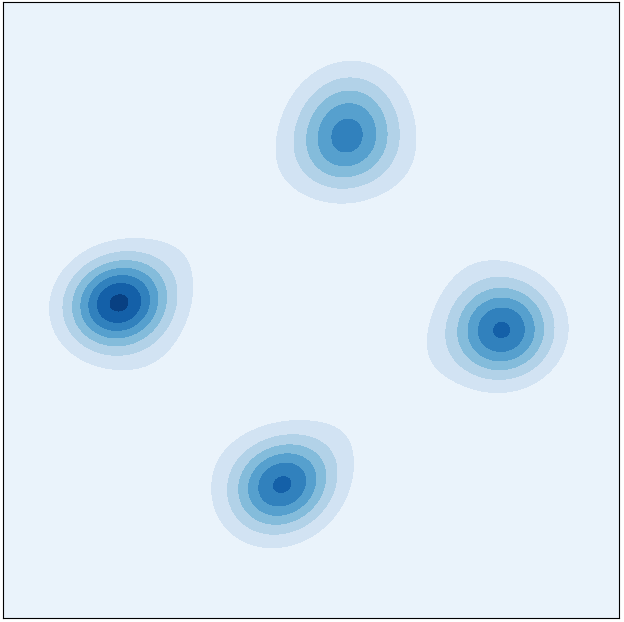}\label{fig:mog_11}} \hfill
  \subfloat[][Dis.]{\includegraphics[height=.10\textwidth]{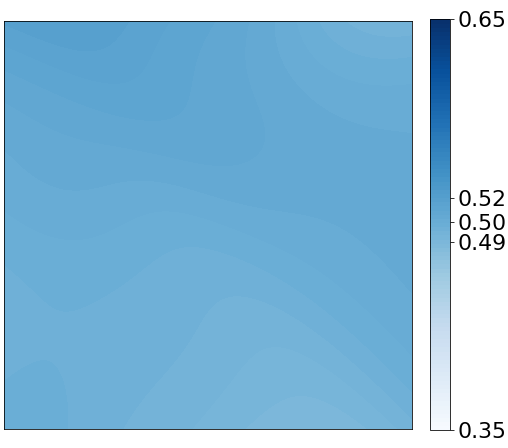}\label{fig:mog_12}}
  \subfloat[][$J$]{\includegraphics[trim={0 0 0 8ex},clip,height=.085\textwidth]{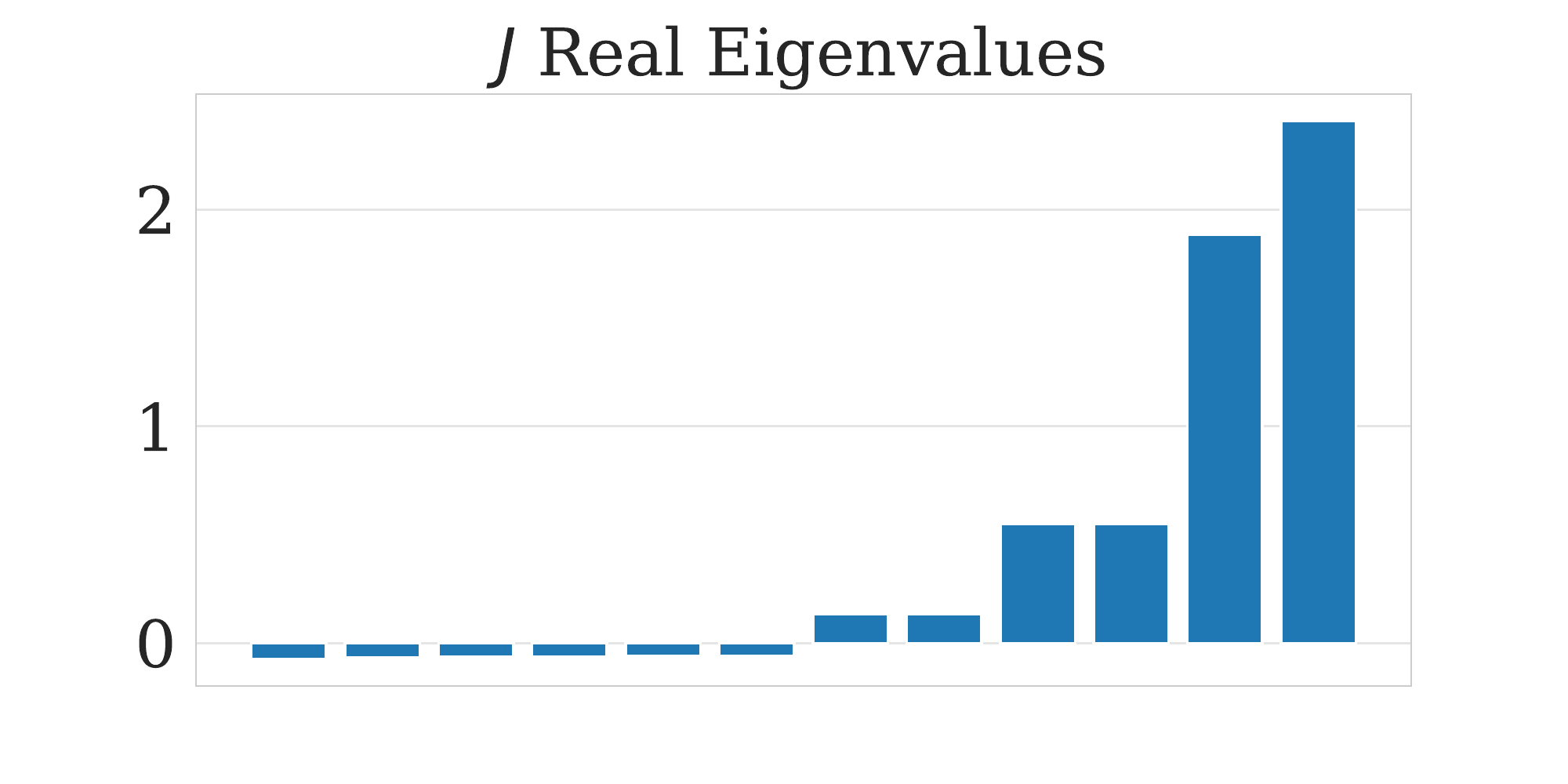}\label{fig:mog1_eig1sg}} 
  \subfloat[][$S_1$]{\includegraphics[trim={0 0 0 8ex},clip,height=.085\textwidth]{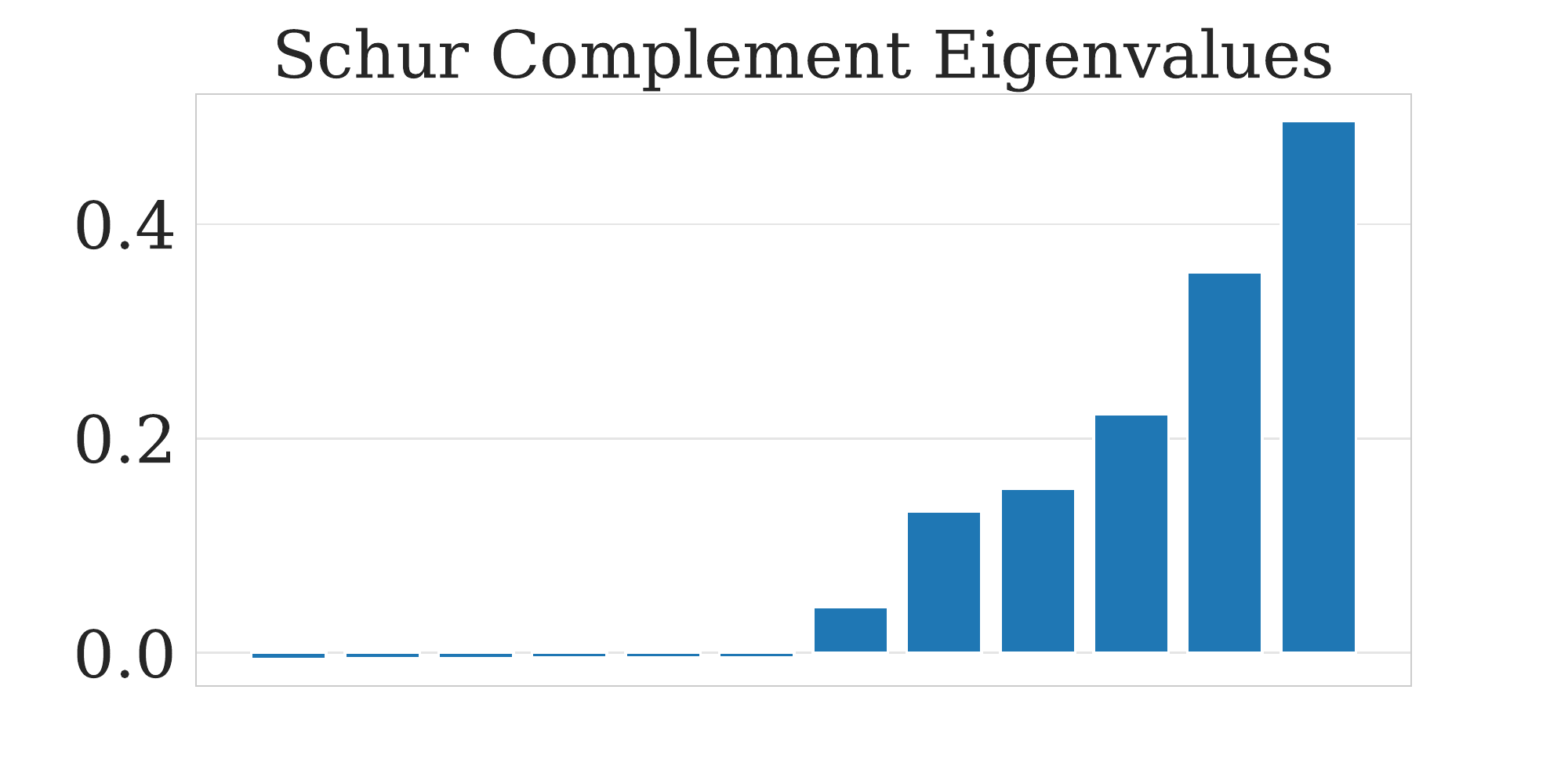}\label{fig:mog1_eig2sg}} 
  \subfloat[][$D_1^2f_1$]{\includegraphics[trim={0 0 0 8ex},clip,height=.085\textwidth]{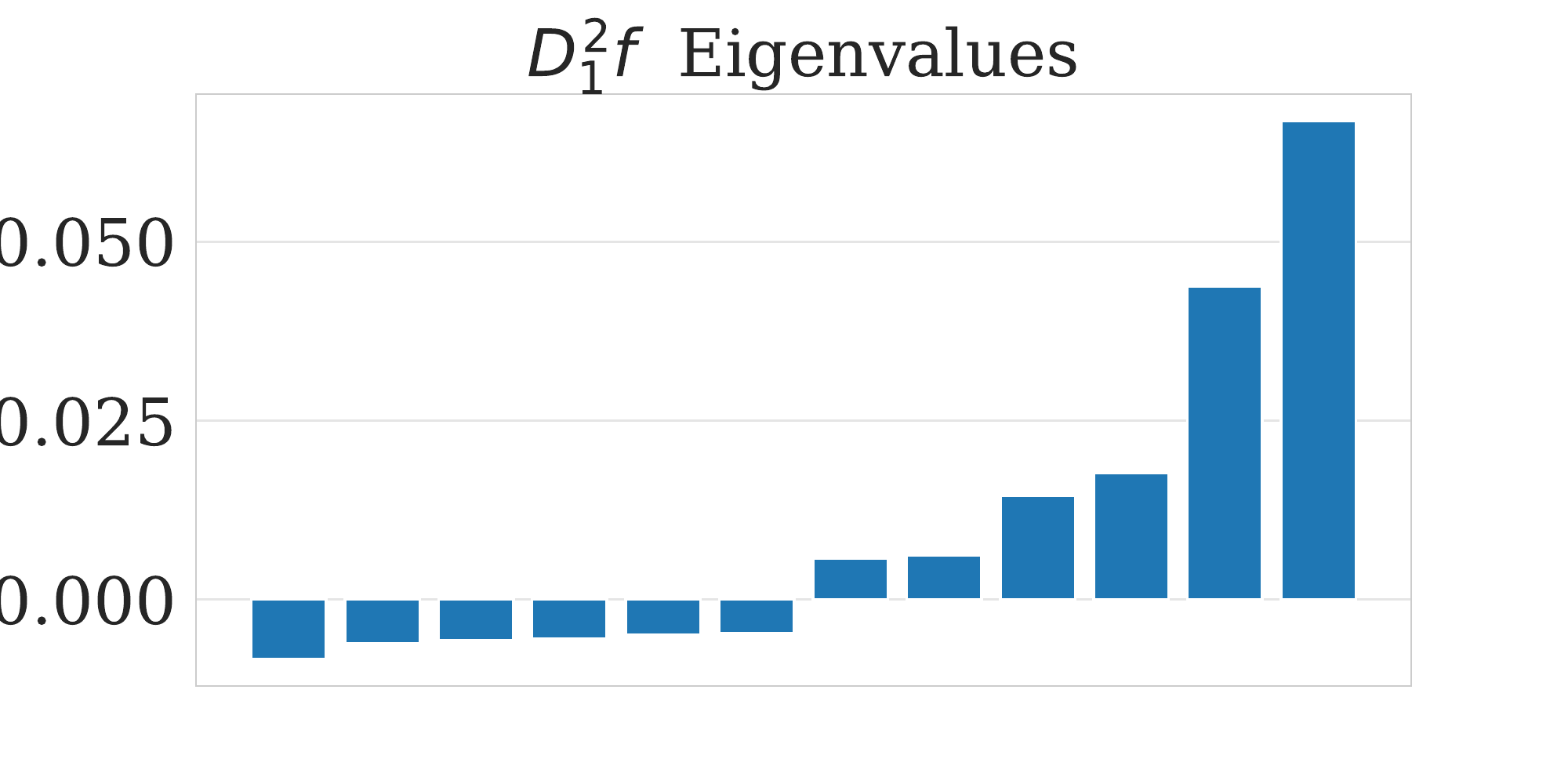}\label{fig:mog1_eig3sg}}
  \subfloat[][$D_2^2f_2$]{\includegraphics[trim={0 0 0 8ex},clip,height=.085\textwidth]{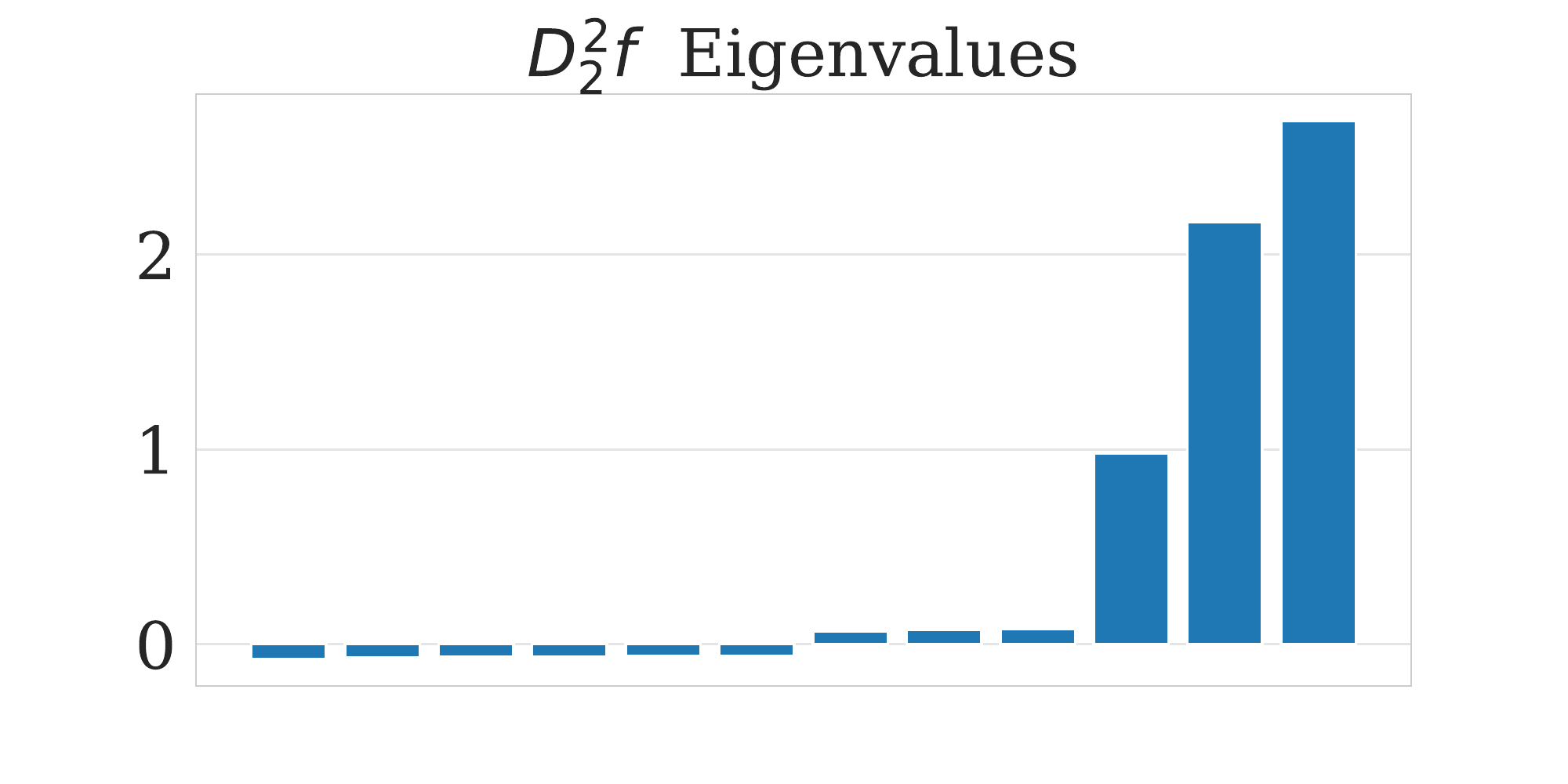}\label{fig:mog1_eig4sg}}
  \vspace{-2ex}
  
  \subfloat[][Gen.]{\includegraphics[trim={0 -2.5ex 0 -3ex},clip,height=.10\textwidth]{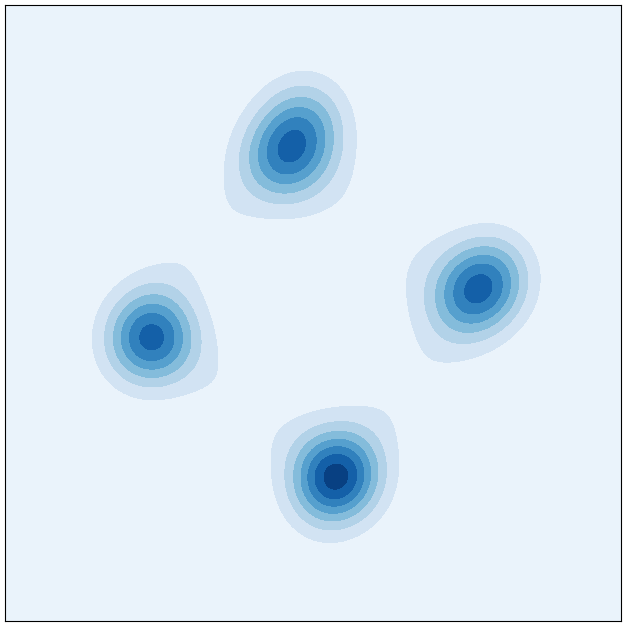}\label{fig:mog_13}} \hfill
  \subfloat[][Dis.]{\includegraphics[height=.10\textwidth]{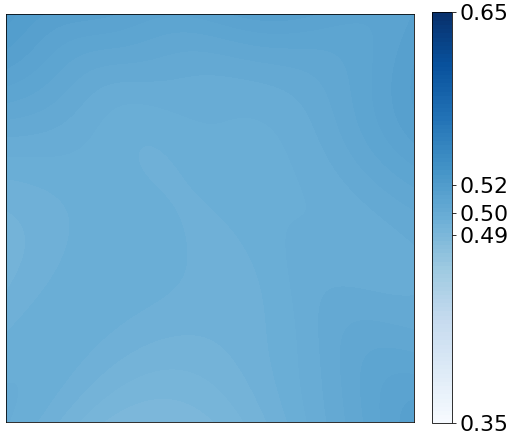}\label{fig:mog_14}} 
  \subfloat[][$J$]{\includegraphics[trim={0 0 0 6.5ex},clip,height=.085\textwidth]{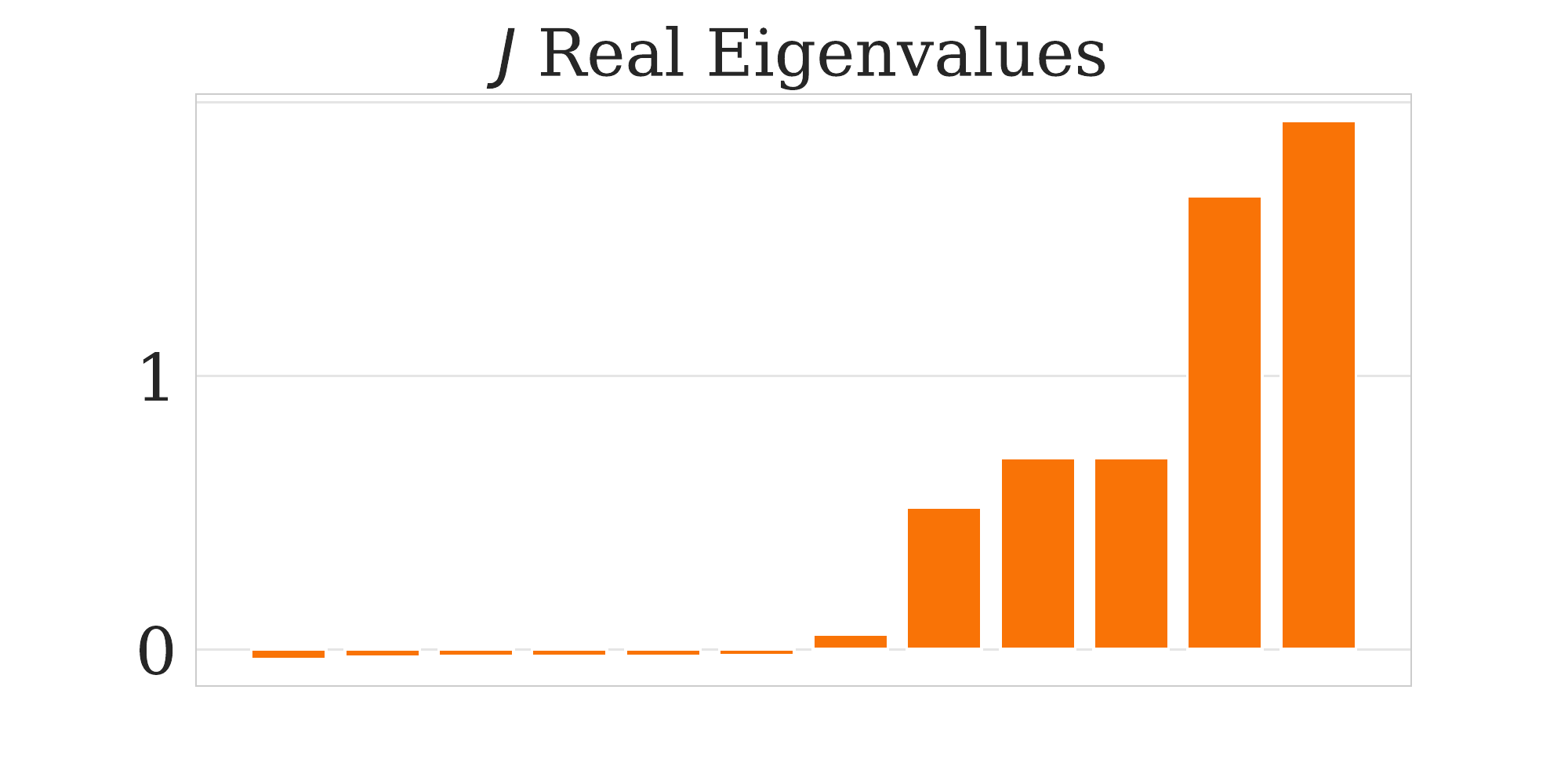}\label{fig:mog1_eig1st}} 
  \subfloat[][$S_1$]{\includegraphics[trim={0 0 0 8ex},clip,height=.085\textwidth]{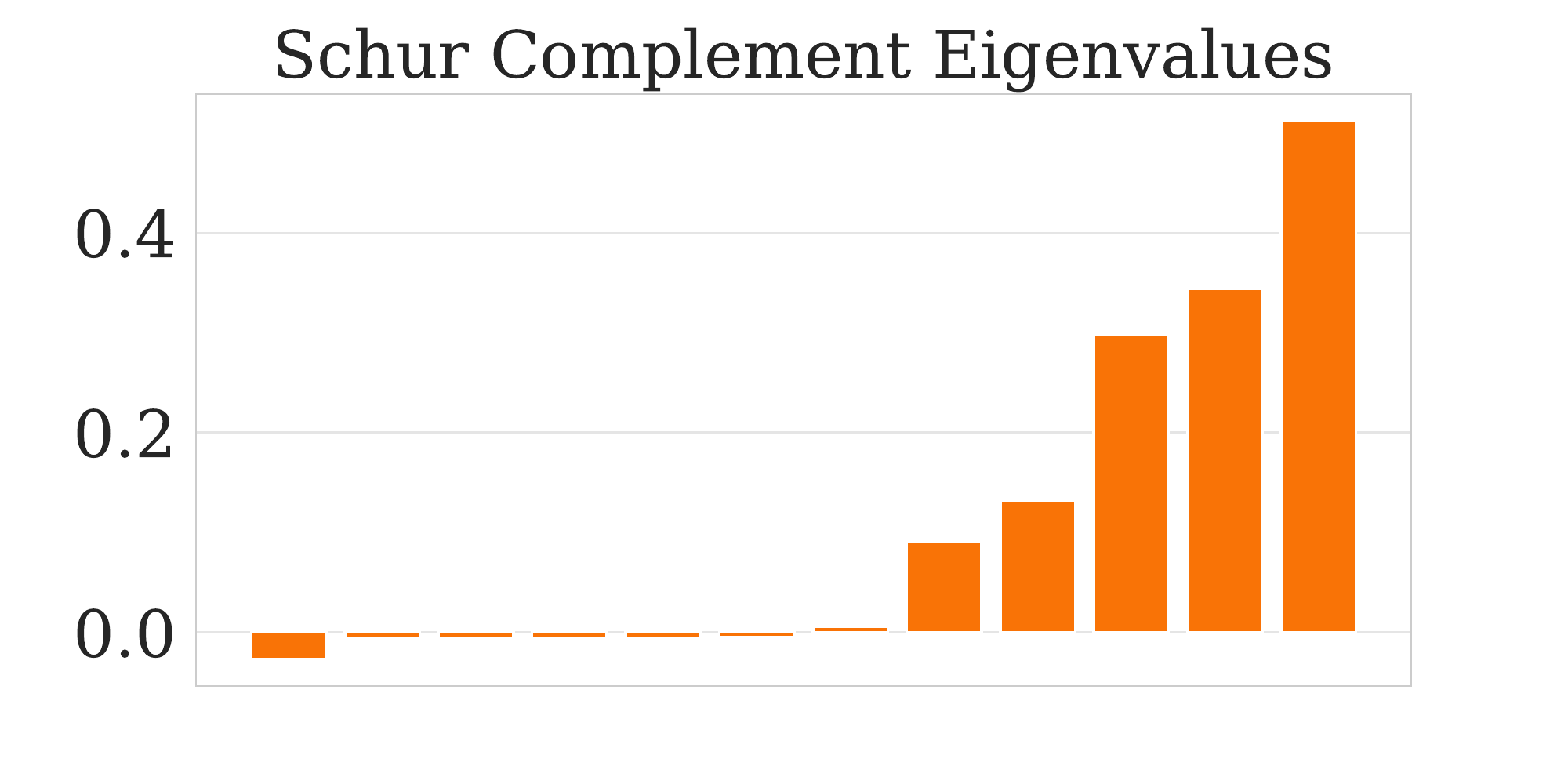}\label{fig:mog1_eig2st}} 
  \subfloat[][$D_1^2f_1$]{\includegraphics[trim={0 0 0 8ex},clip,height=.085\textwidth]{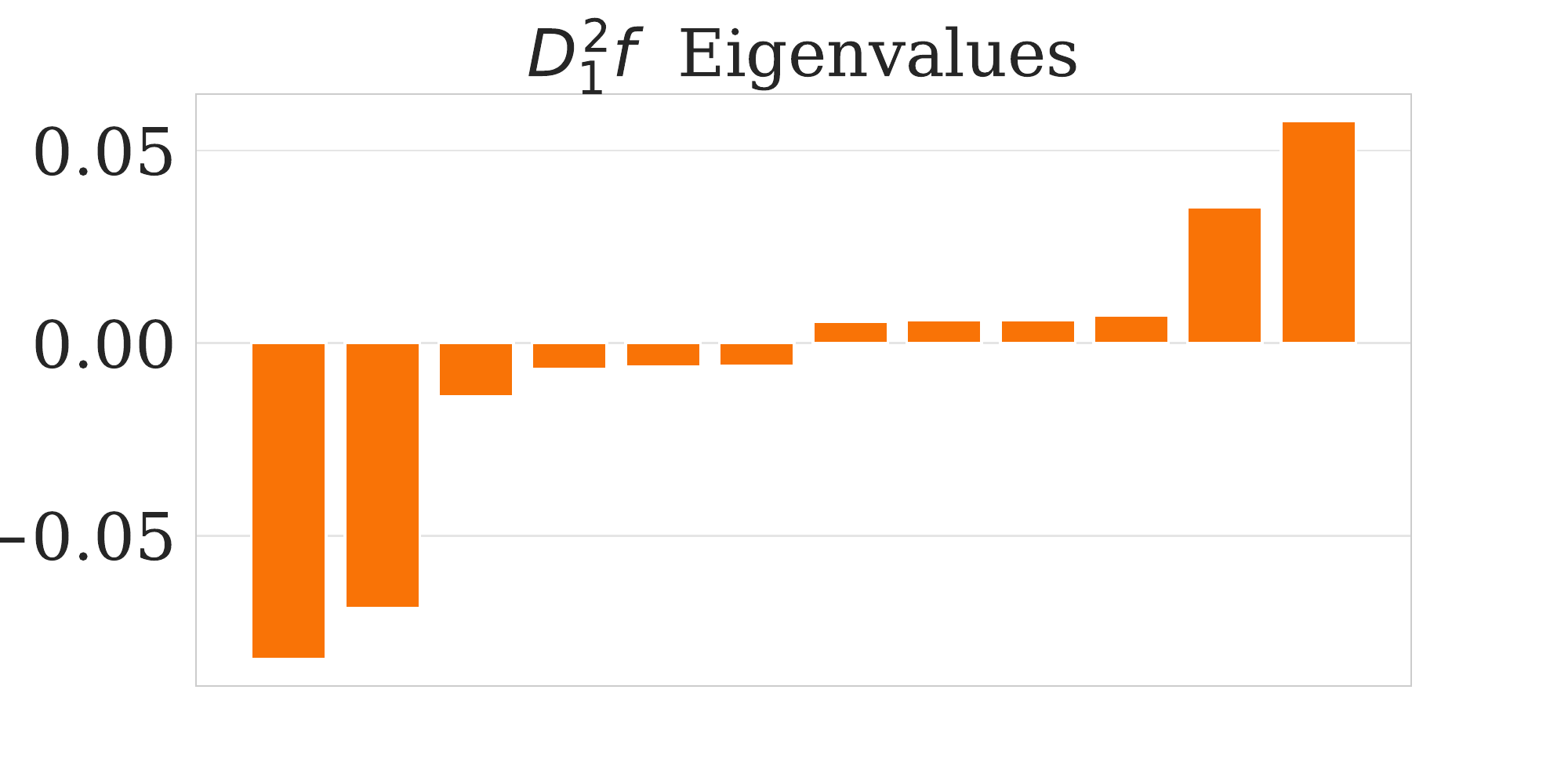}\label{fig:mog1_eig3st}}
  \subfloat[][$D_2^2f_2$]{\includegraphics[trim={0 0 0 8ex},clip,height=.085\textwidth]{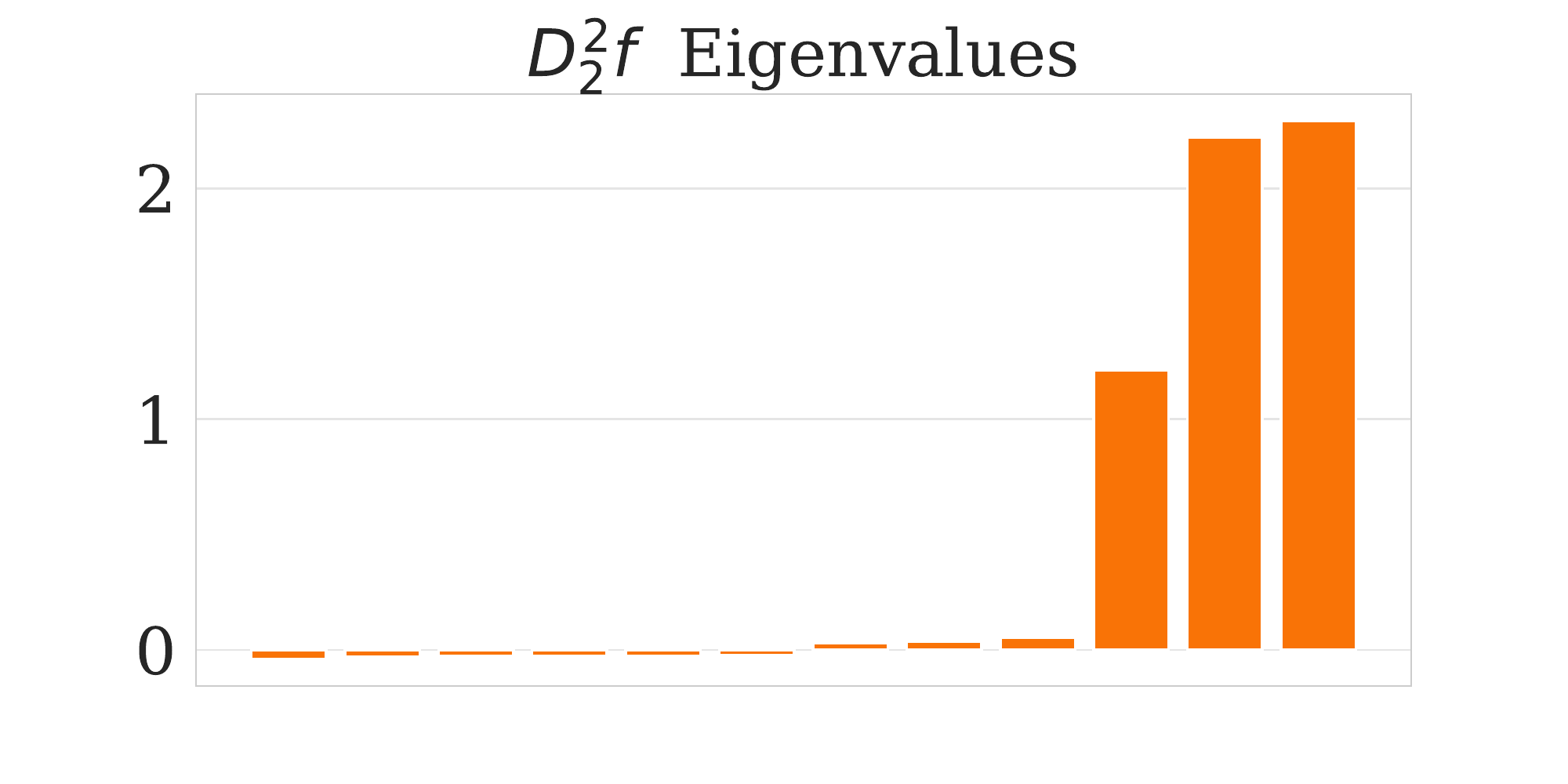}\label{fig:mog1_eig4st}}
  \caption{Convergence to non-Nash Stackelberg equilibria for both simultaneous gradient descent (top row) and Stackelberg learning dynamics (bottom row) in a 2-dimensional mixture of gaussian GAN example. The performance of the generator (player 1) and discriminator (player 2) are plotted in (a)--(b) and (g)--(h).
  To determine the positive definiteness of the game Jacobian, Schur complement and the individual Hessians, we compute the six smallest real eigenvalues and six largest real eigenvalues for each in (c)-(f) and (i)-(l). We observe that for both updates, the leader's Hessian is non-positive while the Schur complement is positive.
  }
  \label{fig:mog_1}
  \end{figure*}

We compare the deterministic gradient update for Stackelberg learning dynamics and simultaneous gradient descent, and analyze the distance from equilibrium as a function of time. We plot $\|\Sigma - VV^\top\|_2$ for the generator's performance and $\|\frac{1}{2}(W+W^\top)\|_2$ for the discriminator's performance in Fig.~\ref{fig:cov} for varying dimensions $m$ with learning rate where $\gamma_{1,k} = o(\gamma_{2,k})$ and fixed regularization terms $\eta = m/5$.
We observe that Stackelberg learning converges to an equilibrium in fewer iterations than simultaneous gradient descent. For zero-sum games, our theory provides reasoning for this behavior since at any critical point the eigenvalues of the game Jacobian are purely real. This is in contrast to the game Jacobian for the simultaneous gradient descent, which can admit imaginary eigenvalue components that are know to cause rotational forces in the dynamics. This example provides empirical evidence that the Stackelberg dynamics cancel out rotations in general-sum games.

\textbf{Example 2: Mixture of Gaussian (Diamond).}
We also 
train 
 a GAN to learn a mixture of Gaussian distributions, where
the generator is the leader and the discriminator is the follower. The generator network has two hidden layers and the discriminator has one hidden layer; each hidden layer has $32$ neurons. We train using a batch size of $256$, a latent dimension of $16$, and the default ADAM optimizer configuration in PyTorch version 1. Since the updates are stochastic, we decay the learning rates to satisfy our timescale separation assumption and regularize the implicit map of the follower using the parameter $\eta= 1$. We derive the regularized leader update in Appendix~\ref{app:regfollower}.

The underlying data distribution for this problem consists of Gaussian distributions with means given by $\mu = [1.5 \sin(\omega), 1.5\cos(\omega)]$ for $\omega \in \{k\pi/2\}_{k=0}^{3}$ and each with covariance $\sigma^2 I$ where $\sigma^2=0.15$. Each sample of real data given to the discriminator is selected uniformly at random from the set of Gaussian distributions. We train each learning rule using learning rates that begin at $0.0001$. Moreover, in this example, the activation following the hidden layers in each network is the tanh function. 

We train this experiment using the saturating GAN objective~\cite{goodfellow2014generative}. In Fig.~\ref{fig:mog_11}--\ref{fig:mog_12} and Fig.~\ref{fig:mog_13}--\ref{fig:mog_14} we show a sample of the generator and the discriminator for simultaneous gradient descent and the Stackelberg dynamics after 40,000 training batches. Each learning rule converges so that the generator can create a distribution that is close to the ground truth and the discriminator is nearly at the optimal probability throughout the input space. In Fig.~\ref{fig:mog1_eig1sg}--\ref{fig:mog1_eig4sg} and Fig.~\ref{fig:mog1_eig1st}--\ref{fig:mog1_eig4st}, we show eigenvalues from the game that allow us to get a deeper view of the convergence behavior. 
We observe that the simultaneous gradient dynamics appear be in a neighborhood of a non-Nash equilibrium since the individual Hessian for the leader is indefinite, the individual Hessian for the follower is positive definite, and the Schur complement is positive definite. Moreover, the eigenvalues of the leader individual Hessian are nearly zero, which would reflect the realizable assumption from Section 2. The Stackelberg learning dynamics converge to a point with similar eigenvalues, which would be a non-Nash Stackelberg equilibrium. This example demonstrates that standard GAN training can converge to non-Nash attractors that are Stackelberg equilibria and the Stackelberg equilibria can produce good generator and discriminator performance. This indicates that it may not be necessary to look only for Nash equilibria and instead it may be easier to find Stackelberg equilibria and the performance could be as desirable.

\begin{figure*}[t!]
  \subfloat[][Real]{\includegraphics[width=.10\textwidth]{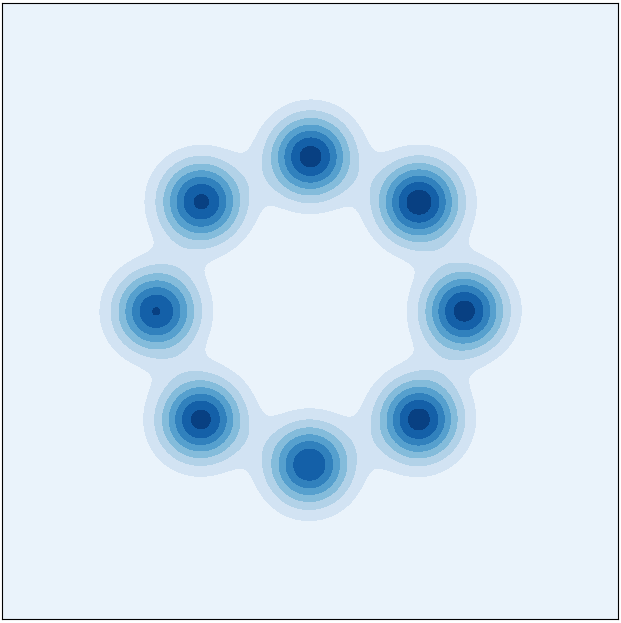}\label{fig:circle}} \hfill
  \subfloat[][8k]{\includegraphics[width=.10\textwidth]{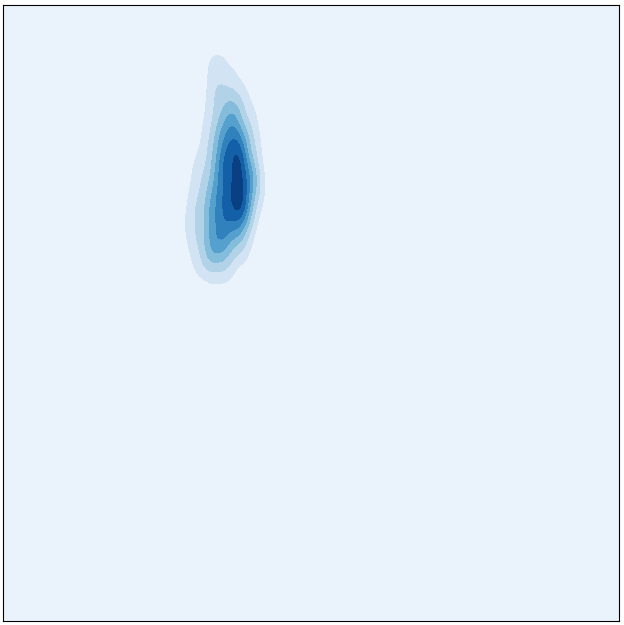}\label{fig:circle}} 
  \subfloat[][20k]{\includegraphics[width=.10\textwidth]{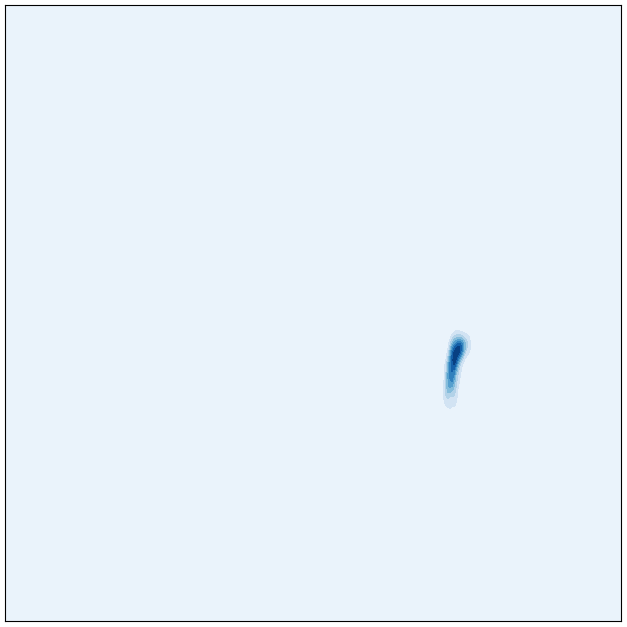}\label{fig:circle}} 
  \subfloat[][40k]{\includegraphics[width=.10\textwidth]{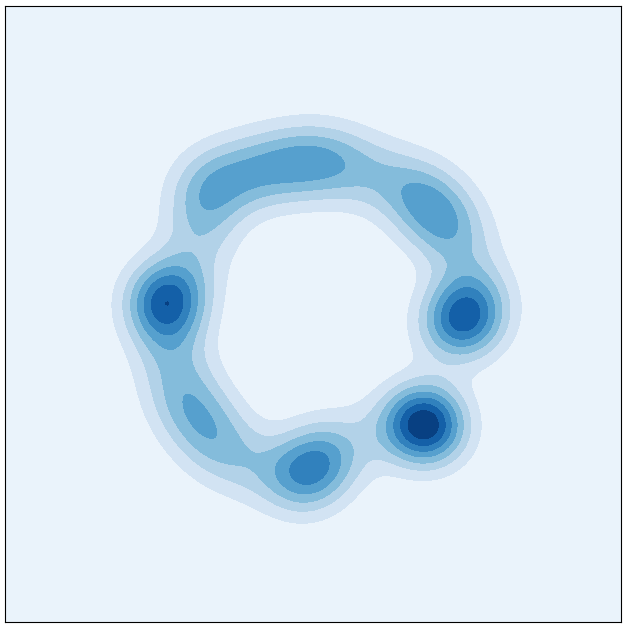}\label{fig:circle}} 
  \subfloat[][60k]{\includegraphics[width=.10\textwidth]{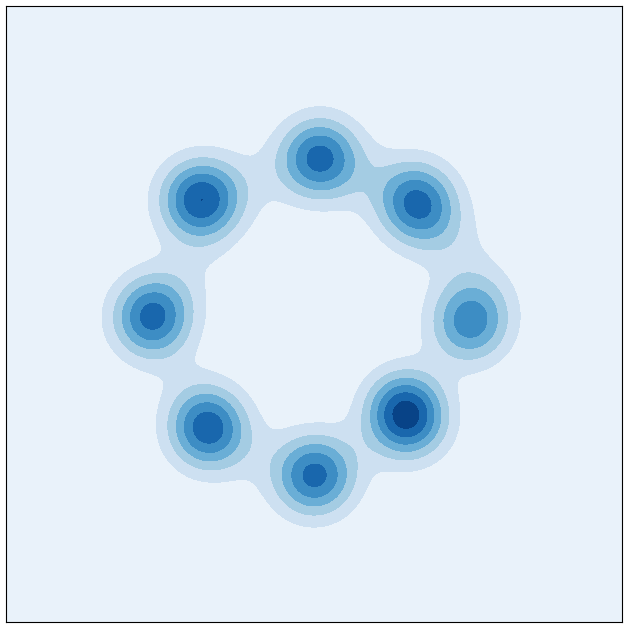}\label{fig:circle}} \hfill
  \subfloat[][8k]{\includegraphics[width=.10\textwidth]{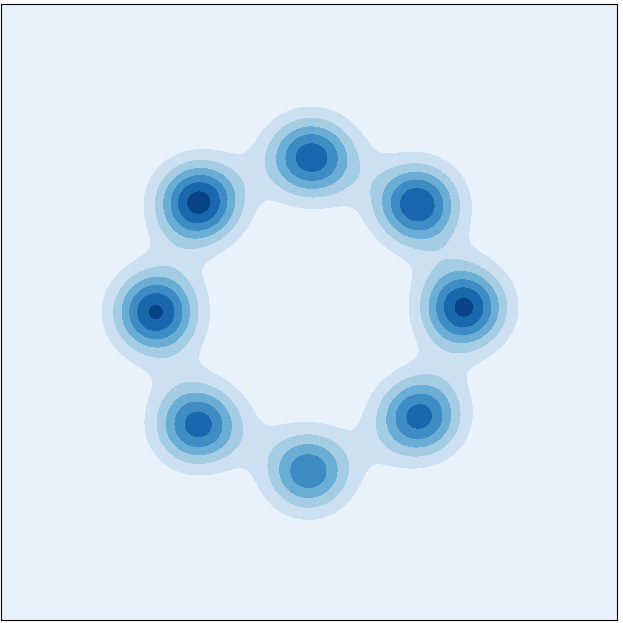}\label{fig:circle}} 
  \subfloat[][20k]{\includegraphics[width=.10\textwidth]{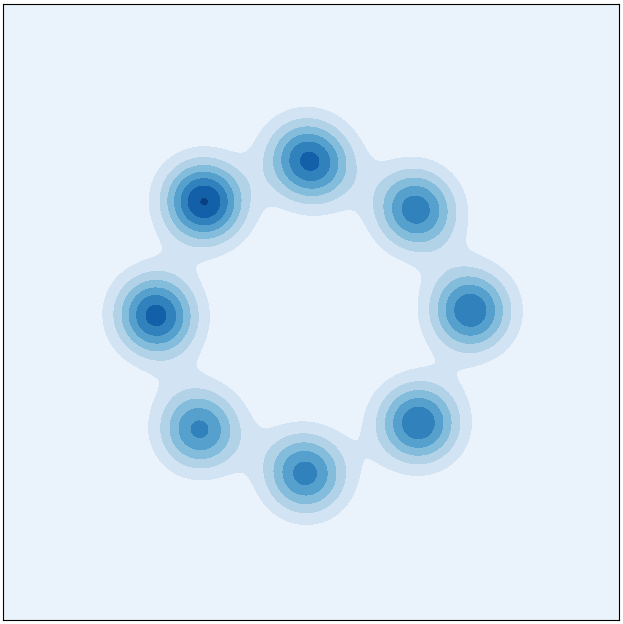}\label{fig:circle}} 
  \subfloat[][40k]{\includegraphics[width=.10\textwidth]{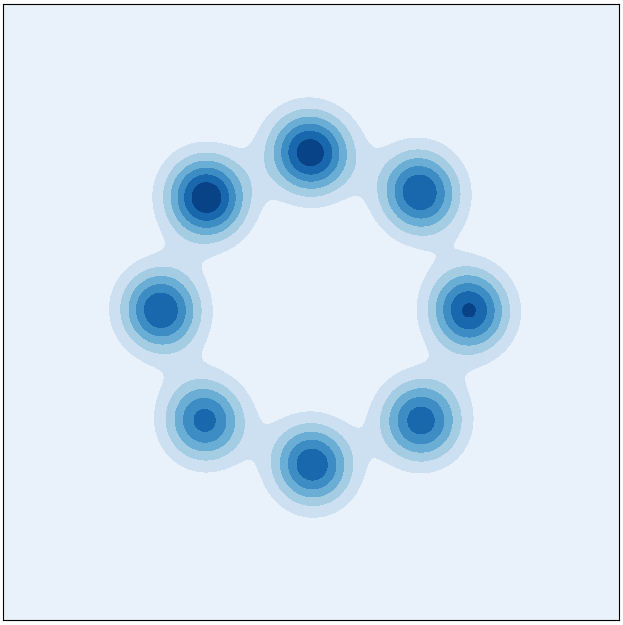}\label{fig:circle}} 
  \subfloat[][60k]{\includegraphics[width=.10\textwidth]{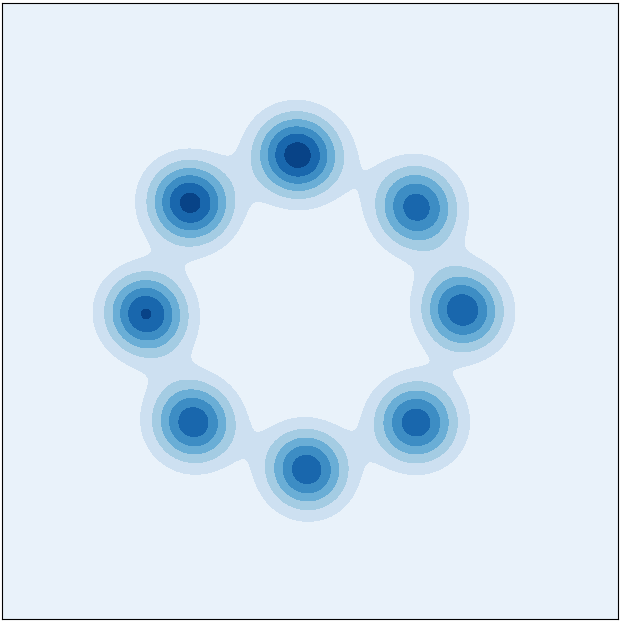}\label{fig:circle}} 
  \vspace{-2ex}

  \subfloat[][$J$]{\includegraphics[trim={0 0 0 8ex},clip,width=.2\textwidth]{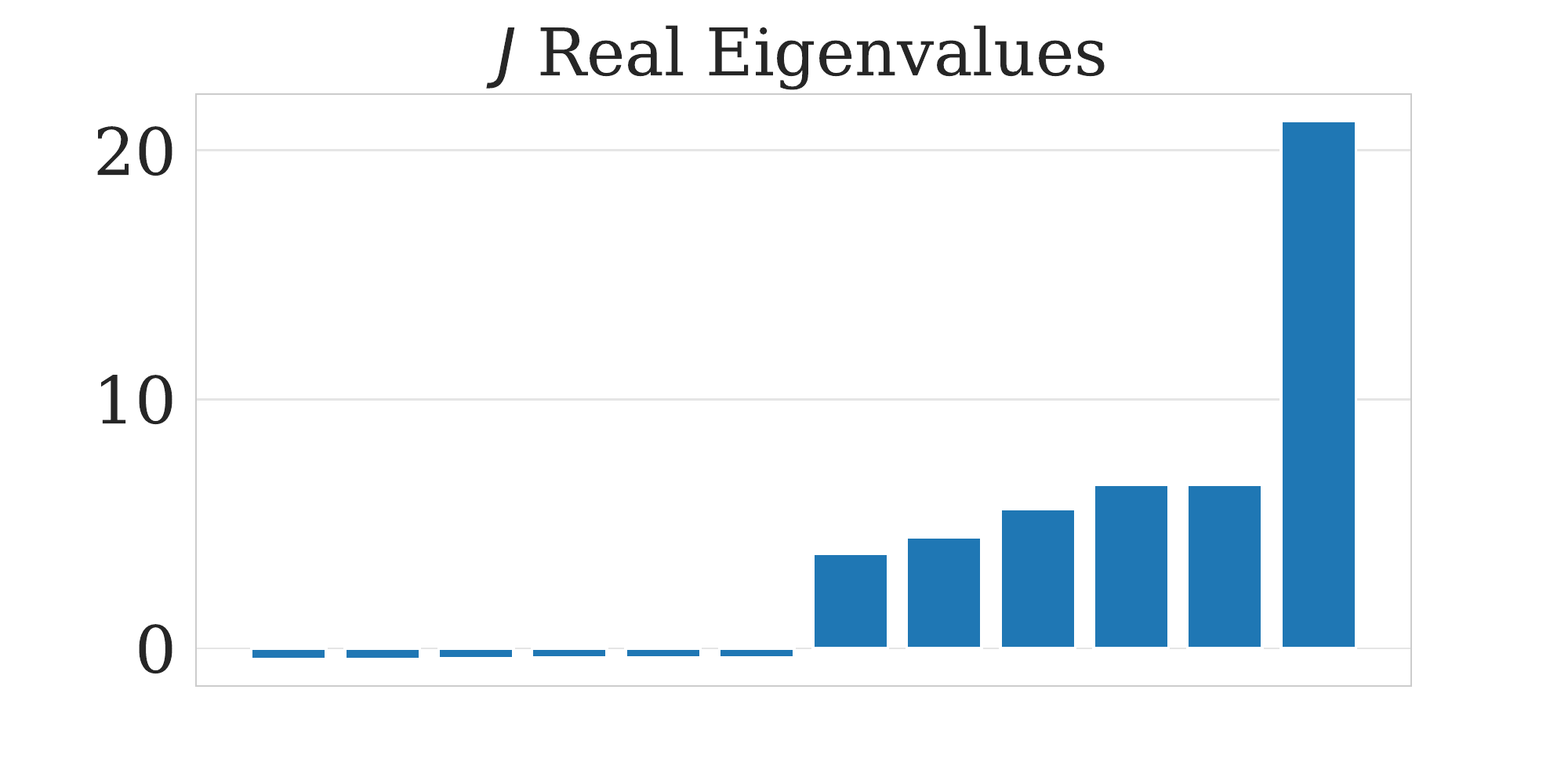}\label{fig:circle}} \hfill
  \subfloat[][$S_1$]{\includegraphics[trim={0 0 0 8ex},clip,width=.2\textwidth]{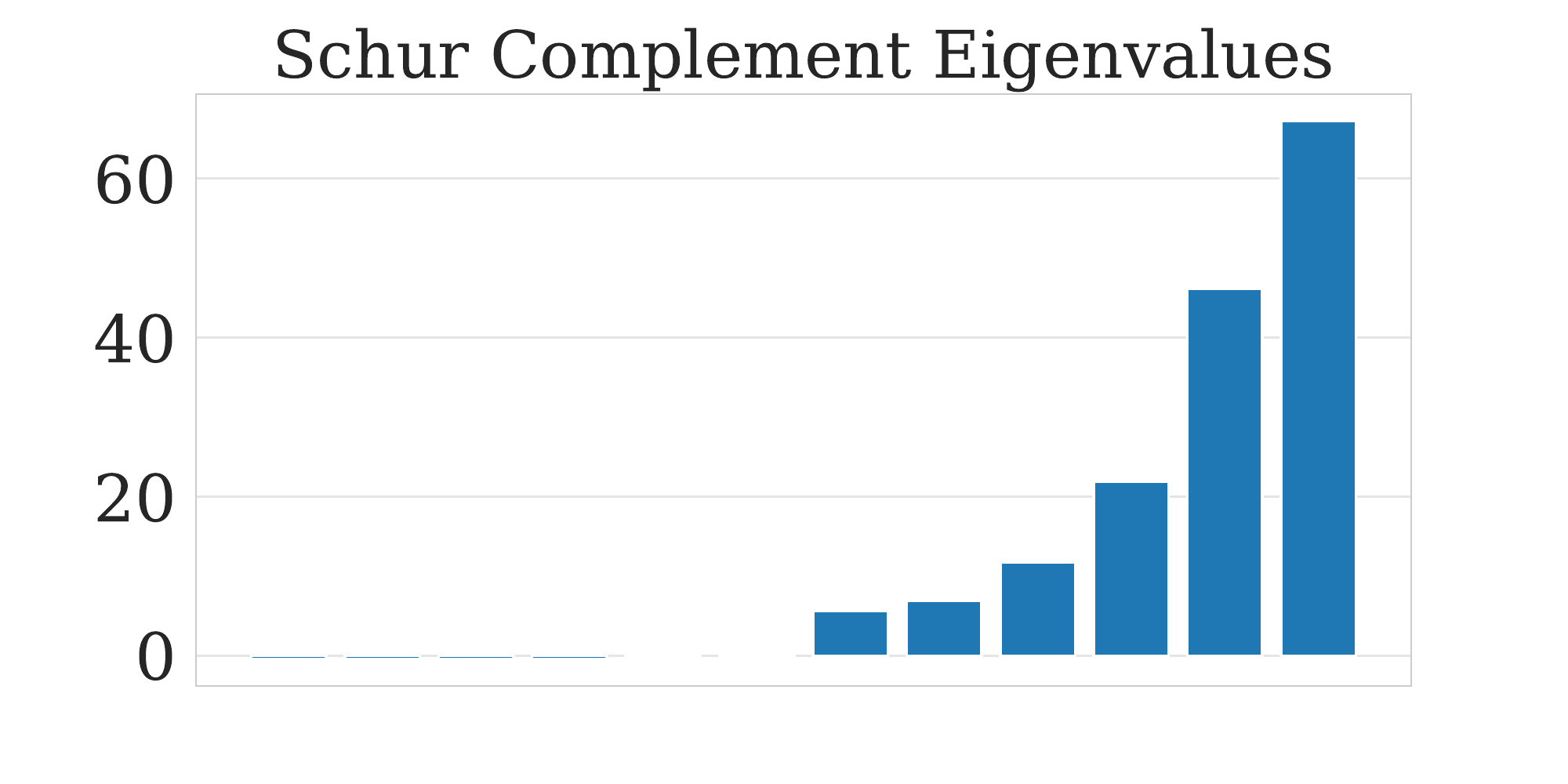}\label{fig:circle}} \hfill
  \subfloat[][$D_1^2f_1$]{\includegraphics[trim={0 0 0 8ex},clip,width=.2\textwidth]{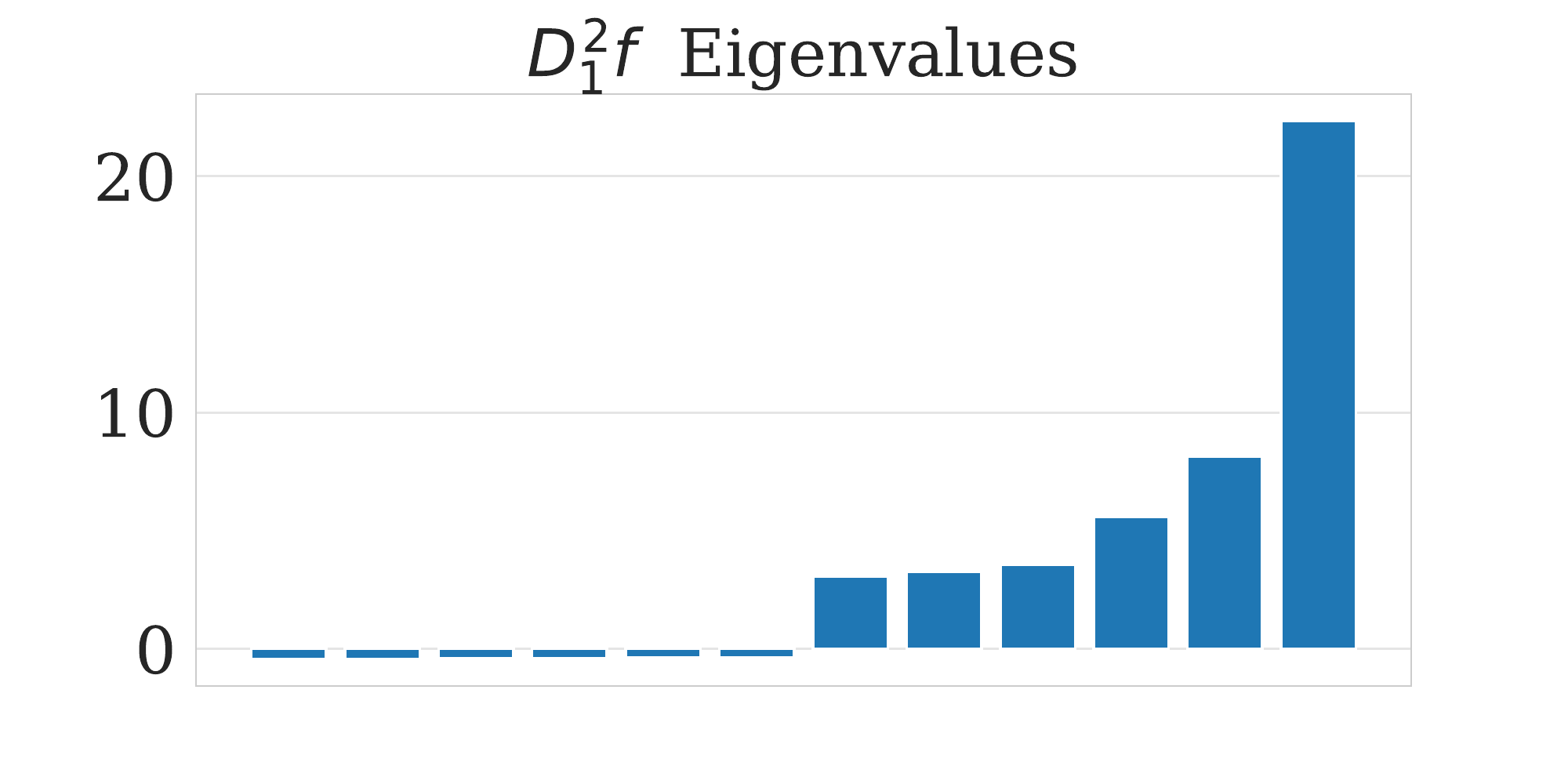}\label{fig:circle}} \hfill
  \subfloat[][$D_2^2f_2$]{\includegraphics[trim={0 0 0 6.5ex},clip,width=.2\textwidth]{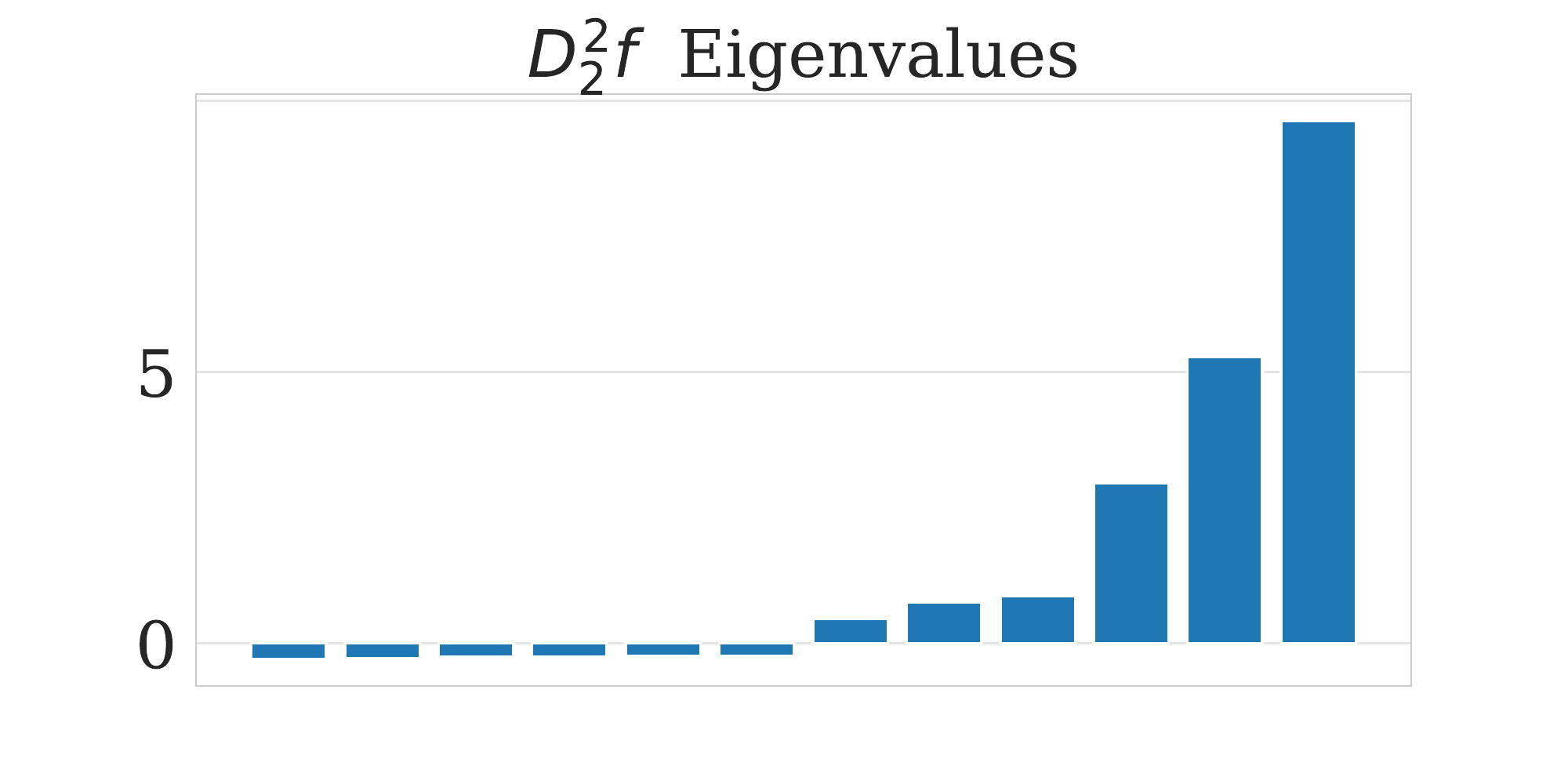}\label{fig:circle}}
  \vspace{-2ex}
  
  \subfloat[][$J$]{\includegraphics[trim={0 0 0 8ex},clip,width=.2\textwidth]{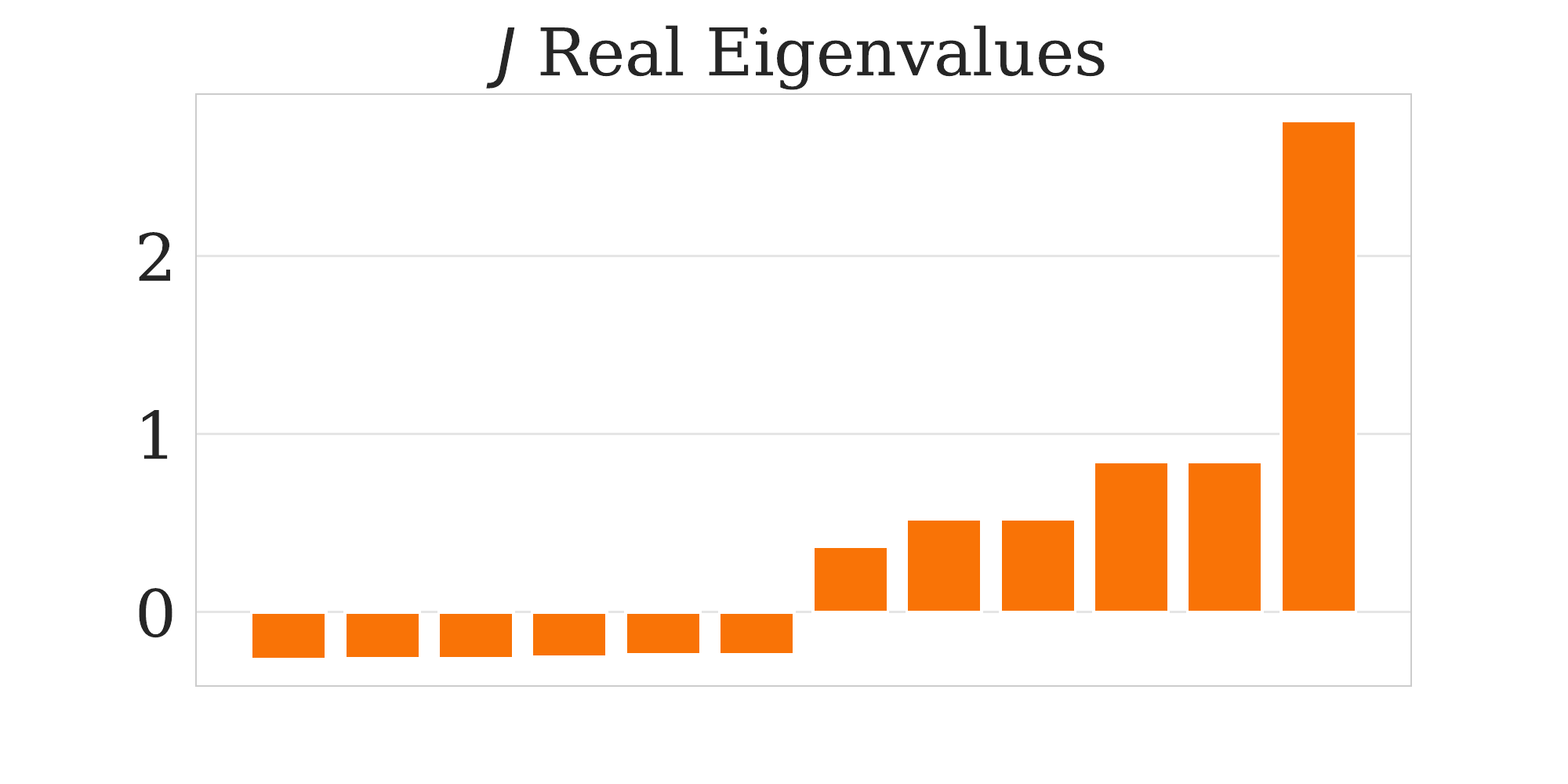}\label{fig:circle}} \hfill
  \subfloat[][$S_1$]{\includegraphics[trim={0 0 0 8ex},clip,width=.2\textwidth]{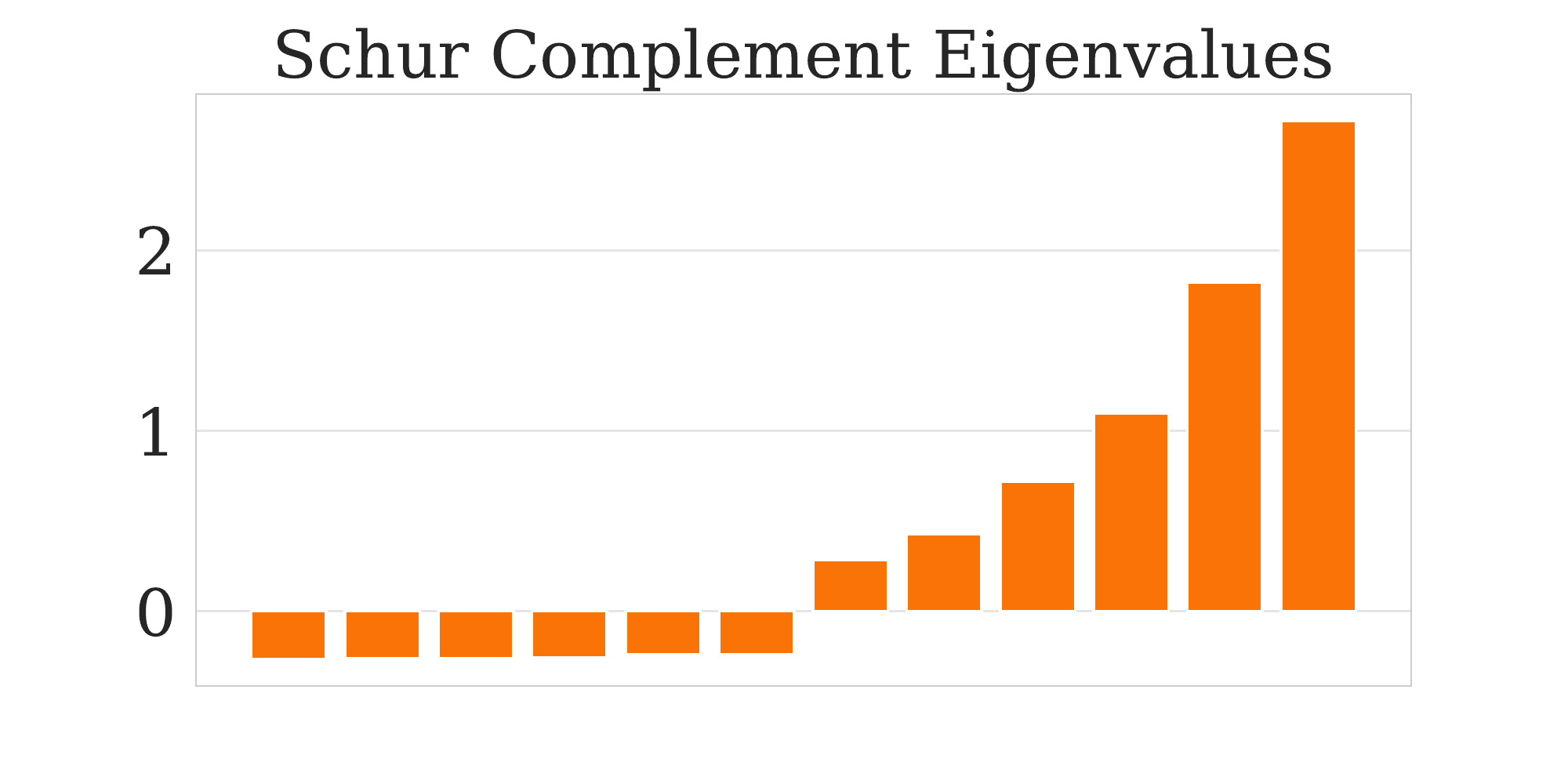}\label{fig:circle}} \hfill
  \subfloat[][$D_1^2f_1$]{\includegraphics[trim={0 0 0 8ex},clip,width=.2\textwidth]{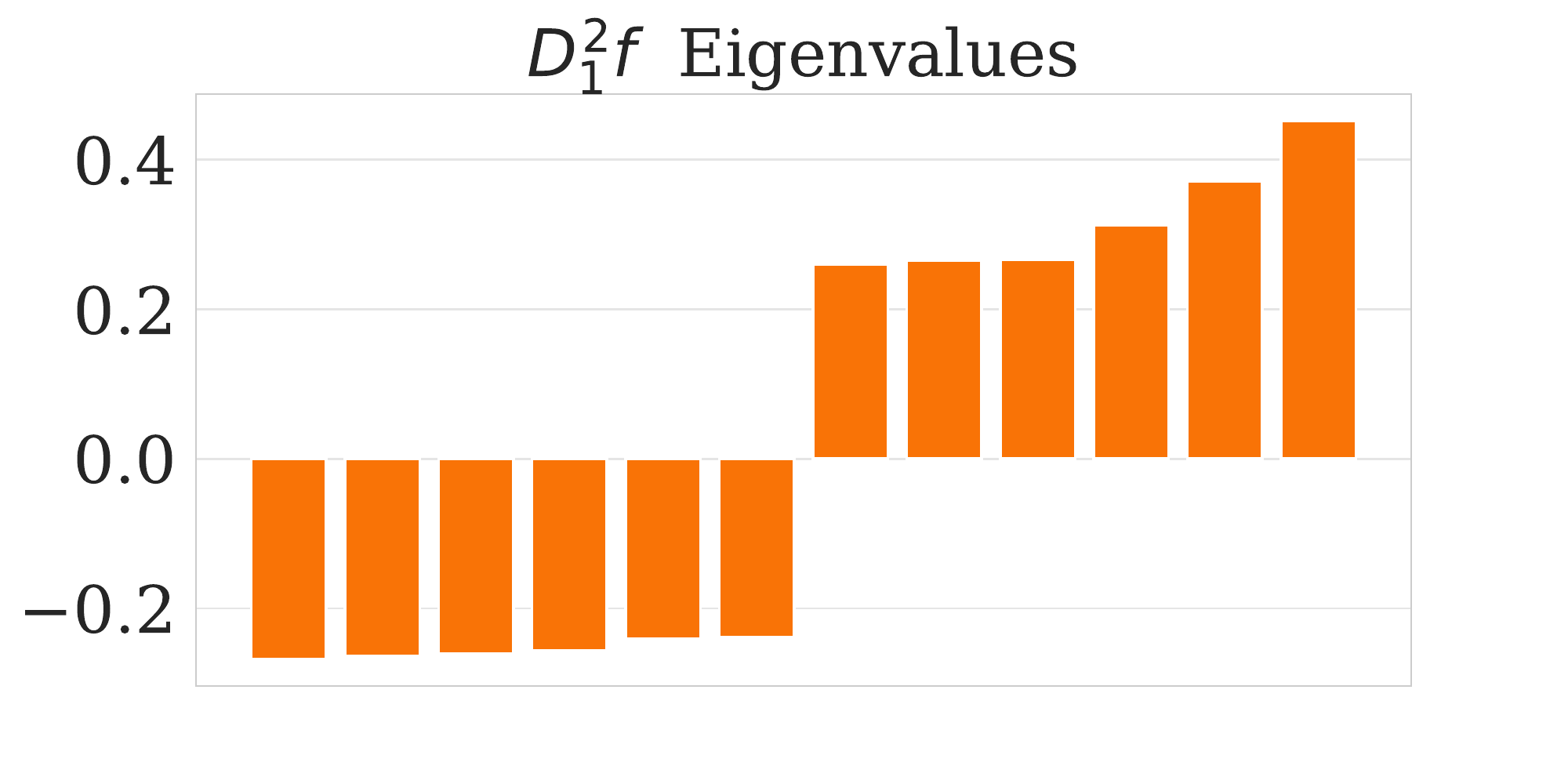}\label{fig:circle}} \hfill
  \subfloat[][$D_2^2f_2$]{\includegraphics[trim={0 0 0 8ex},clip,width=.2\textwidth]{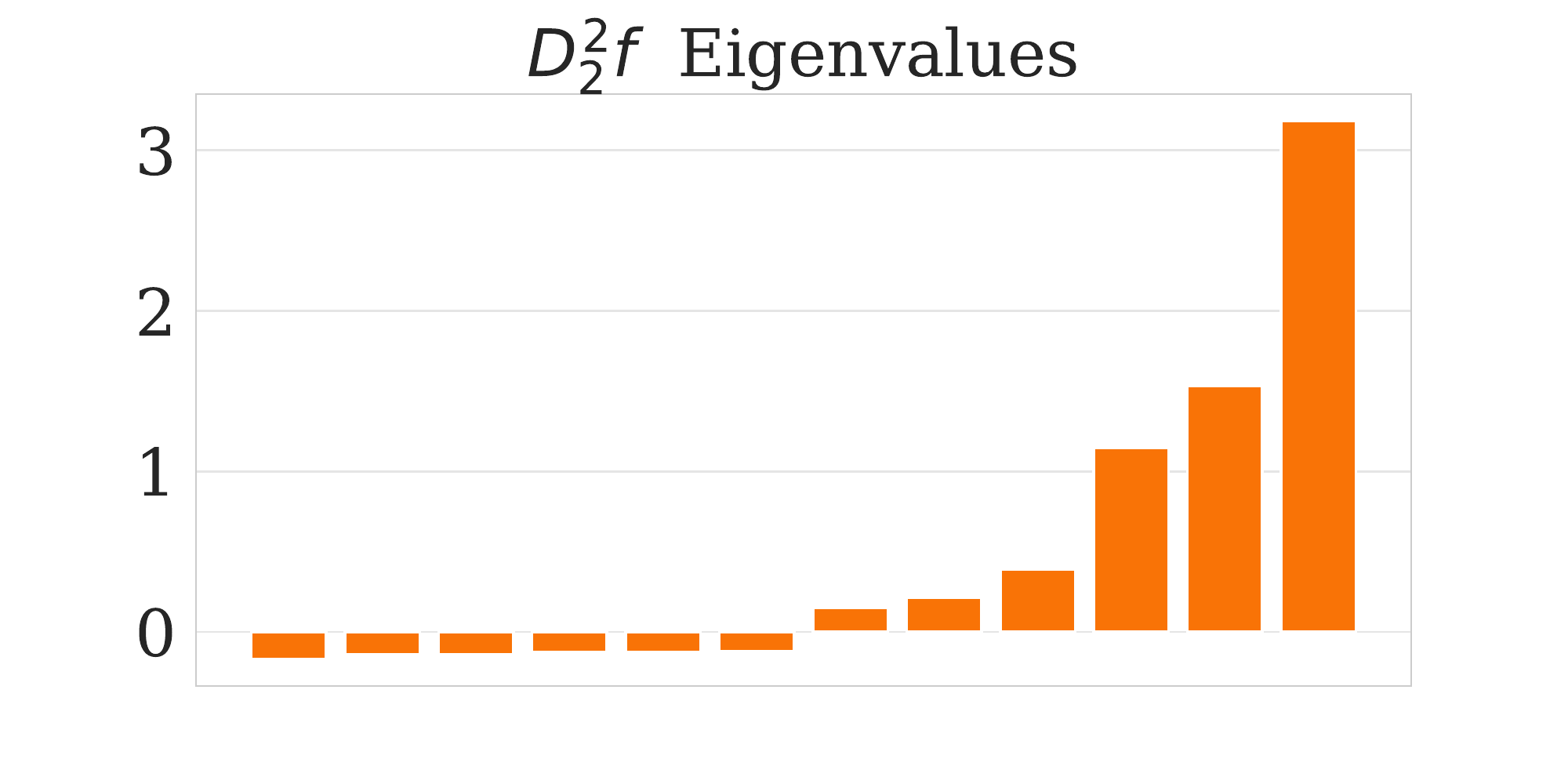}\label{fig:circle}}
  \caption{
      Convergence to Nash for simultaneous gradient descent in Fig. (b)--(e) and convergence to non-Nash Stackelberg for Stackelberg learning in Fig. (f)--(i) for the mixture of gaussian example. We plot the smallest six and largest six eigenvalues of the game Jacobian, Schur complement and individual hessians in (j)--(m) for simultaneous gradient descent and in (n)--(q) for Stackelberg learning at iteration 60k.
      The eigenvalues in this example seem to indicate that simultaneous gradient descent converged to a Nash equilibrium and that the Stackelberg learning dynamics converged to a non-Nash Stackelberg equilibria.}
  \label{fig:mog2}
  \end{figure*}

\textbf{Example 3: Mixure of Gaussian (Circle).}
The underlying data distribution for this problem consists of Gaussian distributions with means given by $\mu = [\sin(\omega), \cos(\omega)]$ for $\omega \in \{k\pi/4\}_{k=0}^7$ and each with covariance $\sigma^2 I$ where $\sigma^2=0.3$, sampled in the similar manner as the previous example. We train each learning rule using learning rates that begin at $0.0004$. Moreover, in this example, the activation following the hidden layers in each network is the ReLU function.  

We train the GAN with the non-saturating
objective~\cite{goodfellow2014generative}. We show the the performance in
Fig.~\ref{fig:mog2} along the learning path for the simultaneous gradient
descent dynamics and the Stackelberg learning dynamics. The simultaneous
gradient descent dynamics cycle and perform poorly until the learning rates have
decayed enough to stabilize the training process. The Stackelberg learning
dynamics converge quickly to a solution that nearly matches the
ground truth distribution. In a similar fashion as in the covariance example,
the leader update is able to cancel out rotations and converge to a desirable
solution with a learning rate that destabilizes the training process for
standard training techniques. We show the eigenvalues after training and see
that for this configuration the simultaneous gradient dynamics converge to a
Nash equilibrium and the Stackelberg learning dynamics converge again to a
non-Nash Stackelberg equilibrium. This provides further evidence that Stackelberg equilibria may be easier to reach and can provide suitable generator performance.

\textbf{Example 4: MNIST dataset.}
To demonstrate that the Stackelberg learning dynamics can scale to high dimensional problems, we train a GAN on the MNIST dataset using the DCGAN architecture adapted to handle $28\times 28$ images. We train on an MNIST dataset consisting of only the digits 0 and 1 from the training images and on an MNIST dataset containing the entire set of training images. We train using a batch size of $256$, a latent dimension of $100$, and the ADAM optimizer with the default parameters for the DCGAN network.  
We regularize the implicit map of the follower as detailed in Appendix~\ref{app:regfollower} using the parameter $\eta= 5000$. If we view the regularization as a linear function of the number of parameters in the discriminator, then this selection of regularization is nearly equal to that from the mixture of Gaussian experiments.

We show the results in Fig.~\ref{fig:mnist} after 2900 batches. For each dataset we show a sample of 16 digits to get a clear view of the generator performance and a sample of 256 digits to get a broader view of the generator output. The Stackelberg dynamics are able to converge to a solution that generates realistic handwritten digits. The primary purpose of this example is to show that the learning dynamics including second order information and an inverse is not an insurmountable problem for training large scale networks with millions of parameters. We believe the tools we develop for our implementation can be helpful to researchers working on GANs since a number of theoretical works on this topic require second order information to strengthen the convergence guarantees.

\begin{figure*}[t!]
  \centering
  \subfloat[][]{\includegraphics[width=.175\textwidth]{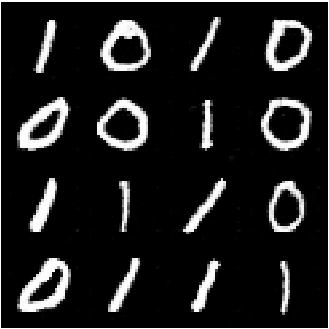}\label{fig:mnist_simple}}\hfill
  \subfloat[][]{\includegraphics[width=.175\textwidth]{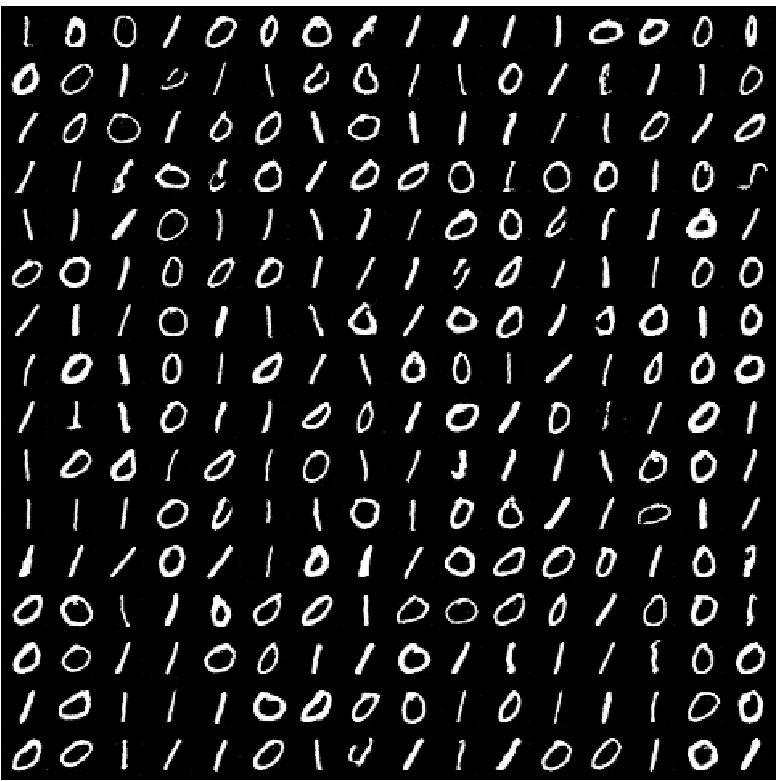}\label{fig:mnist_simple_big}}\hfill
  \subfloat[][]{\includegraphics[width=.175\textwidth]{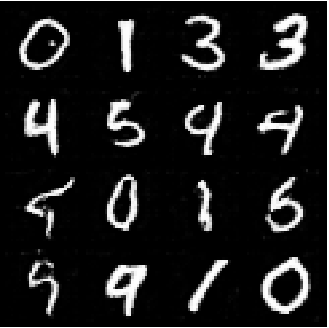}\label{fig:mnist_complicated}}\hfill
  \subfloat[][]{\includegraphics[width=.175\textwidth]{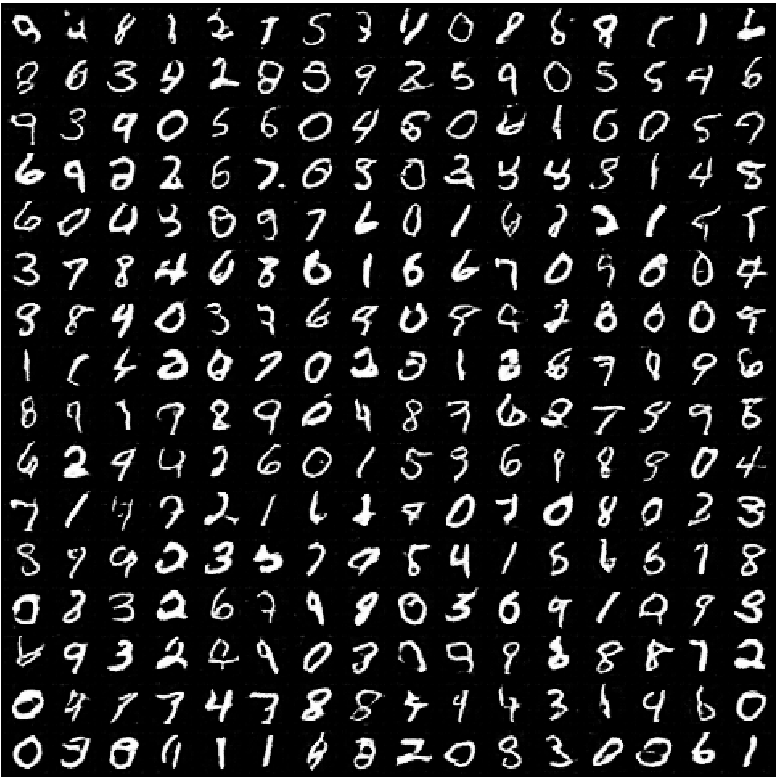}\label{fig:mnist_complicated_big}}
  \caption{We demonstrate Stackelberg learning on the MNIST dataset for digits for 0s and 1s in (a)-(b) and for all digits in (c)-(d).}
  \label{fig:mnist}
\end{figure*}

\section{Conclusion}
\label{sec:discussion}
We study the convergence of learning dynamics in Stackelberg games. This class of games broadly pertains to any application in which there is an order of play between the players in the game. However, the problem has not been extensively analyzed in the way the learning dynamics of simultaneous play games have been. Consequently, we are able to give novel convergence results and draw connections to existing work focused on learning Nash equilibria.

\bibliographystyle{plainnat}
\bibliography{refs}

\appendix
\onecolumn

\section{Mathematical Preliminaries}
\label{app:prelims}
In this appendix, we show some preliminary results on linear algebra and recall 
 some definitions and results from dynamical systems
theory that are needed to state and prove the results in the main paper.
\subsection{Proofs of Propositions~\ref{prop:nec} and
\ref{prop:suf}}
\label{app:linalg}
The results in this subsection follow from the theory of block operator matrices
and indefinite linear algebra~\cite{tretter:2008aa}.

The following lemma is a very well-known result in linear algebra and can be
found in nearly any advanced linear algebra text such as~\cite{horn:2011aa}.
\begin{lemma}
Let $W\in \mb{C}^{n\times n}$ be Hermitian with $k$ positive eigenvalues
(counted with multiplicities) and let $U\in \mb{C}^{m\times n}$. Then
\[\lambda_j(UWU^\ast)\leq \|U\|^2\lambda_j(W)\]
for $j=1, \ldots, \min\{k,m,\rank(UWU^\ast)\}$.
\label{lem:hermorder}
\end{lemma}

Let us define $|M|=(MM^\trans)^{1/2}$ for a matrix $M$. 
Recall also that for Propositions~\ref{prop:nec} and \ref{prop:suf}, we have
 defined $\spec(D_1^2f(x^\ast))=\{\mu_j, \ j\in \{1, \ldots, m\}\}$
where \[\mu_1\leq \cdots \leq \mu_r<0\leq \mu_{r+1}\leq \cdots \leq \mu_m,\] and
 $\spec(-D_2^2f(x^\ast))=\{\lambda_i, \ i\in \{1, \ldots, n\}\}$ where
$\lambda_1\geq \cdots \geq \lambda_n>0$, given an attractor $x^\ast$.

We can now use the above Lemma to prove Proposition~\ref{prop:nec}. The proof
follows the main arguments in the proof of Lemma 3.2 in the work
by~\citet{berger:2018aa} with some minor changes due to the nature of our
problem. 
\begin{proof}[Proof of Proposition~\ref{prop:nec}]
    Let $x^\ast$ be a stable attractor of $\dot{x}=-\omegas(x)$ such that
    $-D_2^2f(x^\ast)>0$. For the sake of presentation, define
    $A=D_1^2f(x^\ast)$, $B=D_{12}f(x^\ast)$, and $C=D_2^2f(x^\ast)$. Recall that
    $x_1\in \mb{R}^n$ and $x_2\in \mb{R}^m$.
    Suppose that $A-BC^{-1}B^\trans>0$. 

   \noindent\textbf{Claim: $r\leq n$ is necessary.}  We argue by contradiction. Suppose
    not---i.e., assume that
$r>n$. Note that if $m<n$, then this is not possible. In this case, we
automatically satisfy that $r\leq n$. Otherwise, $r\geq m>n$. Let
$\mc{S}_1=\ker(B(-C^{-1}+|C^{-1}|)B^\trans)$ and consider the subspace $\mc{S}_2$ of $\mb{C}^m$ spanned by the all the eigenvectors of $A$ corresponding to non-positive eigenvalues. Note that
\[\dim \mc{S}_1=m-\rank(B(-C^{-1}+|C^{-1}|)B^\trans)\geq m-\rank(-C^{-1}+|C^{-1}|)=m-n\]
By assumption, we have that $\dim \mc{S}_2=r$ so that, since $r>n$,  
\[\dim \mc{S}_1+\dim\mc{S}_2\geq (m-n)+r=m+(r-n)>m.\]
 Thus, 
$\mc{S}_1\cap \mc{S}_2\neq \{0\}$. Now, 
$\mc{S}_1=\ker(B(-C^{-1}+|C^{-1}|)B^\trans)$. Hence, 
for any non-trivial vector $v\in \mc{S}_1\cap \mc{S}_2$, 
$(BC^{-1}B^\trans-B|C^{-1}|B^\trans) v=0$ so that
we have 
\begin{equation}\langle (A-BC^{-1}B^\trans)v, v\rangle=\langle Av,v\rangle-\langle
    B|C^{-1}|B^\trans v,v\rangle\leq 0.
\label{eq:ineq}
\end{equation}
Note that the inequality in \eqref{eq:ineq} holds because the vector $v$ is in
the non-positive eigenspace of $A$ and the second term is clearly non-positive.
Thus, $A-BC^{-1}B^\trans$ cannot be positive definite, which gives a contradiction so that $r\leq n$.

\noindent\textbf{Claim: $\kappa^2\lambda_i+\mu_i>0$ is necessary}. Let the maps $\lambda_i(\cdot)$
denote the eigenvalues of its argument arranged in non-increasing order. Then,
by the  Weyl theorem for Hermitian matrices~\cite{horn:2011aa}, we have that
\[0<\lambda_m(A-BC^{-1}B^\trans)\leq
\lambda_i(A)+\lambda_{m-i+1}(-BC^{-1}B^\trans), \ i\in\{1, \ldots, m\}.\]
We can now combine this inequality with Lemma~\ref{lem:hermorder}. Indeed, we
have that 
\[0<\lambda_i(A)+\|B\|^2\lambda_{m-i+1}(-C^{-1})<\mu_{m-i+1}+\kappa^2\lambda_{m-i+1},
\ \ \forall \ i\in\{m-r+p+1, \ldots, m\}\]
which gives the desired result.

Since we have shown both the necessary conditions, this concludes the proof.
\end{proof}

Now, let us prove  Proposition~\ref{prop:suf} which gives sufficient conditions
for when a stable non-Nash attractor $x^\ast$ of $\dot{x}=-\omega(x)$ is a differential
Stackelberg equilibrium. Then, combining this with 
Proposition~\ref{prop:allstack}, we have a sufficient condition under which
stable non-Nash attractors are in fact stable attractors of
$\dot{x}=-\omegas(x)$.

\begin{proof}[Proof of Proposition~\ref{prop:suf}]
Let $x^\ast$ be a stable non-Nash attractor of $\dot{x}=-\omega(x)$ such
    that $D_1^2f(x^\ast)$ and $D_2^2f(x^\ast)>0$ are Hermitian.     
      Since $D_i^2f(x^\ast)$, $i=1,2$ are both Hermitian, let
      $D_1^2f(x^\ast)=\poneevecs M \poneevecs^\ast$ with
$\poneevecs \poneevecs^\ast=I_{n\times n}$ and $M=\diag(\mu_1, \ldots, \mu_m)$, and
$-D_2^2f(x^\ast)=\ptwoevecs\Lambda \ptwoevecs^\ast$ with
$\ptwoevecs\ptwoevecs^\ast=I_{m\times m}$ and $\Lambda=\diag(\lambda_1, \ldots, \lambda_n)$.

By assumption,    there exists a diagonal matrix $\Sigma\in \mb{R}^{m \times n}$ such that
$D_{12}f(x^\ast)=\poneevecs\Sigma \ptwoevecs^\ast$ where $\poneevecs$ are the orthonormal eigenvectors
    of $D_1^2f(x^\ast)$ and $\ptwoevecs$ are orthonormal eigenvectors of
    $-D_2^2f(x^\ast)$.
   Then, 
\begin{align*}
    D_1^2f(x^\ast)-D_{21}f(x^\ast)^\trans
(D_2^2f(x^\ast))^{-1}D_{21}f(x^\ast)&=\poneevecs M \poneevecs^\ast+\poneevecs
\Sigma \ptwoevecs^\ast
(\ptwoevecs\Lambda \ptwoevecs^\ast)^{-1}\ptwoevecs \Sigma^\ast \poneevecs^\ast\\
&=\poneevecs(M+\Sigma \Lambda^{-1} \Sigma^\ast)\poneevecs^\ast
\end{align*}
Hence, to understand the eigenstructure of the Schur complement, we simply need
to compare the all negative eigenvalues of $D_1^2f(x^\ast)$ in increasing order
with the most positive eigenvalues of $-D_2^2f(x^\ast)$ in decreasing order.
Indeed, 
 by assumption, $r\leq
n$ and $\kappa^2\lambda_i+\mu_i>0$ for each $i\in \{1, \ldots, r-p\}$. Thus, 
\[D_1^2f(x^\ast)-D_{21}f(x^\ast)^\trans
(D_2^2f(x^\ast))^{-1}D_{21}f(x^\ast)>0\]
since it is a symmetric matrix. Combining this with the fact that
$-D_2^2f(x^\ast)>0$, $x^\ast$ is a differential Stackelberg equilibrium. Hence,
by Proposition~\ref{prop:allstack} it is an attractor of $\dot{x}=-\omegas(x)$.
\end{proof}

\subsection{Dynamical Systems Theory Primer}
\begin{definition}
    Given $T>0$, $\delta>0$, if there exists an increasing sequence of times
    $t_j$ with $t_0=0$ and $t_{j+1}-t_j\geq T$ for each $j$ and solutions
    $\xi^j(t)$, $t\in[t_j,t_{j+1}]$ of $\dot{\xi}=F(\xi)$ with initialization
    $\xi(0)=\xi_0$ such that $\sup_{t\in [t_j,t_{j+1}]}\|\xi^j(t)-z(t)\|<\delta$
    for some bounded, measurable $z(\cdot)$, the we call $z$ a
    $(T,\delta)$--perturbation.
\end{definition}
\begin{lemma}[Hirsch Lemma]
    Given $\vep>0$, $T>0$, there exists $\bar{\delta}>0$ such that for all
    $\delta\in(0,\bar{\delta})$, every $(T,\delta)$--perturbation of
    $\dot{\xi}=F(\xi)$ converges to an $\vep$--neighborhood of the global
    attractor set for $\dot{\xi}=F(\xi)$.
    \label{lem:hirsch}
\end{lemma}

A key tool used in the finite-time two-timescale analysis is the nonlinear
variation of constants formula of Alekseev~\cite{alekseev:1961aa},
\cite{borkar:2018aa}. 
\begin{theorem}
    Consider a differential equation 
    \[\dot{u}(t)=f(t,u(t)), \ t\geq 0,\]
    and its perturbation
    \[\dot{p}(t)=f(t,p(t))+g(t,p(t)), \ t\geq 0\]
    where $f,g:\mb{R}\times \mb{R}^d\rar \mb{R}^d$, $f\in C^1$, and $g\in
    C$. Let $u(t,t_0,p_0)$ and $p(t,t_0,p_0)$ denote the solutions of the above
    nonlinear systems for $t\geq t_0$ satisfying
    $u(t_0,t_0,p_0)=p(t_0,t_0,p_0)=p_0$, respectively. Then,
    \begin{align*}
        p(t,t_0,p_0)&=u(t,t_0,p_0)+\int_{t_0}^t\Phi(t,s,p(s,t_0,p_0))g(s,p(s,t_0,p_0))\
        ds, \ t\geq t_0
    \end{align*}
    where $\Phi(t,s,u_0)$, for $u_0\in \mb{R}^d$, is the fundamental matrix of
    the linear system 
    \begin{equation}
        \dot{v}(t)=\frac{\partial f}{\partial u}(t,u(t,s,u_0))v(t), \ t\geq s
        \label{eq:vdot}
    \end{equation}
    with $\Phi(s,s,u_0)=I_d$, the $d$--dimensional identity matrix.
    \label{thm:alekseev}
\end{theorem}
Typical two-timescale analysis has historically leveraged the discrete
Bellman-Grownwall lemma~\cite[Chap.~6]{borkar:2008aa}. Recent application of
Alekseev's formula has lead to tighter bounds, and is thus becoming commonplace
in such analysis.

\section{Extended Analysis}
\label{app:proofs}
The results in Section~\ref{sec:finitetime} leverage classical results from
stochastic
approximation
\cite{bhatnagar:2013aa,borkar:2008aa,kushner2003stochastic,benaim:1999aa} including recent
advances in that same domain~\cite{borkar:2018aa, thoppe:2018aa}.
Here we provide more detail on the derivation of the bounds presented in
Section~\ref{sec:finitetime} in order to provide insight into what the constants
are in the concentration bounds in Theorems~\ref{thm:conjecturetrack} and
\ref{thm:lockin}. Moreover, the presentation here is somewhat distilled and the
aim is to help the reader through the analysis in~\citet{borkar:2018aa} and
\citet{thoppe:2018aa} as it pertains to the setting we consider. We refer the reader to each of these papers and
references therein for
even more detail. 

As in the main body of the paper, consider a locally asymptotically stable differential Stackelberg equilibrium
$x^\ast=(x_1^\ast,
\conj(x_1^\ast))\in X$ and let $B_{q_0}(x^\ast)$ be an ${q}_0>0$ radius ball around $x^\ast$ contained in the
region of attraction.
Stability implies that the Jacobian $\Js(x_1^\ast,
\conj(x_1^\ast))$ is positive definite and by the converse Lyapunov
theorem~\cite[Chap.~5]{sastry:1999aa} there
exists local Lyapunov functions for the dynamics $\dot{x}_1(t)=-\tau
Df_1(x_1(t),\conj(x_1(t)))$ and for the dynamics $\dot{x}_2(t)=-D_2f_2(x_1,
x_2(t))$, for each fixed $x_1$. In particular, there exists a local Lyapunov
function $V\in C^1(\mb{R}^{d_1})$ with $\lim_{\|x_1\|\uparrow
\infty}V(x_1)=\infty$, and $\la \nabla V(x_1), Df_1(x_1,\conj(x_1))\ra<0$ for $x_1\neq
x_1^\ast$. 

For $q>0$, let $V^q=\{x\in \text{dom}(V):\ V(x)\leq q\}$. Then, there is also 
$q>q_0>0$ and $\epsilon_0>0$ such that for $\epsilon<\epsilon_0$,
\[\{x_1\in \mb{R}^{d_1}|\ \|x_1-x_1^\ast\|\leq \epsilon\}\subseteq V^{q_0}\subset
\mc{N}_{\epsilon_0}(V^{q_0})\subseteq V^q\subset \text{dom}(V)\] where
\[\mc{N}_{\epsilon_0}(V^{q_0})=\{x\in \mb{R}^{d_1}|\ \exists x'\in V^{q_0}\
\text{s.t.} \|x'-x\|\leq \epsilon_0\}.\] An analogously defined $\tilde{V}$
exists for the dynamics $\dot{x}_2$ for each fixed $x_1$. 

For now, fix $n_0$ sufficiently large; we specify the values of $n_0$ for which
the theory holds before the statement of Theorem~\ref{thm:conjecturetrack}. Define the
event $\mc{E}_n=\{\bar{x}_2(t)\in V^q\ \forall t\in [\tilde{t}_{n_0},
\tilde{t}_n]\}$ where
\[\textstyle\bar{x}_2(t)=x_{2,k}+\frac{t-\tilde{t}_k}{\gamma_{2,k}}(x_{2,k+1}-x_{2,k})\] are
linear interpolates---i.e., \emph{asymptotic pseudo-trajectories}---defined for $t\in (\tilde{t}_k, \tilde{t}_{k+1})$ with
$\tilde{t}_{k+1}=\tilde{t}_k+\gamma_{2,k}$ and $\tilde{t}_0=0$. 

We can express the asymptotic pseudo-trajectories for any
$n\geq n_0$ as
\[\textstyle\bar{x}_2(\tilde{t}_{n+1})=\bar{x}_2(\tilde{t}_{n_0})-\sum_{k=n_0}^n\gamma_{2,k}(D_2f_2(x_{k})+w_{2,k+1}).\]
Note that 
\[\textstyle\sum_{k=n_0}^n\gamma_{2,k}D_2f_2(x_{k})=\sum_{k=n_0}^n\int_{\tilde{t}_k}^{\tilde{t}_{k+1}}D_2f_2(x_{1,k},\bar{x}_2(\tilde{t}_k))\
ds\]
and similarly for the $w_{2,k+1}$ term, due to the fact that
$\ti{t}_{k+1}-\ti{t}_k=\gamma_{2,k}$ by construction. 
Hence,
for $s\in [\ti{t}_k,\ti{t}_{k+1})$, the above can be rewritten as 
    \begin{align*}\bar{x}_2(t)=\textstyle\bar{x}_2(\ti{t}_{n_0})+\int_{\ti{t}_{n_0}}^t-D_2f_2(x_1(s),\bar{x}_2(s))+\zeta_{21}(s)+\zeta_{22}(s)\
    ds\end{align*}
where 
$\zeta_{21}(s)=-D_2f_2(x_1(\ti{t}_k),\bar{x}_2(\ti{t}_k))-D_2f_2(x_1(s),\bar{x}_2(s))$
and $\zeta_{22}(s)=-w_{2,k+1}$. In the main body of the paper
$\zeta_2(s)=\zeta_{21}(s)+\zeta_{22}(s)$.

Then, by the nonlinear variation of constants formula (Alekseev's formula), we
have
\[
  \textstyle  \bar{x}_2(t)=x_2(t)+\Phi_2(t,s,x_1(\tilde{t}_{n_0}),\bar{x}_2(\tilde{t}_{n_0}))(\bar{x}_2(\tilde{t}_{n_0})-x_2(\tilde{t}_{n_0}))+\int_{\tilde{t}_{n_0}}^t\Phi_2(t,s,x_1(s),\bar{x}_2(s))(\zeta_{21}(s)+\zeta_{22}(s))\ ds
\]
where $x_1(t)\equiv x_1$ is constant (since $\dot{x}_1=0$) and
$x_2(t)=\conj(x_1)$. Moreover, for $t\geq s$, $\Phi_2(\cdot)$ satisfies linear
 system 
 \begin{equation*}
     \dot{\Phi}_2(t,s,x_{0})=J_2(x_1(t),x_2(t))\Phi_2(t,s,x_{0}),
 \end{equation*}
 with initial data $\Phi_2(t,s,x_{0})=I$ and $x_0=(x_{1,0},x_{2,0})$ and where $J_2$ the Jacobian of
$-D_2f_2(x_1,\cdot)$.

Given that $x^\ast=(x_1^\ast, \conj(x_1^\ast))$ is a stable differential
Stackelberg equilibrium, $J_2(x^\ast)$ is positive definite. Hence,  as in
\cite[Lem.~5.3]{thoppe:2018aa}, we can find
$M$,
$\kappa_2>0$ such that for $t\geq s$, $x_{2,0}\in V^r$,
\[\|\Phi_2(t,s,x_{1,0},x_{2,0})\|\leq Me^{-\kappa_2(t-s)}.\] This result follows
from 
standard results on stability of linear systems (see, e.g., \citet[\S7.2,
 Thm.~33]{callier:1991aa})  along with 
 a bound on
 $\int_{s}^t\|D^2_2f_2(x_{1},x_{2}(\tau,s,\tilde{x}_0))-D_2^2f_2(x^\ast)\|d\tau$
 for $\tilde{x}_0\in V^q$ (see, e.g.,~\citet[Lem~5.2]{thoppe:2018aa}).

Analogously we can define linear interpolates or asymptotic pseudo-trajectories
for $x_{1,k}$. Indeed,
\[\textstyle\bar{x}_1(t)=x_{1,k}+\frac{t-\hat{t}_k}{\gamma_{1,k}}(x_{1,k+1}-x_{1,k})\]
are
the linear interpolated points between the samples $\{x_{1,k}\}$ where
$\hat{t}_{k+1}=\hat{t}_k+\gamma_{1,k}$, and $\hat{t}_0=0$. Then, as above,
Alekseev's formula can again be applied to get 
\begin{align*}
    \bar{x}_1
    &\textstyle(t)=x_1(t,\hat{t}_{n_0},y(\hat{t}_{n_0}))+\Phi_1(t,\hat{t}_{n_0},
    \bar{x}_1(\hat{t}_{n_0}))(\bar{x}_1(\hat{t}_{n_0})-x_1(\hat{t}_{n_0}))\\
    &\textstyle\qquad+\int_{\hat{t}_{n_0}}^t\Phi_1(t,s,\bar{x}_1(s))(\zeta_{11}(s)+\zeta_{12}(s)+\zeta_{13}(s))\
    ds
\end{align*}
where $x_1(t)\equiv x_1^\ast$ (again, since $\dot{x}_1=0$) and the following
hold:
\begin{align*}
    \zeta_{11}(s)&=Df_1(x_{1,k},r_2(x_{1,k}))-Df_1(\bar{x}_1(s),r_2(\bar{x}_1(s)))\\
    \zeta_{12}(s) &=Df_1(x_k)-Df_1(x_{1,k},\conj(x_{1,k}))\\
    \zeta_{13}(s) &=w_{1,k+1}
\end{align*}
Moreover, $\Phi_1$ is the solution to a linear system
with dynamics $J_1(x_1^\ast, \conj(x_1^\ast))$, the Jacobian of
$-Df_1(\cdot,\conj(\cdot))$, and with initial data $\Phi_1(s,s,x_{1,0})=I$. This
linear system, as above, has bound \[\|\Phi_1(t,s,x_{1,0})\|\leq
M_1e^{\kappa_1(t-1)}\] for some $M_1,\kappa_1>0$. 

Now, in addition to the linear iterpolates for $x_{1,k}$ and $x_{2,k}$, we
define an auxiliary sequence representing the leader's conjecture about the
follower with the goal of bounding the normed difference between follower's
response and this auxiliary sequence. Indeed, using a Taylor expansion of
the implicitly defined map $\conj$, we get
\begin{equation}
    z_{k+1}=z_k+D \conj(x_{1,k})(x_{1,k+1}-x_{1,k})+\delta_{k+1}
    \label{eq:taylor}
\end{equation}
where $\delta_{k+1}$ are the remainder terms which satisfy $\|\delta_{k+1}\|\leq
L_{r}\|x_{1,k+1}-x_{1,k}\|^2$ by assumption. Plugging in $x_{1,k+1}$, 
\begin{align*}
    z_{k+1}&=z_k+\gamma_{2,k}\big(-D_2f_2(x_{1,k},z_k)+\tau_kD\conj(x_{1,k})(w_{1,k+1}-Df_1(x_{1,k},x_{2,k})\big)+\gamma_{2,k}^{-1}\delta_{k+1}).
\end{align*}
The terms after $-D_2f_2$ are $o(1)$, and hence asymptotically negligible, so
that this $z$ sequence tracks dynamics as $x_{2,k}$. 
Using similar techniques as above, we can express linear interpolates of the
leader's belief regarding the follower's reaction as 
\begin{align*}
    \bar{z}(t)&=\textstyle\bar{z}(\ti{t}_{n_0})+\int_{\ti{t}_{n_0}}^t-D_2f_2(x_1(s),\bar{z}(s))+\sum_{j=1}^4\zeta_{3j}(s)\
    ds
\end{align*}
where the $\zeta_{3j}$'s are defined as follows:
\begin{align*}
    \zeta_{31}(s) & =
    -D_2f_2(x_1(\ti{t}_k),\bar{z}(\ti{t}_k))+D_2f_2(x_1(s),\bar{z}(s))\\
    \zeta_{32}(s)&=\tau_kD \conj(x_{1,k})w_{1,k+1}\\
    \zeta_{33}(s)&=-\tau_kDf_1(x_{1,k},x_{2,k})D \conj(x_{1,k})\\
    \zeta_{34}(s)&=\textstyle\frac{1}{\gamma_{2,k}}\delta_{k+1}
\end{align*}
with $\tau_k=\gamma_{1,k}/\gamma_{2,k}$.
Once again, Alekseev's formula can be applied where $x_2(t)=\conj(x_1)$ and
$\Phi_2$ is the same as in the application of Alekseev's to $x_{2,k}$. Indeed,
this gives us
\begin{align*}
    \bar{z}(\ti{t}_n)&=x_2(\ti{t}_n)+\Phi_2(\ti{t}_n,
    \ti{t}_{n_0},x_1(\ti{t}_{n_0}),
    \bar{z}(\ti{t}_{n_0}))(\bar{z}(\ti{t}_{n_0})-x_2(\ti{t}_{n_0}))\\
    &\quad\textstyle+\sum_{k=n_0}^{n-1}\int_{\ti{t}_k}^{\ti{t}_{k+1}}\Phi_2(\ti{t}_n,s,x_1(s),\bar{z}(s))(-D_2f_2(x_1(\ti{t}_k),\bar{z}(\ti{t}_k))+D_2f_2(x_1(s),\bar{z}(s)))\
    ds\tag{a}\\
    &\quad\textstyle+\sum_{k=n_0}^{n-1}\int_{\ti{t}_k}^{\ti{t}_{k+1}}\Phi_2(\ti{t}_n,s,x_1(s),\bar{z}(s))\tau_kD\conj(x_{1,k})w_{1,k+1}\
    ds\tag{b}\\
    &\quad\textstyle-\sum_{k=n_0}^{n-1}\int_{\ti{t}_k}^{\ti{t}_{k+1}}\Phi_2(\ti{t}_n,s,x_1(s),\bar{z}(s))\tau_kDf_1(x_{1,k},x_{2,k})D\conj(x_{1,k})\
    ds\tag{c}\\
    &\quad\textstyle
    +\sum_{k=n_0}^{n-1}\int_{\ti{t}_k}^{\ti{t}_{k+1}}\Phi_2(\ti{t}_n,s,x_1(s),\bar{z}(s))\frac{1}{\gamma_{2,k}}\delta_{k+1}\
    ds\tag{d}
\end{align*}
Applying the linear system stability results, we get that
\begin{equation}\|\Phi_2(\ti{t}_n,\ti{t}_{n_0},x_1(\ti{t}_{n_0}),
    \bar{z}(\ti{t}_{n_0}))(\bar{z}(\ti{t}_{n_0})-x_2(\ti{t}_{n_0}))\|\leq
e^{-\kappa_2(\ti{t}_n-\ti{t}_{n_0})}\|\bar{z}(\ti{t}_{n_0})-x_2(\ti{t}_{n_0})\|.
\label{eq:phibd}
\end{equation}
Each of the terms (a)--(d) can be bound as in Lemma III.1--5 in
\cite{borkar:2018aa}. The bounds are fairly straightforward using \eqref{eq:phibd}.

Now that we have each of these asymptotic pseudo-trajectories, we can show that with high
probability, $x_{2,k}$ and $z_k$ 
asymptotically contract to one another, leading to the conclusion that the
follower's dynamics track the leader's belief about the follower's reaction. 
Moreover, we can bound the difference between each  $x_{i,k}$, using
$\bar{x}_i(t_{i,k})=x_{i,k}$, and the continuous flow $x_i(t)$ on each interval
$[t_{i,k},t_{i,k+1})$ for each $i=1,2$ and where  $t_{1,k}=\hat{t}_k$ and
$t_{2,k}=\tilde{t}_k$.
These normed-difference bounds can then be leveraged to obtain concentration
bounds by taking a union bound across all continuous time intervals defined
after sufficiently large $n_0$ and conditioned on the events
 $\mc{E}_n=\{\bar{x}_2(t)\in V^q\ \forall t\in [\tilde{t}_{n_0},
\tilde{t}_n]\}$ and $\hat{\mc{E}}_n=\{\bar{x}_1(t)\in V^{q}\ \forall t\in [\hat{t}_{n_0},
\hat{t}_n]\}$.

Towards this end, define
$H_{n_0}=(\|\bar{x}_2(\tilde{t}_{n_0}-x_2(\tilde{t}_{n_0})\|+\|\bar{z}(\tilde{t}_{n_0})-x_2(\tilde{t}_{n_0})\|)$, 
\[\textstyle S_{1,n}=\sum_{k=n_0}^{n-1}\left(\int_{\hat{t}_{k}}^{\hat{t}_{k+1}}\Phi_1(\hat{t}_n,s,\bar{x}_1(\hat{t}_k))ds
\right) w_{1,k+1},\]
and
\[\textstyle S_{2,n}=\sum_{k=n_0}^{n-1}\left(\int_{\tilde{t}_k}^{\tilde{t}_{k+1}}\Phi_2(\tilde{t}_n,s,x_{1}(\tilde{t}_k),\bar{x}_2(\tilde{t}_k))
ds\right)w_{2,k+1}.\]

 Applying Lemma~5.8 \cite{thoppe:2018aa}, conditioned on $\mc{E}_n$, we get there
 exists some constant $K>0$ such that
 \begin{align*}
     \|\bar{x}_2(\ti{t}_n)-x_2(\ti{t}_n)\|&\leq \textstyle\|\Phi_2(\ti{t}_n, \ti{t}_{n_0},
     x_1,\bar{x}_2(\ti{t}_{n_0}))(\bar{x}_2(\ti{t}_{n_0})-x_2(\ti{t}_{n_0}))\|+K\Big(\|S_{2,n}\|\\
     &\textstyle\qquad+\sup_{n_0\leq
     k\leq n-1}\gamma_{2,k}+\sup_{n_0\leq
     k\leq n-1}\gamma_{2,k}\|w_{2,k+1}\|^2\Big)
 \end{align*}
Using the bound on the linear system $\Phi_2(\cdot)$, this exactly leads to the bound 
 \begin{align*}
     \|\bar{x}_2(\ti{t}_n)-x_2(\ti{t}_n)\|&\leq\textstyle
     K\Big(e^{-\kappa_2(\ti{t}_n-\ti{t}_{n_0})}\|\bar{x}_2(\ti{t}_{n_0})-x_2(\ti{t}_{n_0})\|\\
     &\textstyle\quad+\|S_{2,n}\|+\sup_{n_0\leq
     k\leq n-1}\gamma_{2,k}+\sup_{n_0\leq
     k\leq n-1}\gamma_{2,k}\|w_{2,k+1}\|^2\Big)
 \end{align*}
Thus, leveraging Lemma III.1--5 \cite{thoppe:2018aa}, we obtain the result of
Lemma~\ref{lem:defK} in the main body of the paper, and stated here for easy
access.
\begin{lemma}[Lemma~\ref{lem:defK} of main body]
    For any $n\geq n_0$, there exists $K>0$ such that conditioned on
    ${\mc{E}}_n$,
    \begin{align*}
        \|x_{2,n}-z_n\|\leq&\textstyle
       K\Big(\|S_{2,n}\|+e^{-\kappa_2(\tilde{t}_n-\tilde{t}_{n_0})}H_{n_0}+\sup_{n_0\leq k\leq n-1}\gamma_{2,k}+\sup_{n_0\leq k\leq
        n-1}\gamma_{2,k}\|w_{2,k+1}\|^2\notag\\
        &\textstyle+\sup_{n_0\leq k\leq n-1}\tau_k+\sup_{n_0\leq k\leq
        n-1}\tau_k\|w_{1,k+1}\|^2\Big).
    \end{align*}
    
        \label{lem:defKa}
\end{lemma}
Lastly, in a similar fashion we can obtain a bound for the leader's sample path
$x_{1,k}$.
\begin{lemma}[Lemma~\ref{lem:defbarK} of main body]
    For any $n\geq n_0$, there exists $\bar{K}>0$ such that   conditioned on
    $\tilde{\mc{E}}_n$,
    \begin{align*}
       \|\bar{x}_1(\hat{t}_n)-x_1(\hat{t}_n)\|\leq&\textstyle
        \bar{K}\Big(\|S_{1,n}\|+\sup_{n_0\leq k\leq n-1}\|S_{2,k}\|+\sup_{n_0\leq k\leq n-1}\gamma_{2,k}+\sup_{n_0\leq k\leq n-1}\tau_k
\notag\\
        &\textstyle +\sup_{n_0\leq k\leq
        n-1}\gamma_{2,k}\|w_{2,k+1}\|^2   +\sup_{n_0\leq k\leq
        n-1}\tau_k\|w_{1,k+1}\|^2\\
        &\quad \textstyle+e^{\kappa_1(\hat{t}_n-\hat{t}_{n_0})}\|\bar{x}_1(\hat{t}_{n_0})-x_1(\hat{t}_{n_0})\| +\sup_{n_0\leq
        k\leq n-1}\tau_kH_{n_0}\Big).
    \end{align*}
    \label{lem:defbarKa}
\end{lemma}

To obtain concentration bounds, the results are exactly as in Section
IV~\cite{borkar:2018aa} which follows the analysis in \cite{thoppe:2018aa}.
Fix $\vep\in[0,1)$ and let $N$ be such that $\gamma_{2,n}\leq
\vep/(8K)$, $\tau_n\leq \vep/(8K)$ for all $n\geq N$. Let $n_0\geq N$ and with
$K$ as in Lemma~\ref{lem:defK}, let $T$ be such that
$e^{-\kappa_2(\tilde{t}_n-\tilde{t}_{n_0})}H_{n_0}\leq \vep/(8K)$ for all $n\geq
n_0+T$.

Using Lemma~\ref{lem:defKa} and Lemma 3.1~\cite{thoppe:2018aa}, 
\begin{align*}
     \mathrm{P}(\|x_{2,n}-z_n\|\leq \vep,& \forall n\geq \bar{n}| x_{2,n_0},z_{n_0}\in
     B_{q_0})\notag \\
     &\geq\textstyle 1-\mathrm{P}(\bigcup_{n=n_0}^\infty\mc{A}_{1,n}\cup\bigcup_{n=n_0}^\infty\mc{A}_{2,n}\cup\bigcup_{n=n_0}^\infty\mc{A}_{3,n}|\ x_{2,n_0},z_{n_0}\in
     B_{q_0})
\end{align*}
where 
\begin{align*}
    \mc{A}_{1,n}= \textstyle\left\{\mc{E}_n,
    \|S_{2,n}\|>\frac{\vep}{8K}\right\},\ \ \mc{A}_{2,n}=\left\{\mc{E}_n,
\gamma_{2,k}\|w_{2,n+1}\|^2>\frac{\vep}{8K}\right\},
\end{align*}
and
\[\mc{A}_{3,n}=\textstyle\left\{\mc{E}_n,
\tau_n\|w_{1,n+1}\|^2>\frac{\vep}{8K}\right\}.\]
Taking a union bound
gives
\begin{align*}
    \mathrm{P}(\|x_{2,n}-z_n\|\leq \vep, \forall n\geq \bar{n}| x_{2,n_0},z_{n_0}\in
     B_{q_0}) \geq&\textstyle
     1-\sum_{n=n_0}^\infty\mathrm{P}(\mc{A}_{1,n}|\ x_{2,n_0},z_{n_0}\in
     B_{q_0})\\
     &\ \ \textstyle+\sum_{n=n_0}^\infty\mathrm{P}(\mc{A}_{2,n}|\ x_{2,n_0},z_{n_0}\in
     B_{q_0})\\
     &\ \ +\textstyle
    \sum_{n=n_0}^\infty\mathrm{P}(\mc{A}_{3,n})|\ x_{2,n_0},z_{n_0}\in
     B_{q_0}).
\end{align*}
Theorem 6.2~\cite{thoppe:2018aa}, gives bounds 
\begin{align}
 \textstyle   \sum_{n=n_0}^\infty\mathrm{P}(\mc{A}_{2,n}|\ x_{2,n_0},z_{n_0}\in
    B_{q_0})\leq K_1\sum_{n=n_0}^{\infty}\exp\left(
    -\frac{K^2\sqrt{\vep}}{\sqrt{\gamma_{2,k}}} \right),
    \label{eq:bd1}
\end{align}
\begin{align}
  \textstyle  \sum_{n=n_0}^\infty\mathrm{P}(\mc{A}_{3,n})|\ x_{2,n_0},z_{n_0}\in
     B_{q_0}) \leq K_1\sum_{n=n_0}^{\infty}\exp\left(
     -\frac{K^2\sqrt{\vep}}{\sqrt{\tau_k}} \right),
    \label{eq:bd2}
\end{align}
and, by Theorem 6.3~\cite{thoppe:2018aa}
\begin{align}
\textstyle    \sum_{n=n_0}^\infty\mathrm{P}(\mc{A}_{1,n}|\ x_{2,n_0},z_{n_0}\in
    B_{q_0})\leq K_2\sum_{n=n_0}^\infty \exp\left( -\frac{K_3\vep^2}{\beta_n} \right)
    \label{eq:bd3}
\end{align}
with 
\[\beta_n=\max_{n_0\leq k\leq
n-1}e^{-\kappa_2(\sum_{i=k+1}^{n-1}\gamma_{2,i})}\gamma_{2,k}\]
for some $K_1,K_2,K_3>0$.
This gives the result of Theorem~\ref{thm:conjecturetrack} in the main body with
$C_1=K_1$, $C_2=K^2$, $C_3=K_2$, $C_4=K_3$.
An
exactly analogous analysis holds for obtaining the concentration bound in
Theorem~\ref{thm:lockin}.

\section{Regularizing the Follower's Implicit Map}
\label{app:regfollower}
The derivative of the implicit function used in the leader's update requires the
follower's Hessian to be an isomorphism. In practice, this may not always be
true along the learning path. Consider the modified update
\begin{align*}
     x_{k+1,1}&=x_{k,1}-\gamma_1(
     D_1f_1(x_k)-D_{21}f_2(x_k)^{\top}(D_{2}^{2}f_2(x_k)+\eta I)^{-1}D_{2}f_1(x_k))\\
     x_{k+1,2}&=x_{k,2}-\gamma_2D_2f_2(x_k),
\end{align*}
in which we
regularize the inverse of $D_2^2f_2$ term. 
This update can be derived from the following perspective. 
Suppose player 1 views player 2 as optimizing a linearized version of its cost
with a regularization term which captures the leader's lack of confidence in the
local linearization holding globally: 
\[\arg\min_{y} 
 \ (y-x_{2,k})^\top D_2f_2(x_k)+\frac{\eta}{2}\|y-x_{2,k}\|^2.\]
The first-order optimality conditions for this problem are
\begin{align*}
   0&= D_2f_2(x_k)+(y-x_{k,2})^\top D_2^2f_2(x_k)+\eta(y-x_{k,2})\\
   &=D_2f_2(x_k)-\left(\eta I+D_2^2f_2(x_k)\right)x_{k,2}+(D_2^2f_2(x_k)+\eta I)y.
\end{align*}
Hence, if the leader views the follower as updating along the gradient direction
determined by these first order conditions, then the follower's response map is
given by
\[x_{k+1,2}=x_{k,2}-\left(D_2^2f_2(x_k)+\eta I\right)^{-1}D_2f_2(x_k).\]
Ignoring higher order terms in the derivative of the response map, the
approximate Stackelberg update is given by
\begin{align*}
     x_{k+1,1}&=x_{k,1}-\gamma_1( D_1f_1(x_k)-D_{21}f_2(x_k)^{\top}\left(D_2^2f_2(x_k)+\eta I\right)^{-1}D_{2}f_1(x_k))\\
     x_{k+1,2}&=x_{k,2}-\gamma_2D_2f_2(x_k).
\end{align*}

In our GAN experiments, we use the regularized update since it is quite common
for the discriminator's Hessian to be ill-conditioned if not degenerate. Similarly, the Schur complement we present the eigenvalues for in the experiments includes the regularized individual Hessian for the follower.

\begin{proposition}[Regularized Stackelberg: Sufficient Conditions]
 A point $x^\ast$ such that the first order conditions $D_1f_1(x)-D_{21}f_2(x)^{\top}(D_{2}^{2}f_2(x)+\eta
    I)^{-1}D_{2}f_1(x)=0$ and $D_2f_2(x)=0$ hold, and such that $D_1(D_1f_1(x)-D_{21}f_2(x)^{\top}(D_{2}^{2}f_2(x)+\eta
    I)^{-1}D_{2}f_1(x))>0$ and $D_2^2f_2(x)>0$ is a differential Stackelberg
    equilibrium with respect to the regularized dynamics.
\end{proposition}
\begin{proposition}[Regularized Stackelberg: Necessary Conditions]
A differential Stackelberg equilibrium $x^\ast$ of the regularized dynamics
satisfies   $D_1f_1(x)-D_{21}f_2(x)^{\top}(D_{2}^{2}f_2(x)+\eta
    I)^{-1}D_{2}f_1(x)=0$ and $D_2f_2(x)=0$ hold, and  $D_1(D_1f_1(x)-D_{21}f_2(x)^{\top}(D_{2}^{2}f_2(x)+\eta
    I)^{-1}D_{2}f_1(x))\geq0$ and $D_2^2f_2(x)\geq 0$.
\end{proposition}

This result can be seen by examining first and second order sufficient
conditions for the leader's optimization problem given the regularized
conjecture about the follower's update, i.e.
\[\arg\min_{x_1}\left\{f_1(x_1,x_2)|\ x_2\in \arg\min_y
f_2(x_1,y)+\frac{\eta}{2}\|y\|^2\right\},\]
and for the problem follower is actually solving with its update $\arg\min_{x_2}f_2(x_1,x_2)$.

\section{Experiment Details}
\label{app:experiment_details}
This section includes complete details on the training process and hyper-parameters selected in the mixture of Gaussian and MNIST experiments.
\subsection{Mixture of Gaussians}
The underlying data distribution for the diamond experiment consists of Gaussian distributions with means given by $\mu = [1.5 \sin(\omega), 1.5\cos(\omega)]$ for $\omega \in \{k\pi/2\}_{k=0}^{3}$ and each with covariance $\sigma^2 I$ where $\sigma^2=0.15$. Each sample of real data given to the discriminator is selected uniformly at random from the set of Gaussian distributions. The underlying data distribution for the circle experiment consists of Gaussian distributions with means given by $\mu = [\sin(\omega), \cos(\omega)]$ for $\omega \in \{k\pi/4\}_{k=0}^7$ and each with covariance $\sigma^2 I$ where $\sigma^2=0.3$. Each sample of real data given to the discriminator is selected uniformly at random from the set of Gaussian distributions. 

We train the generator using latent vectors $z \in \mathbb{R}^{16}$ sampled from a standard normal distribution in each training batch. The discriminator is trained using input vectors $x\in \mathbb{R}^{2}$ sampled from the underlying distribution in each training batch. The batch size for each player in the game is 256. The network for the generator contains two hidden layers, each of which contain $32$ neurons. The discriminator network consists of a single hidden layer with $32$ neurons and it has a sigmoid activation following the output layer. We let the activation function following the hidden layers be the Tanh function and the ReLU function in the diamond and circle experiments, respectively. The initial learning rates for each player and for each learning rule is $0.0001$ and $0.0004$ in the diamond and circle experiments, respectively. The objective for the game in the diamond experiment is the saturating GAN objective and in the circle experiment it is the non-saturating GAN objective. We update the parameters for each player and in each experiment using the ADAM optimizer with the default parameters of $\beta_1 = 0.9$, $\beta_2 = 0.999$, and $\epsilon = 10^{-8}$. The learning rate for each player is decayed exponentially such that $\gamma_{i, k} = \gamma_i  \nu_i^k$. We let $\nu_1 =\nu_2 = 1-10^{-7}$ for simultaneous gradient descent and $\nu_1 = 1-10^{-5}$ and $\nu_1 = 1-10^{-7}$ for the Stackelberg update. We regularize the implicit map of the follower as detailed in Appendix~\ref{app:regfollower} using the parameter $\eta= 1$.

\subsection{MNIST}
To underlying data distribution for the MNIST experiments consists of digits 0 and 1 from the MNIST training dataset or each digit from the MNIST training dataset. We scale each image to the range $[-1, 1]$. Each sample of real data given to the discriminator is selected sequentially from a shuffled version of the dataset. The batch size for each player is 256. We train the generator using latent $z \in \mathbb{R}^{100}$ sampled from a standard normal distribution in each training batch. The discriminator is trained using input vectorized images $x\in \mathbb{R}^{28 \times 28}$ sampled from the underlying distribution in each training batch. We use the DCGAN architecture~\cite{radford2015unsupervised} for our generator and discriminator. Since DCGAN was built for $64\times 64$ images, we adapt it to handle $28\times 28$ images in the final layer. 
 We follow the parameter choices from the DCGAN paper~\cite{radford2015unsupervised}. This means we initialize the weights using a zero-centered centered Normal distribution with standard deviation $0.02$, optimize using ADAM with parameters $\beta_1 = 0.5$, $\beta_2=0.999$, and $\epsilon=10^{-8}$, and set the initial learning rates to be $0.0002$. The learning rate for each player is decayed exponentially such that $\gamma_{i, k} = \gamma_i  \nu_i^k$ and $\nu_1 = 1-10^{-5}$ and $\nu_1 = 1-10^{-7}$. We regularize the implicit map of the follower as detailed in Appendix~\ref{app:regfollower} using the parameter $\eta= 5000$. If we view the regularization as a linear function of the number of parameters in the discriminator, then this selection of regularization is nearly equal to that from the mixture of Gaussian experiments.

\section{Computing the Stackelberg Update and Schur Complement }
\label{appendix_secs:leaderupdate}

The learning rule for the leader involves computing an inverse-Hessian-vector product for the $D_{2}^2f_2(x)$ inverse term and Jacobian-vector product for the $D_{12}f_2(x)$ term. These operations can be done efficiently in Python by utilizing Jacobian-vector products in auto-differentiation libraries combined with the {\tt sparse.LinearOperator} class in {\tt scipy}. These objects can also be used to compute their eigenvalues, inverses, or the Schur complement of the game dynamics using  the {\tt scipy.sparse.linalg} package. We found that the conjugate gradient method {\tt cg} can compute the regularized inverse-Hessian-vector products for the leader update accurately with 5 iterations and a warm start. 

The operators required for the leader update can be obtained by the following. Consider the Jacobian of the simultaneous gradient descent learning dynamics $\dot x = -\omega(x)$ at a critical point for the general sum game $(f_1, f_2)$: 
\[ 
  J(x) = \bmat{D_{1}^2f_1(x) & D_{12}f_1(x)  \\ 
  D_{21}f_2(x) & D_{2}^2f_2(x) }.
\]
Its block components consist of four operators $D_{ij}f_i(x) : X_j \to X_i,\ i,j \in \{1,2\}$ that can be computed using forward-mode or reverse-mode Jacobian-vector products. Instantiating these operators as a linear operator in {\tt scipy} allows us to compute the eigenvalues of the two player's individual Hessians. Properties such as the real eigenvalues of a Hermitian matrix or complex eigenvalues of a square matrix can be computed using {\tt eigsh} or {\tt eigs} respectively. 
Selecting to compute the smallest or largest $k$ eigenvalues---sorted by either magnitude, real or imaginary values---allows one to examine the positive-definiteness of the operators. 

Operators can be combined to compute other operators relatively efficiently for large scale problems without requiring to compute their full matrix representation. For an example, take the Schur complement of the Jacobian above at fixed network parameters $x\in X_1 \times X_2 $, 
$ D_1^2(x) - D_{12}{f_1}(x) (D_2^2f_2)^{-1}(x) D_{21}f_2(x).$
We create an operator $S_1(x): X_1 \to X_1$ that maps a vector $v$ to $p-q$ by performing the following four operations: 
  $u = D_{21}f_2(x)v$,
  $w = (D_{2}^2f_2)^{-1}(x)u$,
  $q = D_{12}f_1(x)w$, and
  $p = D_1^2(x)v$.
Each of the operations can be computed using a single backward pass through the network except for computing $w$, since the inverse-Hessian
requires an iterative method which can be computationally expensive. It solves the linear equation $D_2^2f_2(x) w = u$ and there are various available methods: we tested (bi)conjugate gradient methods, residual-based methods, or least-squares methods, and each of them 
provide varying amounts of error when compared with the exact solution. 
Particularly, when the Hessian is poorly conditioned, some methods may fail to converge.
More investigation is required to determine which method is best suited for specific uses. For example, a fixed iteration method with warm start might be appropriate for computing the leader update online, while a residual-based method might be better for computing the the eigenvalues of the Schur complement. 
Specifically, for our mixture of gaussians and MNIST GANs, we found that computing the leader update using the conjugate gradient method with maximum of 5 iterations and warm-start works well. We compared using the true Hessian for smaller scale problems and found the estimate to be within numerical precision.

\section{$N$--Follower Setting}
\label{app:extensions}
In this section, we show that the results extend to the setting where there is
a single leader, but $N$ non-cooperative followers.

\subsection{$N+1$ Staggered Learners, All with Non-Uniform Learning Rates}
Note that if there is a layered hierarchy in which each, for example, the first
follower is a leader for the second follower,  the second follower a leader for
the third follower and so on, then the results in Section~3 apply under
additional assumptions on the learning rates. 

For instance, consider a three
player setting where $\gamma_{1,k}=o(\gamma_{2,k})$ and
$\gamma_{2,k}=o(\gamma_{3,k})$ so that player 1 is the slowest player (hence,
the `leader'), player 2
the second slowest, and player 3 the fastest, the `leader'. Then
similar asymptotic analysis can be applied with the following assumptions.
Consider 
\begin{equation}
    \left.\begin{array}{ll} \dot{x}_i & =0, \ i<3\\
        \dot{x}_3 &=F^{3}(x)\end{array}\right\}
    \label{eq:d1}
\end{equation}
where we will explicitly define $F^3$ shortly. Let $x^{<j}=(x_1,\ldots,
x_{j-1})$ and $x^{\geq j}=(x_j, \ldots, x_{N+1})$.
\begin{assumption}
    There exists a Lipschitz continuous function $r_3(x^{<3})$ such that for any
    $x$, solutions of \eqref{eq:d1} asymptotically converge to $(x^{<3},
    r_3(x^{<3}))$ given initial data $x$.
    \label{ass:two}
\end{assumption}
Consider 
\begin{equation}
    \left.\begin{array}{ll} \dot{x}_i & =0, \ i<2\\
        \dot{x}_2 &=F^{2}(x^{<3}, r_3(x^{<3}))\end{array}\right\}
    \label{eq:d2}
\end{equation}
\begin{assumption}
    There exists a Lipschitz continuous function $r_2(x^{<2})$ such that for any
    $x_3$, solutions of \eqref{eq:d2} asymptotically converge to $(x^{<2},
    r_3(x^{<3}))$ given initial data $(x^{<2},x^{\geq 2})$.
    \label{ass:three}
\end{assumption}
Now, define $\xi^{\geq 2}(x^{<2})=(r_2(x^{<2}),r_3(x^{<2},r_2(x^{<2})))$ for
notation simplicity. Let $F^3\equiv -D_3f_3$ and $F^2\equiv -D_{1\to 2}f_2$ where the notation $D_{j\to i}$ indicates the total derivative with respect to
arguments $j$ up to $i$.
\begin{proposition}
    Under Assumptions~\ref{ass:two} and \ref{ass:three} and
    Assumption~\ref{ass:all} from the main paper, 
\[\lim_{k\rar \infty}\|(x_{2,k},x_{3,k})-\xi^{\geq 2}(x_{k,1})\|\rar 0\ \
\text{a.s.}\]
\end{proposition}
Of course the framework naturally extends to $N$-followers; a similar framework
can be found for reinforcement learning algorithms in normal form
games~\cite{collins:2003aa}.
\subsection{$N$ Simultaneously Play Followers}
On the other hand, consider a setting in which 
the followers play a Nash
equilibrium in a simultaneous play game and are assumed to have the same
learning rate. That is, $\gamma_{1,k}=o(\gamma_{2,k})$ where all $N$ followers
use the learning rate $\gamma_{2,k}$ and the leader uses the learning rate
$\gamma_{1,k}$. The results for this section assume that the follower game has a
unique differential Nash equilibrium uniformly in $x_1$. 
\begin{assumption}
    For every $x_1$, 
    \begin{align*}
        \bmat{\dot{x}_2\\ \vdots \\\dot{x}_N} =\bmat{-D_2f_2(x_1,x^{\geq 1}(t))\\
        \vdots \\-D_Nf_N(x_1, x^{\geq 1}(t))}
    \end{align*} has a globally asymptotically stable differential Nash
    equilibrium $\conj(x_1)$ uniformly in $x_1$ with $\conj$ a $L_r$--Lipschitz
    function.
\end{assumption}
All the results in Section~\ref{sec:results} of the main body hold replacing
Assumption~\ref{ass:gas} with the above assumption. This is a somewhat strong
assumption, however, $N$-player convex games that are diagonally strictly
convex admit unique Nash equilibria which are attracting~\cite{rosen:1965aa}.

\end{document}